\documentclass[journal,table]{IEEEtran}

\usepackage{amsmath, amsthm, amssymb, bm}
\usepackage{epsfig}
\usepackage{amsfonts}
\usepackage{enumerate}

\usepackage{amsmath}
\usepackage{color}

\usepackage{bbm}
\usepackage{cite}
\usepackage{epstopdf}

\newtheorem{theorem}{Theorem}

\newtheorem{proposition}{Proposition}

\newtheorem{lemma}{Lemma}

\newtheorem{assumption}{Assumption}

\usepackage{subfigure}

\usepackage[autostyle]{csquotes}

\usepackage{hyperref}

\makeatletter
\newcommand*\rel@kern[1]{\kern#1\dimexpr\macc@kerna}
\newcommand*\widebar[1]{%
  \begingroup
  \def\mathaccent##1##2{%
    \rel@kern{0.8}%
    \overline{\rel@kern{-0.8}\macc@nucleus\rel@kern{0.2}}%
    \rel@kern{-0.2}%
  }%
  \macc@depth\@ne
  \let\math@bgroup\@empty \let\math@egroup\macc@set@skewchar
  \mathsurround\z@ \frozen@everymath{\mathgroup\macc@group\relax}%
  \macc@set@skewchar\relax
  \let\mathaccentV\macc@nested@a
  \macc@nested@a\relax111{#1}%
  \endgroup
}
\makeatother

\newlength\figureheight
\newlength\figurewidth

\newfont{\bbb}{msbm10 scaled 700}

\newfont{\bb}{msbm10 scaled 1100}

\newcommand{\av}{{\bf a}}

\newcommand{\dv}{{\bf d}}
\newcommand{\ev}{{\bf e}}
\newcommand{\fv}{{\bf f}}
\newcommand{\gv}{{\bf g}}
\newcommand{\hv}{{\bf h}}

\newcommand{\kv}{{\bf k}}

\newcommand{\mv}{{\bf m}}
\newcommand{\nv}{{\bf n}}

\newcommand{\pv}{{\bf p}}
\newcommand{\qv}{{\bf q}}
\newcommand{\rv}{{\bf r}}
\newcommand{\sv}{{\bf s}}

\newcommand{\uv}{{\bf u}}
\newcommand{\wv}{{\bf w}}
\newcommand{\vv}{{\bf v}}
\newcommand{\xv}{{\bf x}}
\newcommand{\yv}{{\bf y}}
\newcommand{\zv}{{\bf z}}

\newcommand{\Am}{{\bf A}}

\newcommand{\Cm}{{\bf C}}
\newcommand{\Dm}{{\bf D}}
\newcommand{\Em}{{\bf E}}

\newcommand{\Gm}{{\bf G}}
\newcommand{\Hm}{{\bf H}}
\newcommand{\Id}{{\bf I}}

\newcommand{\Lm}{{\bf L}}
\newcommand{\Mm}{{\bf M}}

\newcommand{\Qm}{{\bf Q}}

\newcommand{\Sm}{{\bf S}}
\newcommand{\Tm}{{\bf T}}
\newcommand{\Um}{{\bf U}}
\newcommand{\Wm}{{\bf W}}

\newcommand{\Xm}{{\bf X}}

\newcommand{\Zm}{{\bf Z}}

\newcommand{\alphav}{\hbox{\boldmath$\alpha$}}
\newcommand{\betav}{\hbox{\boldmath$\beta$}}
\newcommand{\gammav}{\hbox{\boldmath$\gamma$}}

\newcommand{\etav}{\hbox{\boldmath$\eta$}}

\newcommand{\nuv}{\hbox{\boldmath$\nu$}}

\newcommand{\phiv}{\hbox{\boldmath$\phi$}}

\newcommand{\thetav}{\hbox{\boldmath$\theta$}}

\newcommand{\omegav}{\hbox{\boldmath$\omega$}}

\newcommand{\Thetam}{\hbox{\boldmath$\Theta$}}

\newcommand{\Xim}{\hbox{\boldmath$\Xi$}}

\begin{document}

\title{Mathematical Theory of Atomic Norm Denoising In Blind Two-Dimensional Super-Resolution\\ (Extended Version)}
%
%
%

\author{Mohamed~A.~Suliman and Wei Dai

\thanks{M. A. Suliman and W. Dai are with the Department of Electrical and Electronic Engineering, Imperial College London, London, SW7 2AZ, United Kingdom. Emails:$\{$m.suliman17, wei,dai1$\}$@imperial.ac.uk. Part of this work is presented at the IEEE International Conference on Acoustics, Speech, and Signal Processing (ICASSP), Barcelona, Spain, May 2020 \cite{suliman2020atomic}.}}

\maketitle

\begin{abstract}
This paper develops a new mathematical framework for denoising in blind two-dimensional (2D) super-resolution upon using the atomic norm. The framework denoises a signal that consists of a weighted sum of an unknown number of time-delayed and frequency-shifted unknown waveforms from its noisy measurements. Moreover, the framework also provides an approach for estimating the unknown parameters in the signal. We prove that when the number of the observed samples satisfies certain lower bound that is a function of the system parameters, we can estimate the noise-free signal, with very high accuracy, upon solving a regularized least-squares atomic norm minimization problem. We derive the theoretical mean-squared error of the estimator, and we show that it depends on the noise variance, the number of unknown waveforms, the number of samples, and the dimension of the low-dimensional space where the unknown waveforms lie. Finally, we verify the theoretical findings of the paper by using extensive simulation experiments.

\end{abstract}

\section{Introduction}
\label{sec:intro}
Super-resolution refers to the inverse problem of recovering fine-scale information from low-resolution measurements. Such a problem arises in many real-world applications such as astronomy \cite{puschmann2005super}, radar imaging \cite{mi2008radar}, microscopy \cite{mccutchen1967superresolution}, and medical imaging \cite{kennedy2007improved}. Among the different fields of super-resolution, the super-resolution source localization, whose goal is to identify the locations of point sources from its convolution with low-pass point spread function, has gained considerable interest.


In recent years, many convex based methods have been widely used to super-resolve a set of unknown \emph{continuous} parameters. In the noise-free case, the work in \cite{candes2014towards} shows that we can recover a set of point sources at unknown locations in $[0,1]$, with infinite precision, upon solving a convex total-variation norm optimization problem. The problem is reformulated as semidefinite programming (SDP), and exact recovery is shown to hold as long as the distance between the points satisfies a minimum separation. The authors in \cite{tang2013compressed} apply the atomic norm framework to recover precisely a set of continuous frequencies from a random set of samples, whereas \cite{yang2016exact} extends the framework in \cite{tang2013compressed} to the case of multiple measurement vectors. The work in \cite{chi2014compressive} provides a two-dimensional (2D) atomic norm super-resolution framework to estimate a set of 2D frequencies from a random subset of samples gathered from a mixture of 2D sinusoids. The work shows that all unknown frequencies can be recovered under a minimum separation condition upon solving an atomic norm minimization problem. The recovery problem is simplified by assuming that the observed data can be represented using two 2D square matrices. Moreover, the work in \cite{heckel2016super} shows that we can recover precisely a set of 2D continuous frequencies in a linear time-varying system from a single input-output measurement upon using the atomic norm, while \cite{tian2017low} presents a 2D atomic norm super-resolution technique along with an efficient optimization technique for the direction of the arrival estimation problem. Finally, \cite{yang2016vandermonde} provides a super-resolution approach to recover multidimensional (MD) frequencies.


Blind super-resolution, or what is being referred to in the literature as blind deconvolution, is another direction of super-resolution in which the unknown sparse signal is convolved with an unknown point spread function (PSF). An example of that is when the PSF is imperfectly known or a time-varying function that drifts within the period of an experiment in a measurement device. The hurdle within blind super-resolution is that it is a severely ill-posed problem; hence, solving it typically involves imposing extra assumptions. In the 1D space, the work in \cite{chi2016guaranteed} develops an atomic norm framework to estimate unknown frequencies as well as a single point spread function and a spike signal, whereas \cite{yang2016super} generalizes \cite{chi2016guaranteed} to multiple unknown waveforms. Both works show that an exact recovery of the unknowns is guaranteed with high probability, assuming that the number of measurements is proportional to the degrees of freedom in the problem and that the unknown waveforms lie in a known low-dimensional subspace. 

On the other hand, the work in \cite{suliman2018blind} develops a blind 2D super-resolution framework to super-resolve a set of 2D continuous frequencies and unknown waveforms using the atomic norm. The work shows that when the unknown waveforms lie in a known low-dimensional subspace, an exact recovery for the frequencies and the waveforms holds given that the frequencies are well-separated and that the number of observed samples satisfy specific bound that is a function of the system parameters. This work is extended in \cite{suliman2019exact} for the 3D space.




In practical applications with real data, it is substantial to account for model imperfections and perturbation that result from noise. Thus, super-resolution from noisy measurements has been excessively studied. The authors in \cite{candes2013super} address the problem of super-resolving point sources from noisy data with high precision. The work shows that by solving a convex total-variation norm minimization problem, a stable estimate for the unknowns is guaranteed with an error that is proportional to the noise level and the super-resolution factor, provided that the points are well-separated. The work in \cite{bhaskar2013atomic} provides non-asymptotic guarantees on the mean-squared error (MSE) achieved by atomic norm denoising applied to line spectral estimation problem, whereas in \cite{li2015off}, the authors provide an atomic norm denoising approach for multiple signals and a recovery algorithm for their associated frequencies from noisy observations via atomic norm minimization. The denoising error rate in \cite{li2015off} is shown to be directly proportional to the noise level and inversely proportional to the total number of the observed samples. The work in \cite{zhang2017atomic} applies the framework provided in \cite{li2015off} for channel estimation and faulty antenna detection in the context of massive MIMO. 

An atomic norm denoising algorithm for a mixture of sinusoidal from noisy samples is proposed in \cite{tang2014near} in which the denoising error is derived and shown to be directly proportional to the noise level and inversely proportional to the number of samples. Furthermore, the authors in \cite{li2019atomic} study the problem of denoising a sum of complex exponentials with unknown waveform modulations using an atomic norm regularized least-squares problem. The theoretical MSE of the estimator is derived and shown to be proportional to the noise variance and the actual signal parameters. Finally, a comprehensive overview of atomic norm in super-resolution, along with some related applications, is provided in \cite{chi2020harnessing}.

\subsection{Contributions and Related Work}

In this paper, we provide an abstract theory for atomic norm denoising in blind 2D super-resolution that is easily extendible to any higher dimensions. We show that the atomic norm framework can be applied to denoise and then super-resolve a set of unknown continuous parameters in the 2D space. We formulate the denoising problem as a regularized least-square atomic norm minimization problem, and we show that it can be solved via an SDP. Moreover, we derive the theoretical bound on the MSE of the proposed denoising algorithm, and we show that it is a function of the noise level, the number of the observed samples, the number of the unknown parameters, and the subspace dimension where the unknown waveforms are assumed to lie. In particular, we prove that the MSE scales linearly with the noise variance and the square root of the number of the unknown waveforms and is inversely proportional to the number of the observed samples.

The significance of this work is that, to the best of our knowledge, it is the first work to address atomic norm denoising beyond the 1D blind super-resolution scheme. The work in this paper is closely related to that in \cite{suliman2018blind, suliman2019blind} and \cite{li2019atomic}. The focus of \cite{suliman2018blind, suliman2019blind} is on blind 2D super-resolution in the noise-free case, whereas in this paper, we consider the case in which the signal is contaminated by noise. Hence, we start from a different problem formulation, and we devise a distinct optimization recovery problem. Most importantly, while an exact recovery for the unknowns exists in \cite{suliman2018blind, suliman2019blind}, such exactness does not exist here due to the noise. Thus, an entirely different goal based on assessing the framework robustness to noise and the existence of a reasonable estimate for the unknowns is being addressed in this work. 

On the other hand, the work in \cite{li2019atomic} addresses the problem of blind super-resolution in 1D space. Extending the 1D framework to higher dimensions is non-trivial and comes with significant mathematical differences due to multiple reasons. First, the existence of the solution for the problem in \cite{li2019atomic} is shown by using a 1D polynomial that consists of multiple shifted versions of a \emph{single kernel}. Such formulation fails in our problem as our 2D trigonometric vector polynomial has to satisfy different constraints. Thus, we propose using \emph{multiple different kernels}, as we will show in Section~\ref{sec: proof on theorem}. Second, the MSE in \cite{li2019atomic} does not depend explicitly on the dimensionality of the problem and is entirely dependent on the formulation. Third, while \cite{li2019atomic} is based on using a 1D atomic norm, which is shown to be reformulated as an SDP, such reformulation does not exist for MD atomic norm and an alternative approach that is based on relaxing the problem is applied. Finally, our proof techniques allow us to impose less restricted assumptions on the subspace where the unknown waveforms lie than what in \cite{li2019atomic}, as we will discuss in Section~\ref{sec: main results theorem}.

It remains to point out that our proposed framework is easily extensible to any higher M dimensions upon following the same proof techniques provided in this paper. That is, we first obtain an MD trigonometric vector polynomial using our derivations; then, we can apply this polynomial and other proof methodologies provided in this paper to derive the atomic norm denoising framework. For example, the interested reader is referred to \cite{suliman2019exact} on how to obtain a 3D trigonometric vector polynomial using our proof techniques.

\subsection{Paper Organization}
In Section~\ref{sec: model}, we introduce our system model and derive our super-resolution denoising problem based on atomic norm optimization. In Section~\ref{sec: main results theorem}, we provide our main result, which characterizes the theoretical MSE performance of the problem, and we highlight its main assumptions. Section~\ref{sec: optimal con and problem solution} discusses the optimality conditions of the denoising problem introduced in Section~\ref{sec: main results theorem}, and shows how to super-resolve the unknown continuous parameters upon obtaining SDP relaxation for the recovery problem. In Section~\ref{sec: results}, we validate the theoretical findings in the paper using simulations, whereas Section~\ref{sec: proof on theorem} provides detailed proof of the theorem in Section~\ref{sec: main results theorem}. Conclusions and future work directions are drawn in Section~\ref{sec:conclusion}.

\subsection{Notations}
We use boldface lower-case symbols for column vectors (i.e., $\xv$) and upper-case for matrices (i.e., $\Xm$). $\left(\cdot\right)^{T}$, $\left(\cdot\right)^{H}$, and $\text{Tr}\left(\cdot\right)$ denote the transpose, the Hermitian, and the trace, respectively. $[\xv]_{i}$ denotes the $i$-th element of $\xv$ while $[\Xm]_{(i,j)}$ indicates the element in the $(i,j)$ entry of $\Xm$. The $M \times M$ identity matrix is denoted by $\Id_{\text{M}}$ while $\Xm \succeq \bm{0}$ indicates that $\Xm$ is a positive semidefinite matrix. When we use a two-dimensional index for vectors or matrices such as $[\xv]_{((k,l),1)},\  k,l=-N,\dots, N$, we mean that $\xv = [x_{(-N,-N)}, x_{(-N,-N+1)}, \dots , x_{(-N,N)}, \dots\dots, x_{(N,N)}]^{T}$. $\Xm_{(i,1\to K)}$ refers to the elements on the $i$-th row of $\Xm$ spanning from the first column to the $K$-th column. Moreover, we refer to the Kronecker product by $\otimes$. $||\cdot||_{2}$ designates the spectral norm for matrices and the Euclidean norm for vectors while $||\cdot||_{F}$ refers to the Frobenius norm. The infinity norm is denoted by $||\cdot||_{\infty}$. diag $\left(\xv\right)$ represents a diagonal matrix whose diagonal entries are the elements of $\xv$. Furthermore, $\left\langle\cdot,\cdot\right\rangle$ stands for the inner product whilst $\langle\cdot, \cdot\rangle_{\mathbb{R}}$ denotes the real inner product. The probability of an event is indicated by $\text{Pr}[\cdot]$, while the expectation operator is denoted by $\mathbb{E} [\cdot]$. For a given set $\mathcal{X}$, the notation $|\mathcal{X}|$ indicates the cardinality of the set, i.e., the number of the elements. Finally, $C, C_{1}, C^{*}, C^{*}_{1}, \hat{C}, \bar{C}, \dots$ are used to denote numerical constants that can take any real value.


\section{The Super-Resolution Problem}
\label{sec: model}
In this work, we model an observed signal $y^{*}\left(t\right)$ as a sum of $R$ different weighted versions of time-delayed and frequency-shifted \emph{unknown waveforms} $s_{j}\left(t\right)$, i.e., 
\begin{equation}
\label{eq: model}
y^{*} \left(t\right)= \sum_{j=1}^{R} c_{j} s_{j}\left(t-\tilde{\tau}_{j}\right) e^{i2\pi \tilde{f}_{j}t},
\end{equation}
where $\{c_{j}\}_{j=1}^{R} \in \mathbb{C}$ and the pairs $\{(\tilde{\tau}_{j}, \tilde{f}_{j})\}_{j=1}^{R}$ represent the unknown \emph{continuous} time-frequency shifts. Finally, note that both $R$ and $\{c_{j}\}_{j=1}^{R}$ are unknown. Such formulation arises in many signal processing applications such as military radar application, passive indoor source localization, and image restoration in astronomy (see \cite{suliman2018blind}). In this paper, we assume that the unknown waveforms $s_{j}\left(t\right)$ are bandlimited periodic signals with a bandwidth of $W$ and a period of length $T$ and that $y^{*}\left(t\right)$ is recorded over an interval of length $T$. Therefore, $(\tilde{\tau}_{j}, \tilde{f}_{j}) \in \left( \left[-T/2, T/2\right], \left[-W/2, W/2\right] \right)$. The discrete version of (\ref{eq: model}) can be obtained by sampling $y^{*}\left(t\right)$ in $\left[-T/2, T/2\right]$ at a rate of $1/W$ samples-per-second to accumulate a total of $L:=WT= 2N+1$ samples. Then, we can apply the DFT and the IDFT to the result equation and manipulate to obtain
\begin{align}
\label{eq: model sampled}
&y^{*}\left(p\right):= y^{*}\left(p/W\right)=  \nonumber\\
&\frac{1}{L}\sum_{j=1}^{R} c_{j} \left( \sum_{k=-N}^{N} \hspace{-1pt}\left[\hspace{-1pt}\left(\sum_{l=-N}^{N} s_{j}\left(l\right) e^{-i2\pi\frac{k l}{L}}\right)\hspace{-1pt} e^{-i2\pi  k \tau_{j}}\right] e^{i2\pi \frac{k p}{L}} \right)\times\nonumber\\
& e^{i2\pi p f_{j} },  \ \hspace{30pt} p=-N,\dots, N,
\end{align}
where ${\tau_{j}} := \frac{\tilde{\tau}_{j}}{T}$ and ${f_{j}} := \frac{\tilde{f}_{j}}{W}$. Based on these definitions, we can conclude that $\left(\tau_{j},f_{j}\right) \in [-1/2,1/2]^{2}$. From the periodicity property, we can assume without loss of generality that $\left(\tau_{j},f_{j}\right) \in [0,1]^{2}$.

In practical scenarios, the samples $y^{*} \left(p\right)$ are contaminated by noise; therefore, we can write
\begin{equation}
\label{eq: model noise}
y\left(p\right)= y^{*} \left(p\right)+ \omega\left(p\right),\ \  p=-N,\dots, N,
\end{equation}
where $\{\omega\left(p\right)\}_{p=-N}^{N}$ are the noise samples. In this paper, we assume that these samples are independent and identically distributed (i.i.d.) complex Gaussian samples of zero mean and variance $\sigma_{\omegav}^{2}$, i.e., $\{\omega\left(p\right)\}_{p=-N}^{N}\stackrel{i.i.d.}{\sim} \mathcal{CN}\left(0,\sigma_{\omegav}^{2}\right).$ 

To recover $R, c_{j}, \left(\tau_{j},f_{j}\right)$, and $s_{j}\left(l\right)$ from (\ref{eq: model noise}), it is clear that the order of the unknowns is $\mathcal{O}\left(LR\right)$, which is much greater than the number of samples $L$. Hence, the recovery problem is hopelessly ill-posed and can be filled-in randomly to obtain estimates that match the data. To mitigate this, we define $\sv_{j} \in \mathbb{C}^{L \times 1}$ such that $\sv_{j}:= [s_{j}\left(-N\right), \dots, s_{j}\left(N\right)]^{T}$, and we solve the problem under the assumption that $\{\sv_{j}\}_{j=1}^{R}$ lie in a known low-dimensional subspace spanned by the columns of a known matrix $\Dm\in \mathbb{C}^{L\times K}$ with $K \leq L$. Such assumption is applied in many existing works, e.g., \cite{yang2016super, chi2016guaranteed, ahmed2014blind, suliman2018blind}.

Based on the above discussion, we can write $\sv_{j}=\Dm\hv_{j}$ and $s_{j}\left(l\right)=\dv_{l}^{H}\hv_{j}$, where $\{\hv_{j}\}_{j=1}^{R}$ are \emph{unknown} orientation vectors while $\dv_{l} \in \mathbb{C}^{K \times 1}$ is the $l-$th column of $\Dm^{H}$, i.e, $\Dm:= \left[\dv_{-N}, \dots, \dv_{N}\right]^{H}$. Without loss of generality, we assume that $||\hv_{j}||_{2}=1$ for all $j$. Hence, recovering $\sv_{j}$ is equivalent to estimating $\hv_{j}$, and the order of the unknowns in the problem reduces to $\mathcal{O}\left(RK\right)$, which can be less than $L$ if $R, K \ll L$.

The random subspace assumption is applied in many existing works in the literature, e.g., \cite{yang2016super, chi2016guaranteed, ahmed2014blind, suliman2018blind}. The work in \cite{ahmed2014blind} shows that this assumption exists in applications such as image deblurring and in the framework of channel coding for transmitting a message over an unknown multipath channel. Moreover, \cite{yang2016super} illustrates that the random subspace assumption appears in super-resolution imagining. Finally, in multi-user communication systems \cite{luo2006low}, transmitters may send out a random signal for security and privacy reasons. In such a case, the transmitted signal can be represented in a known low-dimensional random subspace.

Now let $\rv_{j}:= [\tau_{j}, f_{j}]^{T}$ and define $\av \left(\rv_{j}\right) \in \mathbb{C}^{L^{2} \times 1}$ such that
\begin{align}
\label{eq: vec 1}
[\av\left(\rv_{j}\right)]_{\left((k,l),1\right)} = & D_{N}\left(\frac{l}{L}-\tau_{j}\right) D_{N}\left(\frac{k}{L}-f_{j}\right), \nonumber\\
& \ \ \  l,k = -N, \dots, N,
\end{align}
where $D_{N}\left(t\right)$ is the Dirichlet kernel defined by $D_{N}\left(t\right) := \frac{1}{L} \sum_{r=-N}^{N} e^{i2\pi t r}$. Starting from (\ref{eq: model sampled}), and upon using (\ref{eq: vec 1}) and setting $s_{j}\left(l\right)=\dv_{l}^{H}\hv_{j}$, we can write (see \cite[Appendix~A]{suliman2018blind})
\begin{align}
\label{eq: model 3}
y^{*}\left(p\right)= &\sum_{j=1}^{R} c_{j} \sum_{k,l=-N}^{N} [\av\left(\rv_{j}\right)]_{\left((k,l),1\right)} \dv^{H}_{\left(p-l\right)} \hv_{j} e^{i 2 \pi \frac{p k}{L}}.
\end{align}
Now, we intend to express (\ref{eq: model noise}) in a matrix-vector form. For that, we define $\widetilde{\Dm}_{p} \in \mathbb{C}^{L^{2} \times K}$ with $\ p=-N,\dots, N$ such as
\begin{equation}
\label{eq: matrix D}
[\widetilde{\Dm}_{p}]_{\left((k,l),1 \to K\right)} = e^{i 2 \pi \frac{p k}{L}} \dv_{\left(p-l\right)}^{H}, \ \ \ k,l = -N, \dots, N. 
\end{equation}
Based on (\ref{eq: model 3}) and (\ref{eq: matrix D}) we can rewrite (\ref{eq: model noise}) as
\begin{align}
\label{eq: final model}
y\left(p\right) &= \sum_{j=1}^{R} c_{j} \av\left(\rv_{j}\right)^{H} \widetilde{\Dm}_{p} \hv_{j} +\omega\left(p\right)\nonumber\\
& = \text{Tr}\left(\widetilde{\Dm}_{p} \Um_{\text{o}}\right)+\omega\left(p\right)= \left\langle \Um_{\text{o}}, \widetilde{\Dm}_{p}^{H} \right\rangle +\omega\left(p\right),
\end{align}
where
\begin{equation}
\label{eq: U matrix}
\Um_{\text{o}}=\sum_{j=1}^{R} c_{j} \hv_{j} \av\left(\rv_{j}\right)^{H}.
\end{equation}
Define the linear operator $\mathcal{X}: \mathbb{C}^{K \times L^{2}} \to \mathbb{C}^{L}$ as
\begin{equation}
\label{eq: operator}
[\mathcal{X}\left(\Um_{\text{o}}\right)]_{p}= \text{Tr}\left(\widetilde{\Dm}_{p}\Um_{\text{o}}\right), \ \ \ p=-N, \dots, N \nonumber
\end{equation}
and its adjoint $\mathcal{X}^{*}: \mathbb{C}^{L} \to \mathbb{C}^{K \times L^{2}} $ such that\begin{equation}
\mathcal{X}^{*}\left(\qv \right)= \sum_{p=-N}^{N} [\qv]_{p} \widetilde{\Dm}_{p}^{H}; \ \ \qv \in \mathbb{C}^{L \times 1}. \nonumber
\end{equation}
Then, $\yv =[y\left(-N\right), \dots, y\left(N\right)]^{T}$ can be expressed as
\begin{equation}
\label{eq: system in vector}
\yv = \mathcal{X}\left(\Um_{\text{o}}\right) + \omegav = \yv^{*}+ \omegav,
\end{equation}
where $ \omegav = [\omega\left(-N\right), \dots, \omega\left(N\right)]^{T}$ is the noise vector. Equation (\ref{eq: system in vector}) suggests that recovering the unknowns in this paper is a two-fold problem: denoising and parameters estimation using $\Um_{\text{o}}$. In practical systems, $R$ is very small, and $\Um_{\text{o}}$ can thus be viewed as sparse linear combinations of multiple matrices of dimension $K \times L$ in the following set of atoms
 \begin{equation}
\mathcal{A} = \left\lbrace \hv \av\left(\rv\right)^{H}: \ \rv \in [0,1]^{2}, ||\hv||_{2}=1, \hv \in \mathbb{C}^{K \times 1} \right\rbrace. \nonumber
\end{equation}
For a general atomic set $\mathcal{A}$, the atomic norm is defined as the gauge function associated with the convex hull of $\mathcal{A}$, i.e., $\text{conv}\left(\mathcal{A}\right)$, and is given by
\begin{align}
&||\Um||_{\mathcal{A}} = \inf \left\lbrace t>0: \Um \in t \ \text{conv}\left(\mathcal{A}\right) \right\rbrace \nonumber\\
&= \inf_{c_{j} \in \mathbb{C}, \rv_{j}\in [0,1]^{2}, ||\hv_{j}||_{2}=1} \left\lbrace \sum_{j} |c_{j}| : \Um = \sum_{j} c_{j} \hv_{j} \av\left(\rv_{j}\right)^{H} \right\rbrace. \nonumber
\end{align}
The dual of the atomic norm is given by 
\begin{equation}
\label{eq: dual atomic for}
||\Cm||_{\mathcal{A}}^{*} = \sup_{||\Um ||_{\mathcal{A}}\leq 1} \big\langle\Cm,\Um \big\rangle_{\mathbb{R}}  = \sup_{\rv \in [0,1]^{2}, ||\hv||_{2}=1} \left|\big\langle\Cm,\hv \av\left(\rv \right)^{H}\big\rangle_{\mathbb{R}}\right|.
\end{equation} 
Now let $\mu \in \mathbb{R}^{+}$ be an appropriately selected regularization parameter to be determined later. Then, to estimate $\mathcal{X}\left(\Um_{\text{o}}\right)$ from (\ref{eq: system in vector}), we consider solving the following regularized least-square atomic norm minimization problem
\begin{align}
\label{eq: reg problem}
\underset{{\Um} \in \mathbb{C}^{K \times L^2}}{\text{minimize}} \ \ \frac{1}{2} \left|\left|\yv -\mathcal{X}({\Um})\right|\right|_{2}^{2} + \mu \  ||{\Um}||_{\mathcal{A}}.
\end{align} 
Here, $\left|\left|\yv -\mathcal{X}({\Um})\right|\right|_{2}^{2}$ is the noise-controlling term, $||{\Um}||_{\mathcal{A}}$ is the sparsity-enforcing term, and $\mu$ rules the trade-off between them. In this paper, we analyze the MSE performance of (\ref{eq: reg problem}).

Solving (\ref{eq: reg problem}) may appear to be challenging since its primal variable is infinite-dimensional. An approach to solve this problem is to reformulate its dual as an SDP. This is discussed in detail in Section~\ref{sec: optimal con and problem solution}. For now, we show in Appendix~\ref{app: dual proof} that the dual equivalent of (\ref{eq: reg problem}) takes the form
\begin{align}
\label{eq: dual reg problem}
&\underset{\qv \in \mathbb{C}^{L \times 1}}{\text{maximize}} \ \langle\qv, \yv\rangle_{\mathbb{R}} - \frac{1}{2} ||\qv||_{2}^{2} \nonumber\\
 &\text{subject to} : ||\mathcal{X}^{*}\left(\qv\right) ||_{\mathcal{A}}^{*} \leq \mu.
\end{align} 
\section{Main Result}
\label{sec: main results theorem}
We start this section by highlighting the main assumptions in the paper. Then, we provide our main theorem that characterizes the MSE performance of (\ref{eq: reg problem}). 
\begin{assumption}\normalfont
\label{as 1}
We assume that the columns of $\Dm^{H}$, i.e., $\dv_{l} \in \mathbb{C}^{K \times 1}$, are independently drawn from any distribution with their entries being i.i.d. and satisfy the following conditions
\begin{align}
\label{eq: G assumption 1}
&\mathbb{E}[\dv_{l}]= \bm{0},\ \ \ \ \mathbb{E}[\dv_{l} \dv_{l}^{H}]= \Id_{\text{K}},\ \ \ \ l=-N, \dots, N, \\
\label{eq: G assumption 3}
&\mathbb{E}\left[\frac{\dv_{l} \dv_{l}^{H}}{\sum_{i=-N}^{N}||\dv_{i}||_{2}^{2}}\right]= \frac{1}{LK}\Id_{\text{K}}, \ \ l=-N, \dots, N.
\end{align}
\end{assumption}

\begin{assumption}\normalfont
\label{as 2}
({Concentration property}) We assume that the rows of $\Dm^{H}$, denoted in their column-form by $\hat{\dv}_{i} \in \mathbb{C}^{L \times 1}$; $i=1, \dots, K$, are $\widetilde{K}$-concentrated with $\widetilde{K} \geq 1$. That is, there exist two constants $C_{1}$ and $C_{2}$ such that for any 1-Lipschitz function $\varphi : \mathbb{C}^{K} \to \mathbb{R}$ and any $t > 0$ it holds that
\begin{equation}
\text{Pr}\left[ \left| \varphi\left(\hat{\dv}_{i}\right)  - \mathbb{E}\left[ \varphi\left(\hat{\dv}_{i}\right)\right] \right| \geq t \right] \leq C_{1} \exp\left( -C_{2}t^{2}/\widetilde{K}^{2}\right) \nonumber.
\end{equation}
\end{assumption}

\begin{assumption}\normalfont
\label{as 3}
The entries of $\hv_{j}$ are i.i.d. from a uniform distribution on the complex unit sphere with $||\hv_{j}||_{2}=1$.
\end{assumption}

\begin{assumption}\normalfont (Minimum separation)
\label{as 4}
The unknown shifts $\left(\tau_{j}, f_{j}\right) \in [0,1]^{2}, j=1, \dots, R$ satisfy the following separation
\begin{align}
\label{eq: seperation condition}
&\min_{ \forall j,j': j \neq j'} \max \left(|\tau_{j}-\tau_{j'}|,|f_{j}-f_{j'}|\right) \geq \frac{2.38}{N}, \nonumber\\
& \hspace{20pt} \forall [\tau_{j}, f_{j}]^{T}, [\tau_{j'}, f_{j'}]^{T} \in \{\rv_{1}, \dots, \rv_{R}\},
\end{align}
where $|a-b|$ is the wrap-around distance on the unit circle.
\end{assumption}

\subsection*{Remarks on the Assumptions:}

The assumption in (\ref{eq: G assumption 1}) is sometimes being referred to as isotropy condition in the literature. This assumption is first motivated by Cand\`es and Plan in their work on the RIPless theory of compressed sensing \cite{candes2011probabilistic}; then, it was imposed for subspace models in many related works. An example of vectors that satisfy this condition is the sensing vectors with independent components such as the Gaussian measurements ensemble and the binary measurements ensemble (which defines the single-pixel camera's sensing mechanism). Another example is vectors formulated by subsampling a tight or a continuous frame as in magnetic resonance imaging (MRI). Moreover, this assumption is also imposed on the problem of blind spikes deconvolution \cite{chi2016guaranteed}. On the other hand, we believe that (\ref{eq: G assumption 3}) is an artifact assumption for our proofs. A closely similar version to this assumption is also imposed in \cite{li2019atomic}. Examples where the vector entries satisfy (\ref{eq: G assumption 3}), is when they are drawn from a Rademacher or Gaussian distribution. 

The concentration property, on the other hand, is satisfied by many random vectors in practice. For example, if the entries of $\hat{\dv}_{i}$ are i.i.d. complex Gaussian, $\hat{\dv}_{i}$ is 1-concentrated. In contrast, if each element in $\hat{\dv}_{i}$ is upper bounded by a constant $C$, then $\hat{\dv}_{i}$ is a $C$-concentrated vector \cite[Theorem F.5]{tao2010random}. Thus, this property is, in fact, a more relaxed (general) assumption than the incoherence assumption imposed on the elements of
the low-dimensional subspace matrix in some existing super-resolution works \cite{yang2016super, chi2016guaranteed, ahmed2014blind, li2019atomic}.

It is worth mentioning that the random assumptions on $\Dm$ and $\hv_{j}$ do not appear to be crucial in practice and are doubtful to be artifacts for our proofs. This fact is also observed in closely similar works in the field, e.g., \cite{suliman2018blind,yang2016super, chi2016guaranteed}. By looking from a different viewpoint, the imposed randomness assumptions on $\Dm$ can be seen as a way to obtain random measurement results from $\Um_{\text{o}}$ based on (\ref{eq: operator}). As known, random measurements are crucial in the derivation of theoretical and empirical results \cite{candes2011probabilistic}. A future extension to this work should look at eliminating these assumptions upon modifying our proof techniques or applying a different proof strategy instead of the dual analysis of the atomic norm. In this paper, we will provide a single simulation experiment to show that our framework still works perfectly when these random assumptions are not satisfied.  

The minimum separation between the time-frequency shifts in Assumption~\ref{as 4} is essential to prevent the shifts from being too clumped together, which results in a severely ill-posed problem \cite{candes2014towards} \cite[Section~2.2]{fern2015superresolution}. To elaborate more, let us consider the more straightforward 1D non-blind super-resolution problem (such as the line spectral estimation problem), which can be obtained from (\ref{eq: model}) by assuming that $s_{j}\left(t\right)$ are known for all $j$ and that $\tau_{j}=0, \forall j$. Obviously, any separation condition necessary for this simpler problem is also necessary for our problem. The first point to raise here is that the work in \cite{moitra2015super} shows that when the separation between the frequencies is less than $2/N$, there exists a pair of spike signals with the same minimal separation, such that no estimator can distinguish them. On the other hand, assume that there are $R'$ elements with frequencies $f_{j}$ in an interval of length less than $\frac{2R'}{L}$. Then, estimating $\left(c_{j}, f_{j}\right)$ is a severely ill-posed problem when $R'$ is large \cite[Section 1.7]{candes2014towards}\cite[Theorem 1.1]{donoho1992superresolution}. Now, in the presence of a tiny amount of noise, stable recovery from $y\left(t\right)$ is impossible. Going back to our super-resolution problem, Assumption~\ref{as 4} allows us to have $0.42R'$ time-frequency shifts in the interval of length $\frac{2R'}{L}$, which is optimal up to the constant $0.5$. We stress that while separation is essential for our super-resolution problem, the form in (\ref{eq: seperation condition}) is a sufficient condition, and smaller separation is expected to be enough.

Now, we are ready to present our main theorem as follows: 
\begin{theorem}
\label{th: main result}
Consider the linear system in (\ref{eq: model}) and its sampled version in (\ref{eq: model sampled}) and assume that the unknown waveforms vectors can be written as $\sv_{j}=\Dm \hv_{j}$ where $\Dm$ satisfies Assumptions~\ref{as 1} and \ref{as 2} while $\hv_{j}$ follows Assumption~\ref{as 3}. Further, let the unknown shifts satisfy the minimum separation in Assumption~\ref{as 4}. Consider the noisy model $\yv = \mathcal{X}\left(\Um_{\text{o}}\right) + \omegav$ where the entries of $\omegav$ are i.i.d. complex Gaussian of zero mean and variance $\sigma_{\omegav}^{2}$. Then, an estimate signal $\hat{\yv}=\mathcal{X}(\widehat{\Um})$, obtained by solving (\ref{eq: reg problem}) with $\mu = 6 \lambda \sigma_{\omegav} ||\Dm||_{F} \sqrt{\log\left(N\right)}$ and $\lambda \geq 1$, satisfies
\begin{align}
\label{eq: denosing error result}
\frac{1}{L}||\yv^{*}\hspace{-1pt}-\hat{\yv}||_{2}^{2} \hspace{-1pt}\leq  \bar{C} \lambda^{2} \sigma_{\omegav}^{2} \frac{\sqrt{K^{3} R}}{L^{3/2}}  \log\left(N\right)\sqrt{\log ( \tilde{C}\left(K+1\right)N\hspace{-1pt})} 
\end{align}
with probability at least $1-\frac{C}{N}$,  provided that $L \geq \bar{C}_{1} R K \widetilde{K}^{4} \log^{2}\left(\frac{  \tilde{C}_{1} R^{2} K^{2} L^{3} }{\delta}\right)\log^{2}\left(\frac{ \tilde{C}_{1} (K+1)  L^{3} }{\delta}\right)$, where $C, \bar{C}, \tilde{C}, \bar{C}_{1}, \tilde{C}_{1}$ are numerical constants and $\delta >0$.
\end{theorem}

\subsection*{Remarks on Theorem~\ref{th: main result}:}
Theorem~\ref{th: main result} shows that the MSE scales linearly with the noise variance and the square-root of the number of unknown shifts. The MSE is also directly proportional to the subspace dimension at a scale of $\mathcal{O}\left(\sqrt{K^{3} \log\left(K+1\right)}\right)$. Moreover, Theorem~\ref{th: main result} shows that the MSE is inversely proportional to the number of samples $L$. On the other hand, and since the choice of the regularization parameter $\mu$ is crucial in the framework's performance, Theorem~\ref{th: main result} shows that $\mu$ is directly proportional to the noise level, the Frobenius norm of the subspace matrix, and the square-root of $\log\left(N\right)$.

Furthermore, the bound on $L$, which is initially derived in \cite{suliman2018blind}, shows that the more concentrated are the rows of $\Dm^{H}$, the fewer number of samples are needed. For a given $\tilde{K}$, the bound states that having $L = \mathcal{O}\left(RK\right)$, which is the same as the number of degrees of freedom in the problem, is a sufficient condition for the MSE bound to be achieved. Our simulations show that the theorem is still valid when this bound is not satisfied. The work on 1D blind super-resolution in \cite{li2019atomic} has a bound of $\mathcal{O}\left(R^{2}K\right)$ on the sample complexity, but without the random assumption on $\hv$. It will be interesting to see how dropping the random assumption on $\hv$, or some of our other assumptions will affect our sample complexity bound in the future extension of this work.

The detailed proof of Theorem~\ref{th: main result}, along with the choice of $\mu$ is provided in Section~\ref{sec: proof on theorem}.

\section{Optimality Conditions and Problem Solution}
\label{sec: optimal con and problem solution}
We start by establishing general properties about (\ref{eq: reg problem}) and its dual certificate (\ref{eq: dual reg problem}). First, the following lemma provides the optimality conditions for $\widehat{\Um}$ as the solution of (\ref{eq: reg problem}).
\begin{lemma}
\label{lemma: suff cond lemma}
The matrix $\widehat{\Um}$ is the solution of (\ref{eq: reg problem}) if and only if
\begin{equation}
\label{eq: suff cond 1}
\left|\left|\mathcal{X}^{*}\left(\yv-\mathcal{X}(\widehat{\Um})\right)\right|\right|_{\mathcal{A}} ^{*}\leq \mu
\end{equation}
and
\begin{equation}
\label{eq: suff cond 2}
\left\langle\mathcal{X}^{*}\left(\yv-\mathcal{X}(\widehat{\Um})\right),\widehat{\Um}\right\rangle = \mu ||\widehat{\Um}||_{\mathcal{A}}.
\end{equation}
\end{lemma}
The proof of Lemma~\ref{lemma: suff cond lemma} is provided in Appendix~\ref{app: suff cond}.

To obtain $\mathcal{X}(\widehat{\Um})$ using (\ref{eq: dual reg problem}), it is clear that since the dual objective function is a strong concave function, (\ref{eq: dual reg problem}) admits a unique solution $\qv$. Then, based on Lemma~\ref{lemma: suff cond lemma} and the strong duality between (\ref{eq: reg problem}) and (\ref{eq: dual reg problem}) we can conclude that
\begin{equation}
\yv = \mathcal{X}(\widehat{\Um})+{\qv}  \nonumber
\end{equation}
which makes $\mathcal{X}(\widehat{\Um})$ directly obtainable from $\qv$.

In the remaining part of this section, we discuss how to estimate the unknowns in (\ref{eq: model}) from $\widehat{\Um}$, and we also show how to solve (\ref{eq: dual reg problem}) numerically. We leave addressing the theoretical performance of the algorithm on estimating the parameters for future work. Starting from the constraint of (\ref{eq: dual reg problem}), and by setting $\widetilde{\mathcal{X}}^{*}\left(\qv\right)=\frac{1}{\mu} \mathcal{X}^{*}\left(\qv\right)$, we can write based on (\ref{eq: dual atomic for})
\begin{align}
||\widetilde{\mathcal{X}}^{*}\left(\qv\right) ||_{\mathcal{A}}^{*}&= \sup_{\rv \in [0,1]^{2}, ||\hv||_{2}=1} \big| \big\langle \hv, \widetilde{\mathcal{X}}^{*}\left(\qv\right)  \av\left(\rv\right) \rangle\big| \nonumber\\
&= \sup_{\rv \in [0,1]^{2}} || \widetilde{\mathcal{X}}^{*}\left(\qv\right)  \av\left(\rv\right)||_{2} \leq 1. \nonumber
\end{align}
Therefore, the constraint in (\ref{eq: dual reg problem}) is equivalent to the fact that the norm of the 2D trigonometric vector polynomial $\fv\left(\rv\right)=\widetilde{\mathcal{X}}^{*}\left(\qv\right) \av\left(\rv\right)$ is bounded by 1. Such polynomial is formulated and discussed in detail in \cite[Section~VI]{suliman2018blind}. The next proposition, which is given in \cite[Proposition~1]{suliman2018blind}, and is a consequence of the strong duality between (\ref{eq: reg problem}) and (\ref{eq: dual reg problem}), provides an approach for identifying the unknown 2D shifts in the support of the solution $\widehat{\Um}$ upon using the dual solution.
\begin{proposition}
\label{pro: main pro}
\cite[Proposition~1]{suliman2018blind}
Let $\mathcal{R}= \left\lbrace \rv_{j}\right\rbrace_{j=1}^{R}$ and refer to the primal-dual solution pair by $(\widehat{\Um},{\qv})$. Then, $\widehat{\Um}$ is the unique optimal solution of (\ref{eq: reg problem}) if
\begin{enumerate}
\item There exists a 2D trigonometric vector polynomial $\fv\left(\rv\right)=\widetilde{\mathcal{X}}^{*}\left({\qv}\right) \av\left(\rv\right)$ such that: 
\begin{equation}
\label{eq: hold assump 1}
\fv\left(\rv_{j}\right)= \text{sign}\left(c_{j}\right) \hv_{j}, \ \ \forall \rv_{j} \in \mathcal{R}
\end{equation}
\begin{equation}
\label{eq: hold assump 2}
||\fv\left(\rv\right)||_{2} < 1 , \ \ \forall \rv \in [0,1]^{2}\setminus \mathcal{R},
\end{equation}
where $\text{sign}\left(c_{j}\right)= \frac{c_{j}}{|c_{j}|}$.

\item $\left\lbrace\begin{bmatrix}\av\left(\rv_{j}\right)^{H} \widetilde{\Dm}_{-N} \\  \vdots \\ \av\left(\rv_{j}\right)^{H} \widetilde{\Dm}_{N} \end{bmatrix}\right\rbrace_{j=1}^{R}$ is a linearly independent set.
\end{enumerate}
\end{proposition}
Based on Proposition~\ref{pro: main pro}, and once we obtain ${\qv}$ by solving (\ref{eq: dual reg problem}), we can formulate $\fv\left(\rv\right)$ as a function of $\rv$ and then estimate $\mathcal{R}$ by either computing the roots of the polynomial $1-||\fv\left(\rv\right)||_{2}^{2}$ on the unit circle or by discretizing the domain $[0,1]^{2}$ on a fine grid and then locating $\hat{\rv}_{j}$ at which $||\fv\left(\hat{\rv}_{j}\right)||_{2} = 1$ (based on Proposition~\ref{pro: main pro} and the fact that $||\hv_{j}||_{2}=1$) (see the discussion in \cite[Section 4]{candes2014towards}). In this paper, we use the latter approach. To estimate ${c}_{j}{\hv}_{j}$, we formulated an overdetermined linear system based on $\yv = \mathcal{X}(\widehat{\Um})$ and then we solve it using the least-square algorithm.

Finally, we address the question of how we can solve (\ref{eq: dual reg problem}). For that, we follow the path of obtaining an SDP relaxation for it using the result in \cite{xu2014precise} as in \cite[Section~IV-C]{suliman2018blind}. First, we define the matrix $\widehat{\Qm} \in \mathbb{C}^{K \times L^2}$ such that 
\begin{equation}
\left[\widehat{\Qm}\right]_{(i,(p,k))}  =  \hspace{-2pt}\left[\frac{1}{\mu L}\qv\left(p\right) e^{i2\pi \frac{k p}{L}} \hspace{-5pt}\sum_{l=-N}^{N} \hspace{-3pt}\dv_{l}e^{-i2\pi \frac{k l}{L}}\right]_{i}\hspace{-2.5pt}, \hspace{-2pt}i=1,\dots, K. \nonumber
\end{equation}
Then, we can show after some manipulations that the equivalent SDP relaxation of (\ref{eq: dual reg problem}) is given by \cite[Section~IV-C]{suliman2018blind}
\begin{align}
\label{eq: dual of the dual}
&\underset{\qv, \Qm\succeq \bm{0}}{\text{maximize}} \ \langle\qv, \yv\rangle_{\mathbb{R}} - \frac{1}{2} ||\qv||_{2}^{2} \ \ \ \text{subject to} :\nonumber\\
 & \begin{bmatrix} \Qm  & \widehat{\Qm}^{H}  \\ \widehat{\Qm} &  \Id_{\text{K} \times \text{K} }  \end{bmatrix} \succeq \bm{0}, \hspace{2pt} \text{Tr}\left(	\left( \widetilde{\Thetam}_{\tilde{l}} \otimes \widetilde{\Thetam}_{\tilde{k}}\right)\Qm\right) =  \delta_{\tilde{l},\tilde{k}},
\end{align}
where $-(L-1) \leq \tilde{l},\tilde{k} \leq (L-1) $ while $\widetilde{\Thetam}_{i}$ is $L \times  L$ Toeplitz matrix with ones on its $i$-th diagonal and zeros elsewhere. Finally, $\delta_{\tilde{l},\tilde{k}}$ is the Dirac function, i.e., $\delta_{\tilde{l},\tilde{k}} =1$ iff $\tilde{l}=\tilde{k}=0$.  
\subsection*{Remarks on (\ref{eq: dual of the dual}):}
The optimization in (\ref{eq: dual of the dual}) is the equivalent SDP relaxation of (\ref{eq: dual reg problem}) upon applying the result in \cite{xu2014precise}. The problem can then be solved to obtain $\qv$ using any SDP solver, such as CVX and Yalmip. As \cite{xu2014precise} shows, this relaxation comes from the fact that the matrices that are used to express the dual constraint of (\ref{eq: dual reg problem}) as in (\ref{eq: dual of the dual}) are, generally speaking, of unspecified dimensions, and an approximation for their dimensions is required. This is opposed to the 1D case, as in \cite{tang2013compressed,candes2014towards, li2019atomic}, where the dimension of $\Qm$ is known to be $L \times L$. In contrast, the dimension of $\Qm$ in our case is not precisely known and may, in practice, be greater than the minimum size of $L^{2} \times L^{2}$. Therefore, fixing the size of $\Qm$ at $L^{2} \times L^{2}$ yields a relaxation to the constraint of (\ref{eq: dual reg problem}). Using a larger value than $L$ will lead, in theory, to a better approximation to the problem. However, using the minimum value $L$ is known to yield the optimal solution in practice. This fact is observed and discussed in detail in \cite{heckel2016super, dumitrescu2007positive}.

The theory behind this relaxation is discussed in \cite{xu2014precise} and in further detail in \cite[Corollary. 4.25]{dumitrescu2007positive}. Hence, we will only highlight the main idea behind it here. First, consider a symmetric $d$-variate trigonometric polynomial of degree $\mv \in \mathbb{N}^{d}$ in the form
\begin{equation}
\label{eq: poly extra info}
R\left(\thetav\right) = \sum_{\pv = -\mv}^{\mv} R_{\pv} e^{i \pv^{T} \thetav}, \ \ \thetav \in [0, 2 \pi)^{d},
\end{equation}
where the index $\pv$ takes values such that $-\mv \leq \pv \leq \mv$. Note that $d=2$ in our case. It is shown in \cite{dumitrescu2007positive}, \cite[Section 5]{dritschel2004factorization} that a polynomial in the form as in (\ref{eq: poly extra info}) can be written as a sum-of-squares in the form
\begin{equation}
\label{eq: poly extra info 2}
R\left(\thetav\right)  = \sum_{l=1}^{\xi}\left| W_{l}\left(\thetav\right)\right|^{2},
\end{equation}
where the degrees of the polynomials $W_{l}\left(\thetav\right)$ can be arbitrarily high and must be bounded to a relaxation vector $\nv = \left[n_{1}, n_{2}, \dots, n_{d}\right]^{T}$ in practical implementation. Hence, with the sum-of-squares relaxation degree vector $\nv$ in place, the sum in (\ref{eq: poly extra info 2}) is equivalent to that in (\ref{eq: poly extra info}) up to an error factor.

Following that, it is shown in \cite{xu2014precise, dumitrescu2007positive, mclean2002spectral, dumitrescu2004multidimensional} that a polynomial in the form as in (\ref{eq: poly extra info}) can be written as a sum-of-squares of degree $\nv$ if and only if there exists a matrix $\Hm \succeq \bm{0}$ of dimension $N \times N$ associated with $R\left(\thetav\right)$ such that
\begin{equation}
\label{eq: parameters for the}
R_{\pv} = \text{Tr}\left(\widetilde{\Thetam}_{\pv} \Hm \right), \ \ \ \widetilde{\Thetam}_{\pv} =  \widetilde{\Thetam}_{p_{d}} \otimes \dots \otimes \widetilde{\Thetam}_{p_{1}},
\end{equation}
where $N = \Pi_{i=1}^{d} \left(n_{i}+1\right)$ is based on the relaxation degree vector $\nv$. Therefore, $N \times N$ is not the actual dimension of $\Hm$ (i.e., must be larger) and is, in fact, a relaxed dimension. Finally, note that the parameters in (\ref{eq: parameters for the}) are applied to convert optimization problems with positive polynomials into SDP forms. Since we must replace the positive polynomials with sum-of-squares of a bounded degree to do such transform, the obtained SDP problems are relaxed versions of the original ones.
 
Going back to our formulation, and as we have shown at the beginning of this section, the bound on the dual of the atomic norm is equivalent to the fact that the norm of a 2D trigonometric vector polynomial $\fv\left(\rv\right)$ is bounded by one. This polynomial is shown in \cite[Equation (26)]{suliman2018blind} to be formulated as a sum of two positive trigonometric vector polynomials of degree $(L,L)$ in the form
\begin{eqnarray}
\label{eq: simple polynomail denoise}
\fv\left(\rv\right)= \hspace{-4pt}\sum_{p,k=-N}^{N}\hspace{-2pt} \left(\frac{1}{\mu L}[\qv]_{p} \hspace{-2pt}\sum_{l=-N}^{N} \hspace{-2pt}\dv_{l}e^{\frac{i2\pi k \left(p-l\right)}{L}}\right) e^{-i2\pi \left(k\tau+pf\right)}. 
\end{eqnarray}
Based on the above discussion, we know that the sum-of-squares representation of our trigonometric polynomial in (\ref{eq: simple polynomail denoise}) possibly involves factors of degree larger than $(L,L)$ and that the constraint of (\ref{eq: dual of the dual}) does not precisely characterize the constraint of (\ref{eq: dual reg problem}); instead, it only provides a sufficient condition.

On the other hand, and as (\ref{eq: dual of the dual}) shows, solving the denoising problem as well as estimating the time-frequency shifts using the dual of the atomic norm minimization problem involves dealing with optimization variables with dimensions $L^2 \times L^2$. Now, any standard solver that solves this optimization problem will have a computational cost of order $\mathcal{O}\left(L^{6}\right)$ \cite[Section 6.4]{nesterov1994interior}, \cite{van1983matrix}. Since this cost increases dramatically for larger $L$, this makes the problem infeasible for real-world applications. This fact limits the framework's applicability and prohibits us from evaluating the algorithm's performance at large values of $L$. Addressing this issue could include applying alternative optimization techniques as in \cite{rao2015forward}, applying the alternating direction method of multipliers (ADMM) to solve (\ref{eq: dual of the dual}) as in \cite{bhaskar2013atomic}, or using a discrete approximation of the atomic norm as in \cite{tang2013sparse}. We leave tackling this issue for the future extension of this work.

\section{Numerical Implementation}
\label{sec: results}
In this section, we validate the theoretical findings of this paper by using simulations. In all experiments, and unless otherwise stated, we let $\hv_{j}$ satisfy Assumption~\ref{as 3}, and we set the regularization parameter $\mu$ as in Theorem~\ref{th: main result} with $\lambda=1.2$. Moreover, we generate the entries of $\omegav$ from i.i.d. complex Gaussian distribution of zero mean and a variance that varies between the experiments, and we average our results over $50$ noise iterations. Finally, we use the CVX solver, which calls SDPT3, to solve (\ref{eq: dual of the dual}). 

In the first experiment, we intend to validate the performance of our denoising framework as well as to verify the dependency of the MSE on the number of samples $L$. For that, we set $K=2, R=1, [\Dm]_{i,j} \stackrel{i.i.d.}{\sim} \mathcal{CN}(0,1)$, and we let the shifts to be $(0.13,0.67)$. Moreover, we fix the noise variance $\sigma_{\omegav}^{2}$ at $0.15$, and we let $c_{1}=1$. Finally, we vary the value of $L$ from 11 to 21. In Fig~\ref{fig: perfromance for N_R1}, we plot the MSE versus $L$ for $\yv$, and we compare it with the MSE of our proposed framework obtained by solving (\ref{eq: dual of the dual}). Fig~\ref{fig: perfromance for N_R1} shows that our framework provides a significant reduction in the MSE. On the other hand, Fig~\ref{fig: normalized for N_R1} plots the scaled MSE, i.e., $\frac{L^{3/2}}{\left(\log\left(N\right)\right)^{3/2}}\left(\frac{1}{L}||\hat{\yv}-\yv^{*}||_{2}^{2}\right)$ versus $L$, and shows that the MSE does scale with $\mathcal{O}\left(L^{-3/2}\left(\log\left(N\right)\right)^{3/2}\right)$ as Theorem~\ref{th: main result} indicates. Thus, the MSE of our denoising framework decreases with the increase of the number of samples. 
\begin{figure}[h!]
\centering
\subfigure[]{\label{fig: perfromance for N_R1}\includegraphics[width=2.5in]{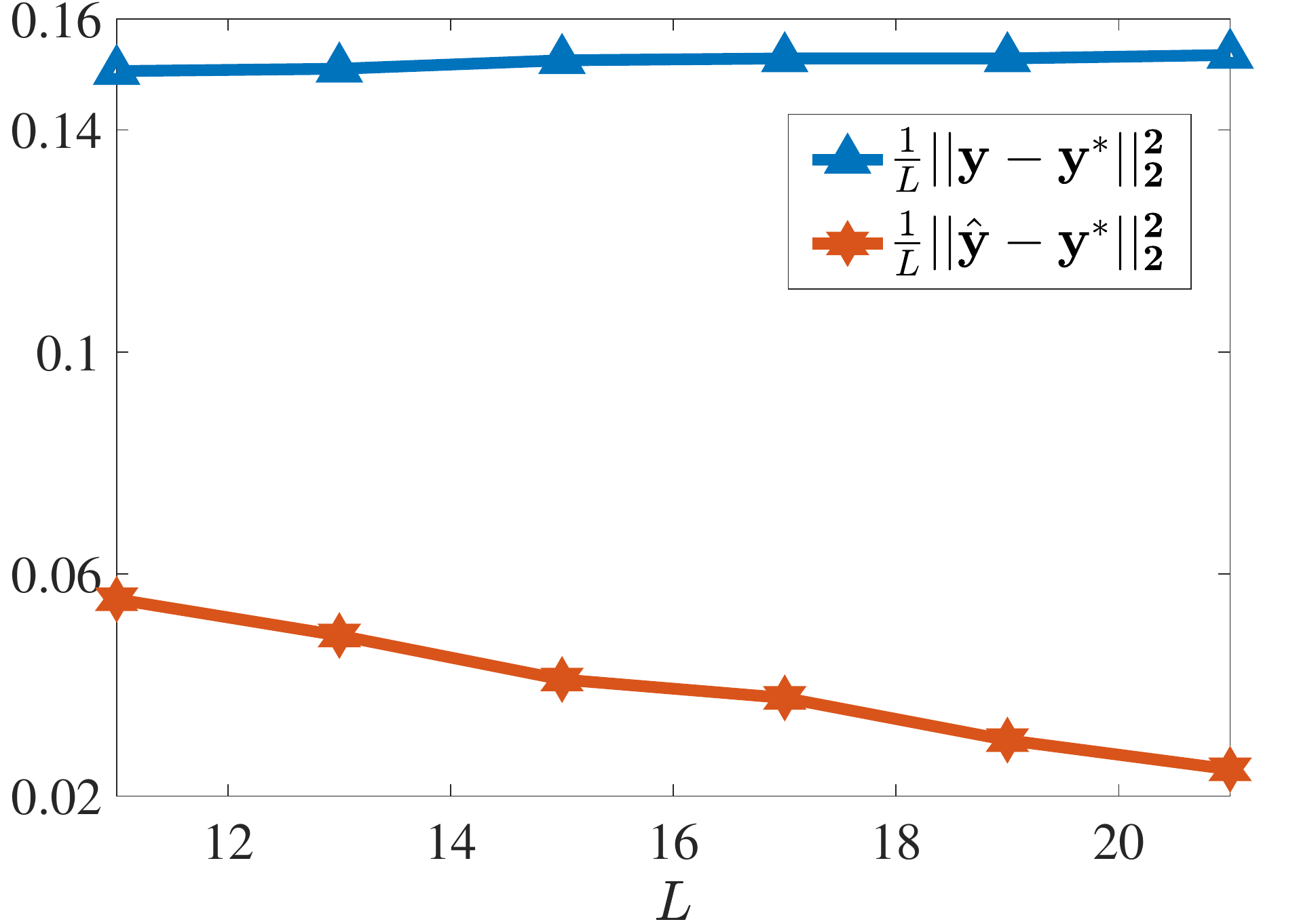}} 
\subfigure[]{\label{fig: normalized for N_R1}\includegraphics[width=2.5in]{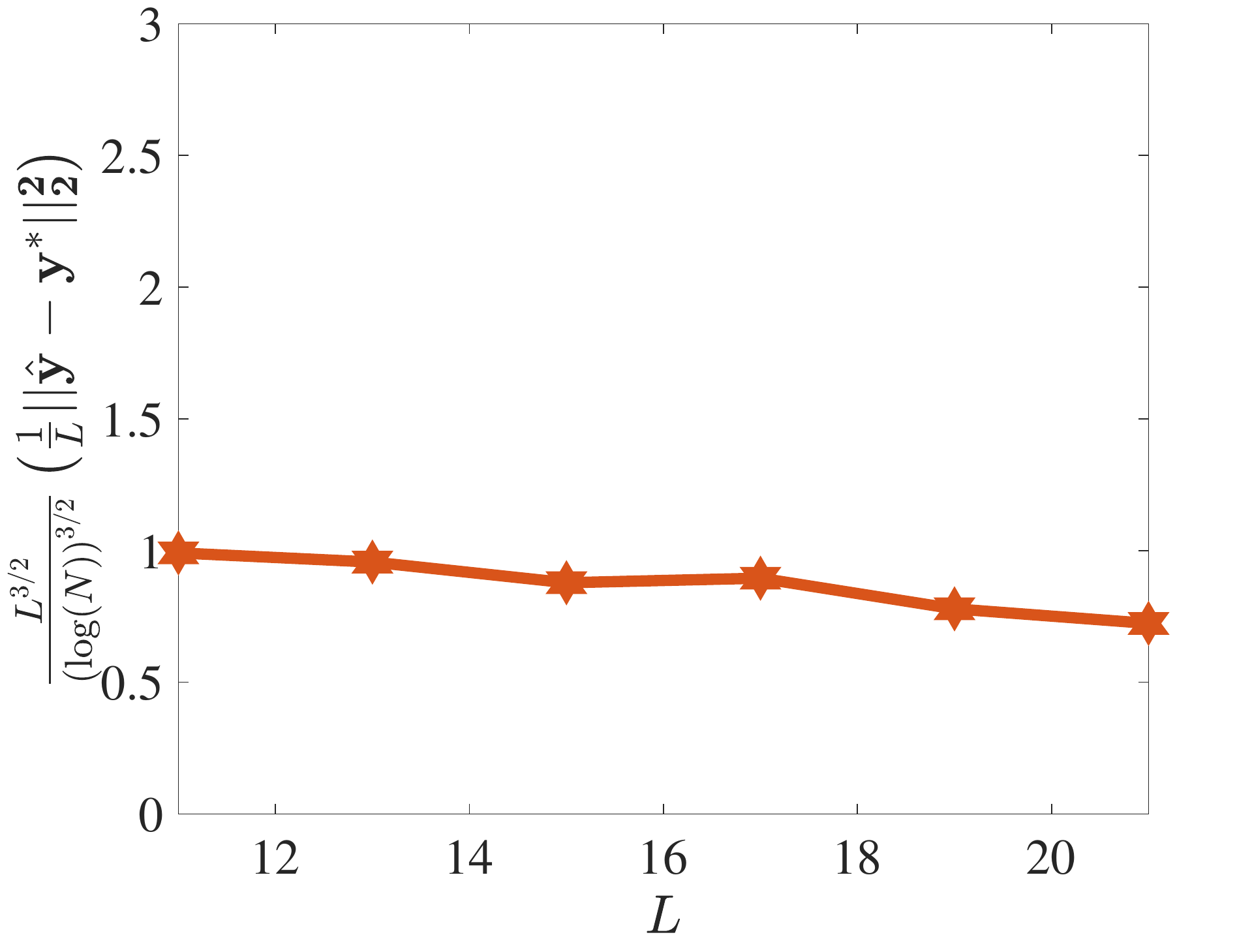}}
\caption{(a) The denoising performance of the proposed framework. (b) The relationship between $\frac{L^{3/2}}{\left(\log\left(N\right)\right)^{3/2}}\left(\frac{1}{L}||\hat{\yv}-\yv^{*}||_{2}^{2}\right)$ and $L$.}
\label{fig: perfromance}
\end{figure}

Next, we provide an additional simulation experiment to validate the dependency of the MSE on $L$. Here, we set $K=3, R=2$, and we let $[\Dm]_{i,j} \stackrel{i.i.d.}{\sim} \mathcal{CN}(0,1)$. The real and imaginary parts of $c_{1}$ and $c_{2}$ are set to be fading, i.e., $0.5+g^{2}$; $g \sim \mathcal{N}(0,1)$ with a uniform sign, while the unknown shifts are generated randomly and found to be $(0.51, 0.30)$ and $(0.94, 0.73)$. Finally, we set $\sigma_{\omegav}^{2}=0.3$, and we vary $L$ from $13$ to $21$. Fig~\ref{fig: normalized for N_R2 sub} shows that the MSE of the framework does scale with $\mathcal{O}\left(L^{-3/2}\left(\log\left(N\right)\right)^{3/2}\right)$ as in Theorem~\ref{th: main result}.

In Fig~\ref{fig: runtime}, we plot the average runtime of the framework versus $L$. As discussed in Section~\ref{sec: optimal con and problem solution}, the figure clearly illustrates the framework's substantial computational cost, resulting mainly from the large dimensions of the constraining matrices. For example, when $L=21$, solving (\ref{eq: dual of the dual}) takes around 30 minutes.

\begin{figure}[h!]
\centering
\subfigure[]{\label{fig: normalized for N_R2 sub}\includegraphics[width=2.5in]{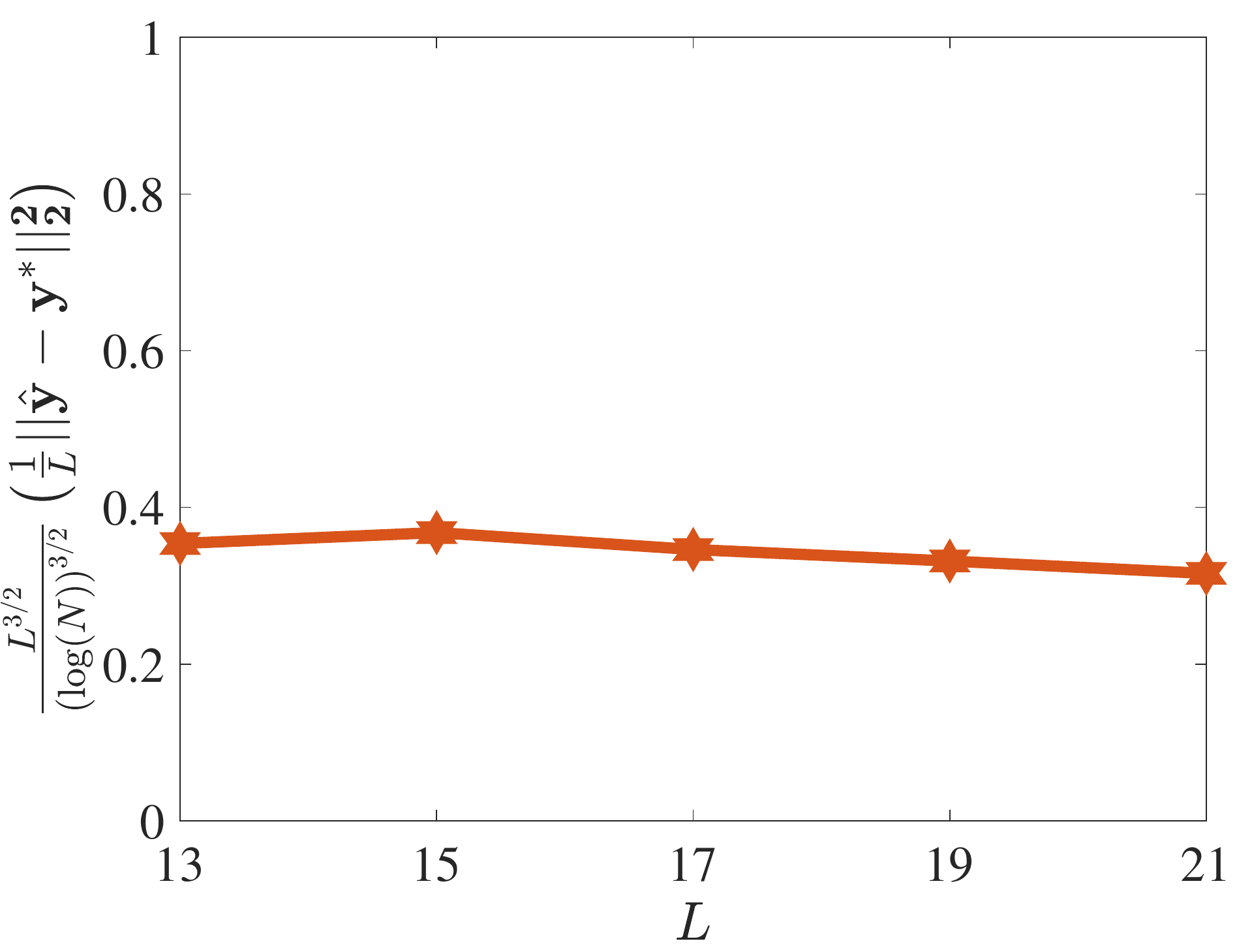}}
\subfigure[]{\label{fig: runtime}\includegraphics[width=2.5in]{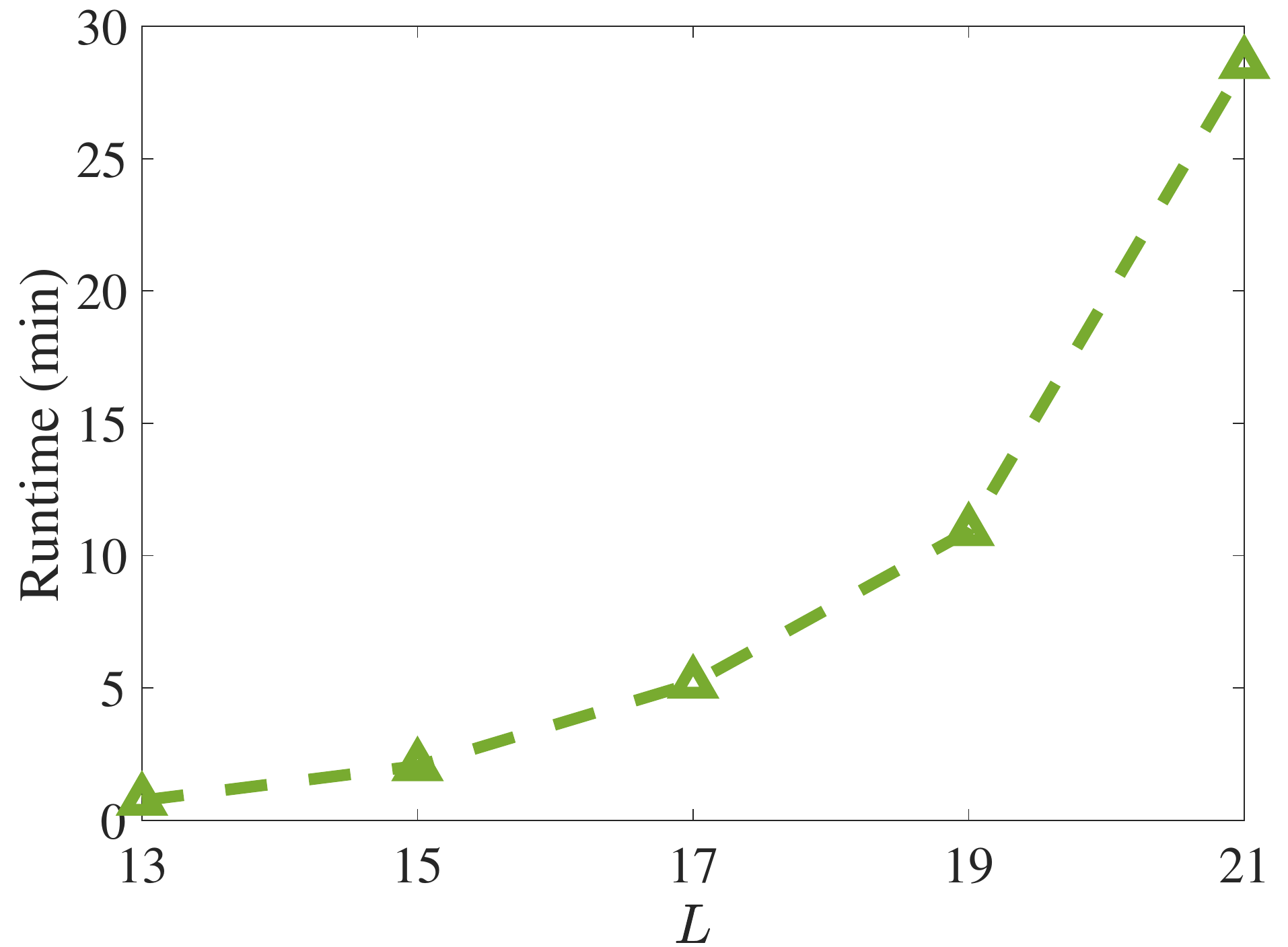}}
\caption{(a) The relationship between $\frac{L^{3/2}}{\left(\log\left(N\right)\right)^{3/2}}\left(\frac{1}{L}||\hat{\yv}-\yv^{*}||_{2}^{2}\right)$ and $L$. (b) The average runtime of the framework (in min).}
\label{fig: normalized for N_R2}
\end{figure} 
 
In the third experiment, we study the dependency of the MSE on $\sigma_{\omegav}^{2}$. For that, we set $L=15, K=3$, $R=3$, and we generate $[\Dm]_{i,j}$ from a Rademacher distribution. The values of $\{c_{j}\}_{j=1}^{3}$ are set as in the previous scenario, whereas the shifts are generated randomly and found to be $(0.1,0.46), (0.61,0.80),$ and $(0.94,0.13)$. Finally, we vary $\sigma_{\omegav}^{2}$ from $0.1$ to $0.75$ and calculate the MSE at each $\sigma_{\omegav}^{2}$ value. Fig~\ref{fig: normalized for noise} indicates that the MSE of our framework does scale linearly with $\sigma_{\omegav}^{2}$ as Theorem~\ref{th: main result} shows.

\begin{figure}[h!]
\centering
{\includegraphics[width=2.5in]{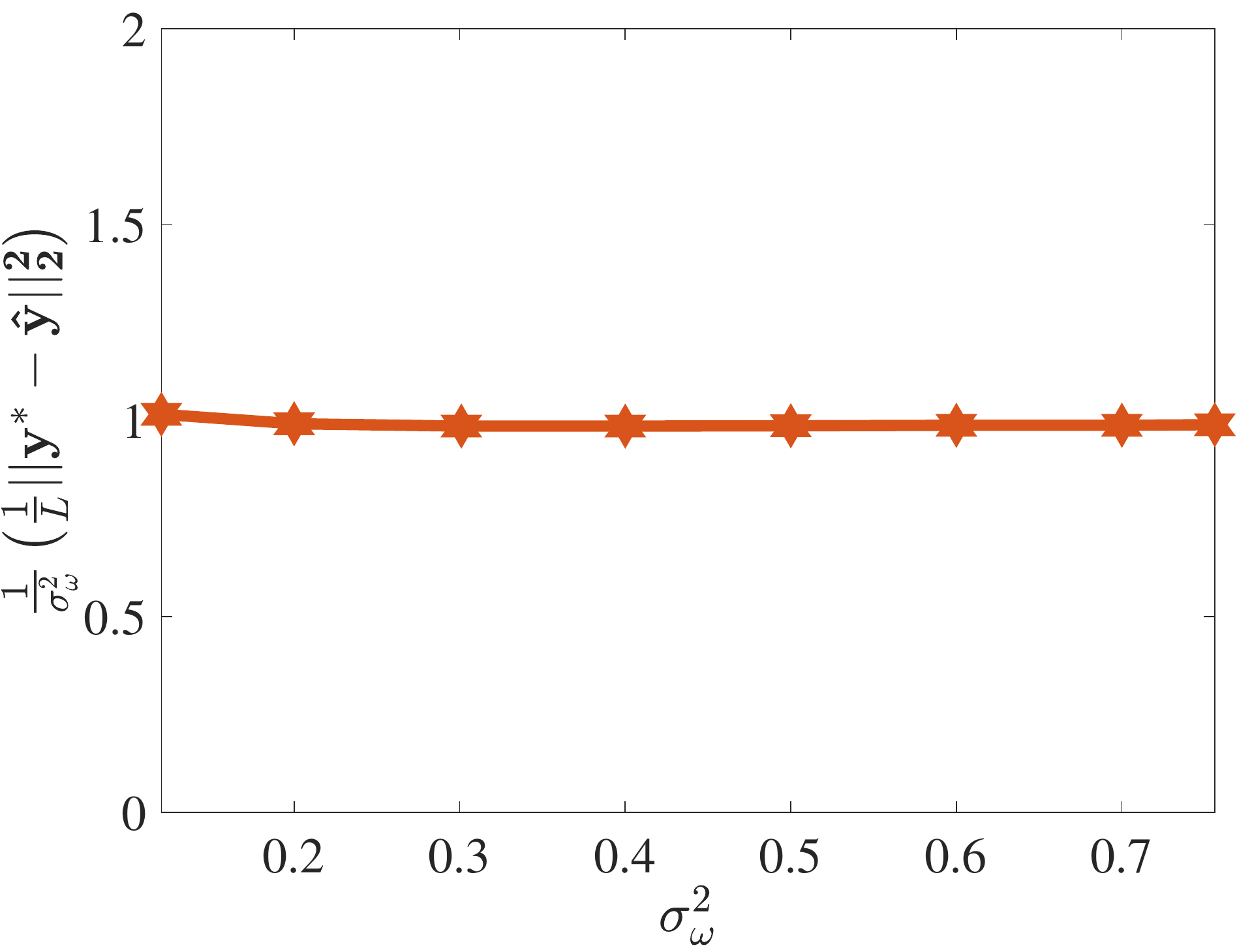}}
\caption{The relationship between $\frac{1}{\sigma_{\omegav}^{2}}\left(\frac{1}{L}||\hat{\yv}-\yv^{*}||_{2}^{2}\right)$ and $\sigma_{\omegav}^{2}$.}
\label{fig: normalized for noise}
\end{figure}

Next, we study the dependency of the MSE on $R$. In this experiment, we set $L=19, K=3, \sigma_{\omegav}^{2}=0.15, [\Dm]_{i,j} \stackrel{i.i.d.}{\sim} \mathcal{CN}(0,1)$, and we vary $R$ from 1 to 7. The values of $\{c_{j}\}_{j=1}^{R}$ are set to be $c_{j}=j$, whereas the shifts are generated as $\tau_{j} \in \{0.1:\zeta:0.1+\zeta R\}, f_{j} \in \{ 0.5:\zeta:0.5+\zeta R\}$, $j=1, \dots, R$, where $\zeta= 2.38/N$. In Fig~\ref{fig: normalized for R}, we plot the normalized MSE, i.e., $\frac{1}{\sqrt{R}}\left(\frac{1}{L}||\hat{\yv}-\yv^{*}||_{2}^{2}\right)$ versus $R$. Fig~\ref{fig: normalized for R} indicates that the MSE does scale with $\mathcal{O}(\sqrt{R})$ as Theorem~\ref{th: main result} shows.

\begin{figure}[h!]
\centering
{\includegraphics[width=2.5in]{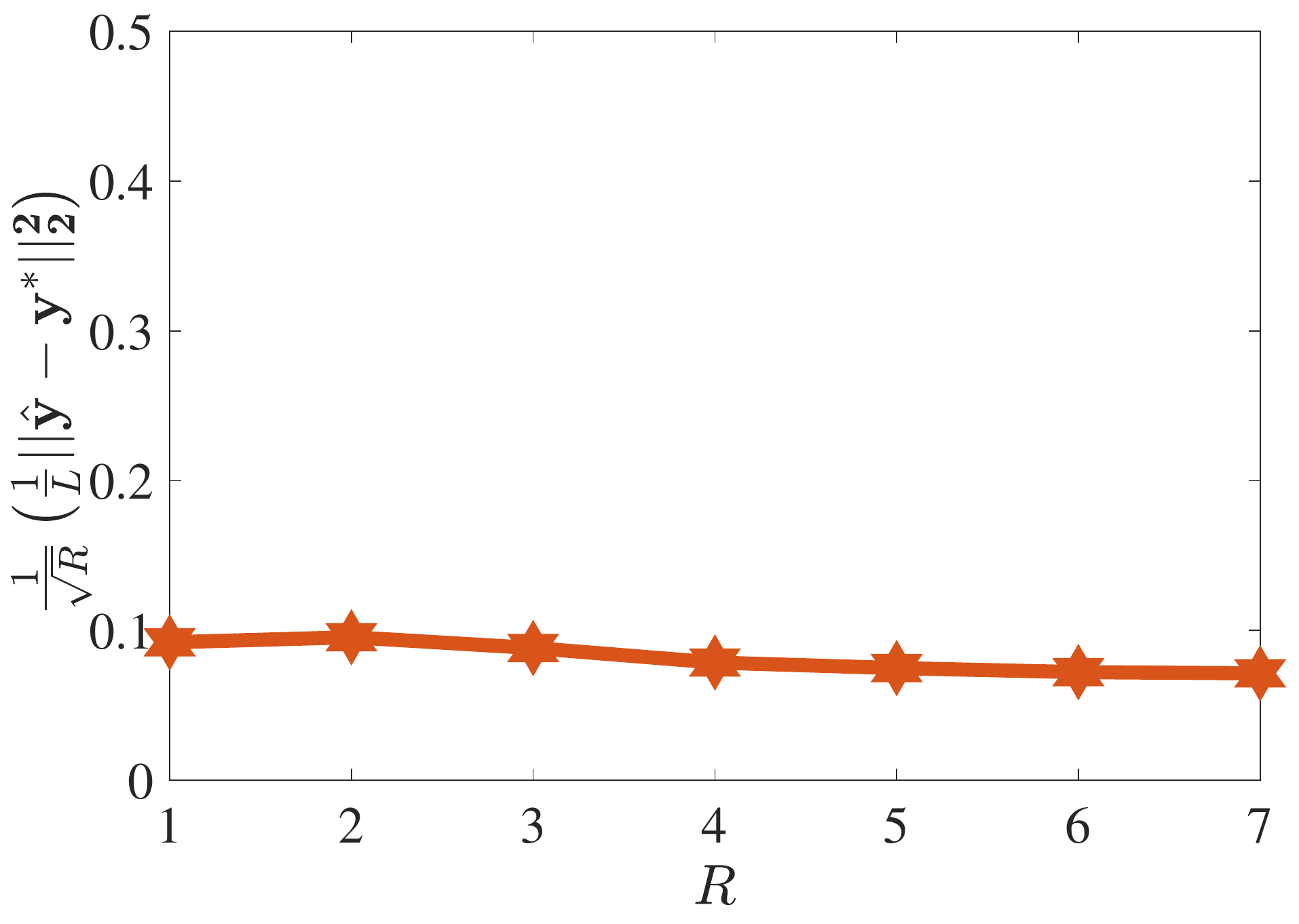}}
\caption{The relationship between $\frac{1}{\sqrt{R}}\left(\frac{1}{L}||\hat{\yv}-\yv^{*}||_{2}^{2}\right)$ and $R$}
\label{fig: normalized for R}
\end{figure}
 
Moreover, we demonstrate the relationship between the MSE of our framework and $K$. For that, we set $L=12, R=2$, and we generate $[\Dm]_{i,j}$ from a Rademacher distribution. The values of $c_{1}$ and $c_{2}$ are set to be fading, whereas the shifts are set to be $(0.1, 0.5), (0.5,0.9)$. Finally, we fix the noise variance at the level that gives an SNR = 15 dB, and we vary $K$ from $2$ to $6$. Fig~\ref{fig: normalized for K} shows that the MSE scales with $\mathcal{O}\left(\sqrt{K^{3}\log\left(K+1\right)}\right)$, as indicated in Theorem~\ref{th: main result}.


\begin{figure}[h!]
\centering
{\includegraphics[width=2.5in]{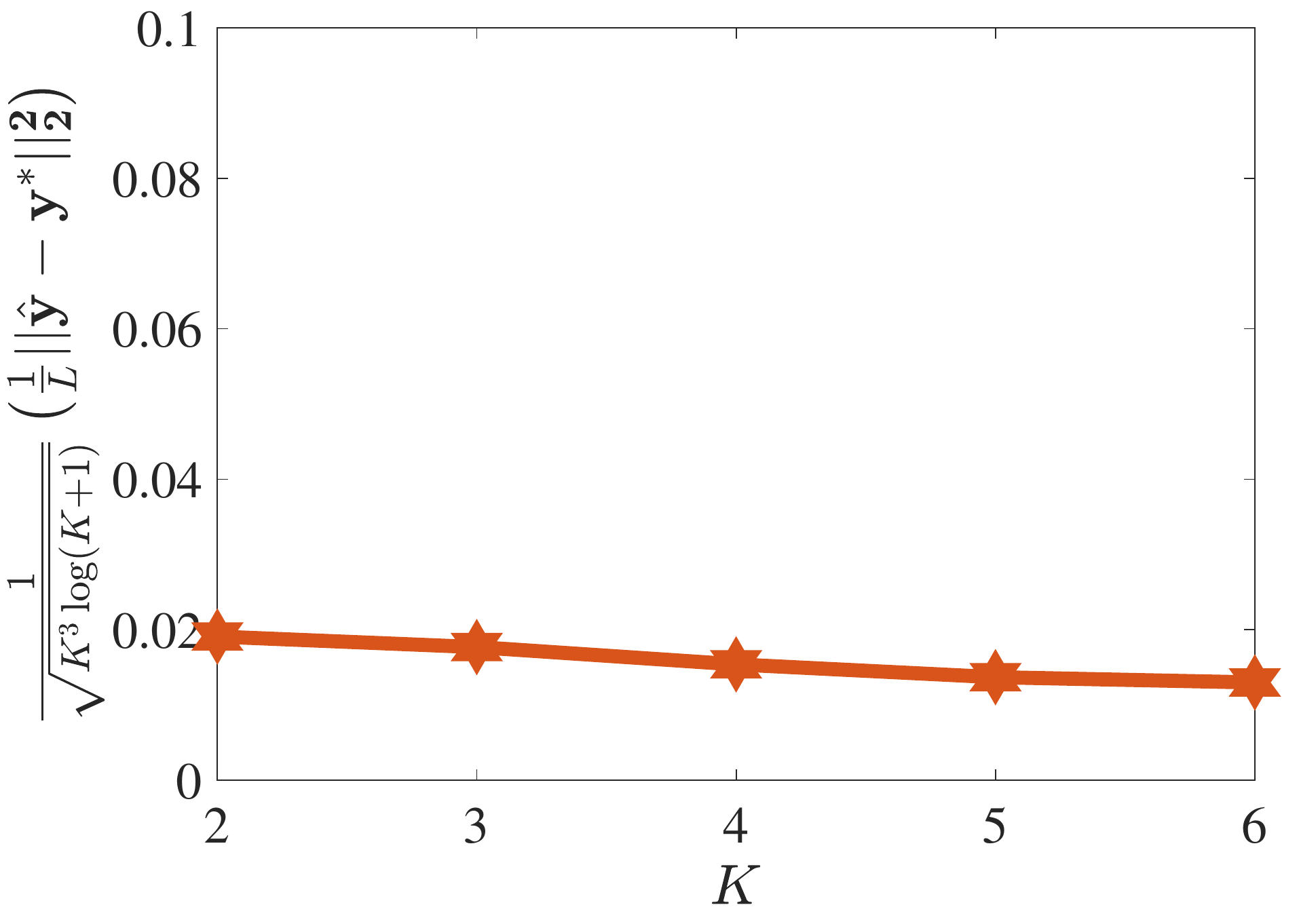}} 
\caption{The relationship between $\frac{1}{\sqrt{K^{3}\log\left(K+1\right)}}\left(\frac{1}{L}||\hat{\yv}-\yv^{*}||_{2}^{2}\right)$ and $K$.}
\label{fig: normalized for K}
\end{figure}

Next, we validate the proposed framework's performance in a more practical scheme by following the simulation setup provided in \cite{yang2016super}. For that, we generate $\Dm$ by extracting the principal components of a structured matrix; hence, making it less random. We let $K=3, R=1$, $\sigma_{\omegav}^{2}=0.15$, $c_{j} = j$, and we set the shifts to be $\left(0.13, 0.67\right)$. Moreover, we set $s_{j}\left(l\right)$ to be the samples of $s_{\sigma^2}\left(t\right) = \frac{1}{\sqrt{2\pi\sigma^2}} e^{-\frac{t^2}{2\sigma^{2}}}$ where $\sigma^{2} \in [0.1, 1]$, by taking $L$ samples from it around the origin with a sampling rate of $1/L$. To construct $\Dm$, we first build the dictionary $\Xim_{\sv}$ such that $\Xim_{\sv} = [\sv_{\sigma^{2}=0.1}  \ \sv_{\sigma^{2}=0.2} \dots \sv_{\sigma^{2}=1} ]$ where the columns of $\Xim_{\sv}$ are the samples of $s_{\sigma^2}\left(t\right)$ evaluated at the given values of $\sigma^{2}$. Following that, we apply PCA on $\Xim_{\sv}$ and obtain its best rank-$3$ approximation, i.e., $\Xim_{\sv, 3}$. Finally, we formulate $\Dm$ by setting its columns to be the left singular vectors of $\Xim_{\sv, 3}$. It should be noted that since $\Xim_{\sv}$ has a sharp singular values decay, $\Xim_{\sv, 3}$ provides a good approximation to it. For example, ${||\Xim_{\sv}-\Xim_{\sv,3}||_{F}}/{||\Xim_{\sv}||_{F}}= 6.244 \times 10^{-5}$ when $L=15$. On the other hand, we let $\hv$ be a deterministic vector by setting $\hv = \frac{\bar{\hv}}{||\bar{\hv}||_{2}}$ where $\bar{\hv} = [1+1i, -1+2i , -2-1i]^{T}$. Hence, not satisfying Assumption~\ref{as 4}.

In Fig~\ref{fig: D deterministic}, we plot the MSE of the proposed denoising framework versus $L$, and we compare it with that of $\yv$. The figure shows clearly that our framework provides a significantly lower MSE that decreases with $L$. On the other hand, Fig~\ref{fig: D deterministic 2} indicates that the MSE does scale with $\mathcal{O}\left(L^{-3/2}\left(\log\left(N\right)\right)^{3/2}\right)$ as in Theorem~\ref{th: main result}.

\begin{figure}[h!]
\centering
\subfigure[]{\label{fig: D deterministic}\includegraphics[width=2.5in]{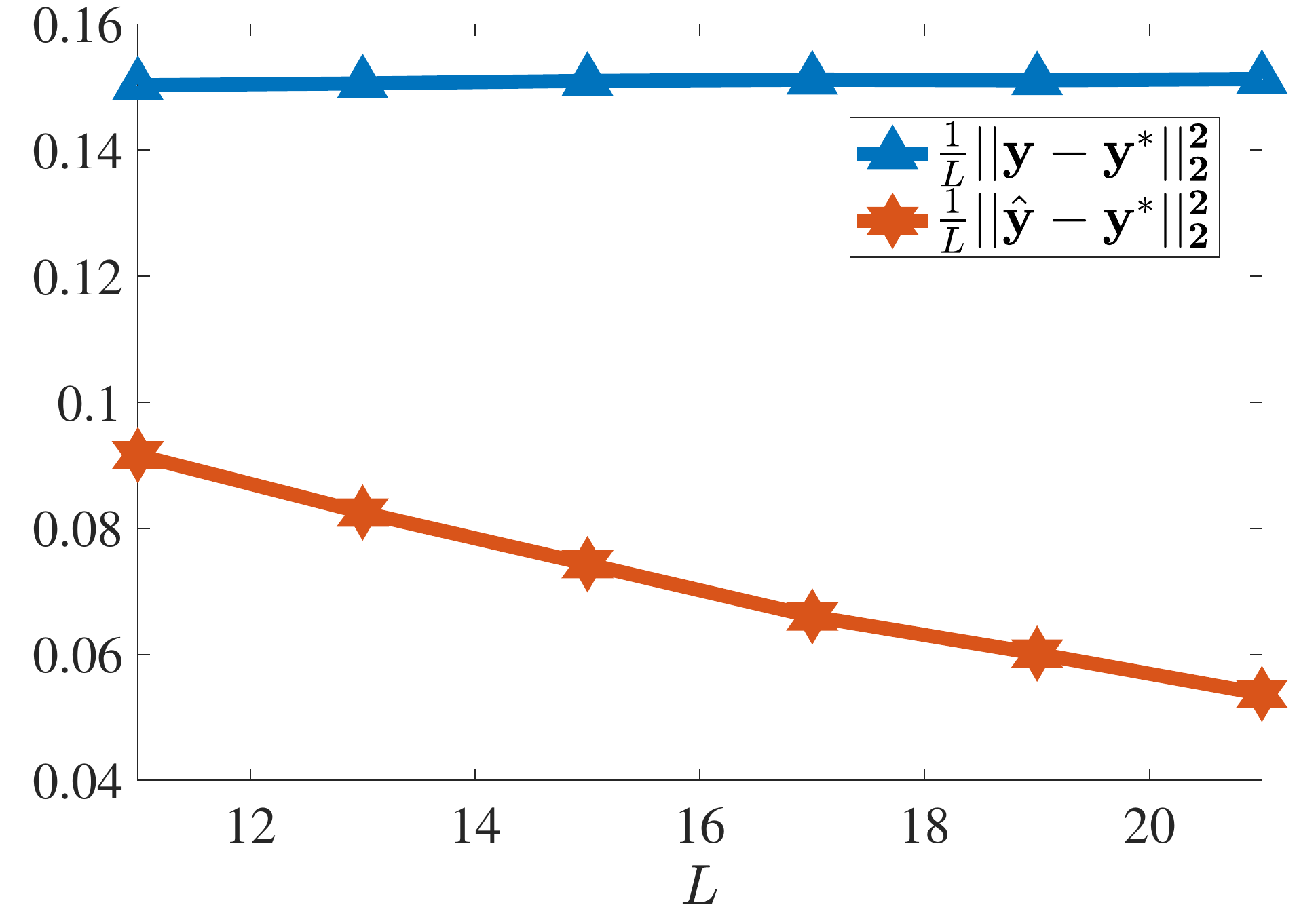}} 
\subfigure[]{\label{fig: D deterministic 2}\includegraphics[width=2.5in]{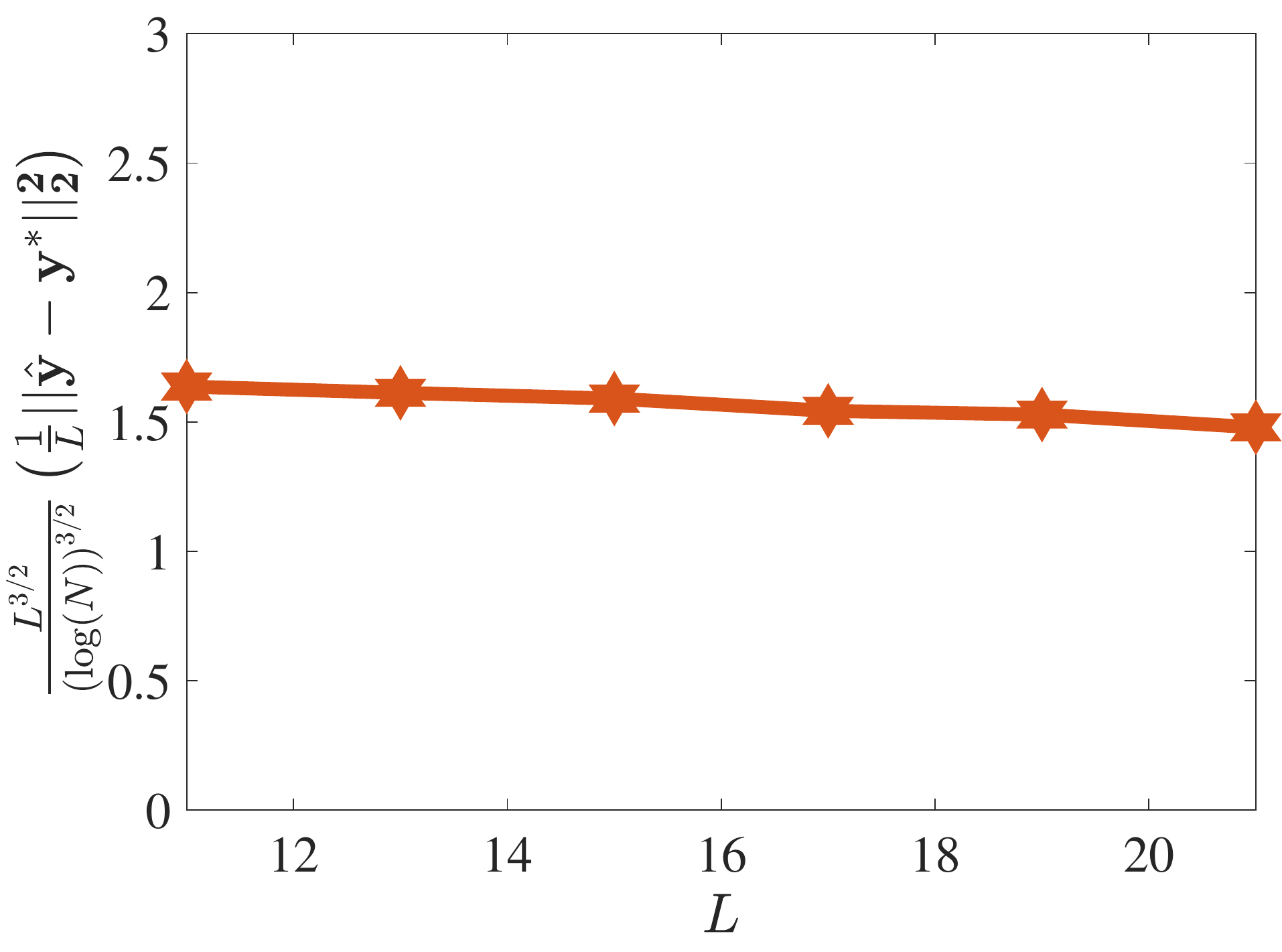}}
\caption{(a) The denoising performance of the proposed framework. (b) The relationship between $\frac{L^{3/2}}{\left(\log\left(N\right)\right)^{3/2}}\left(\frac{1}{L}||\hat{\yv}-\yv^{*}||_{2}^{2}\right)$ and $L$.}
\label{fig: D deterministic all}
\end{figure}

Finally, we provide a simulation experiment with all parameters being randomized. For this experiment, we set $K=2, R=1$, $\sigma_{\omegav}^{2}=0.15$ and we vary $L$ in the range $[15, 21]$. For each noise iteration at a given value of $L$, we generate a random matrix $\Dm$ as in the first experiment, a random vector $\hv_{j}$ that satisfies Assumption~\ref{as 3}, a random shifts that satisfy Assumption~\ref{as 4}, and a random $c_{j}$ as in the second experiment. The final MSE is obtained as an average over many noise iterations. Fig~\ref{fig: monti carlo} shows that the proposed framework provides a significantly lower MSE.   
\begin{figure}[h!]
\centering
\includegraphics[width=2.5in]{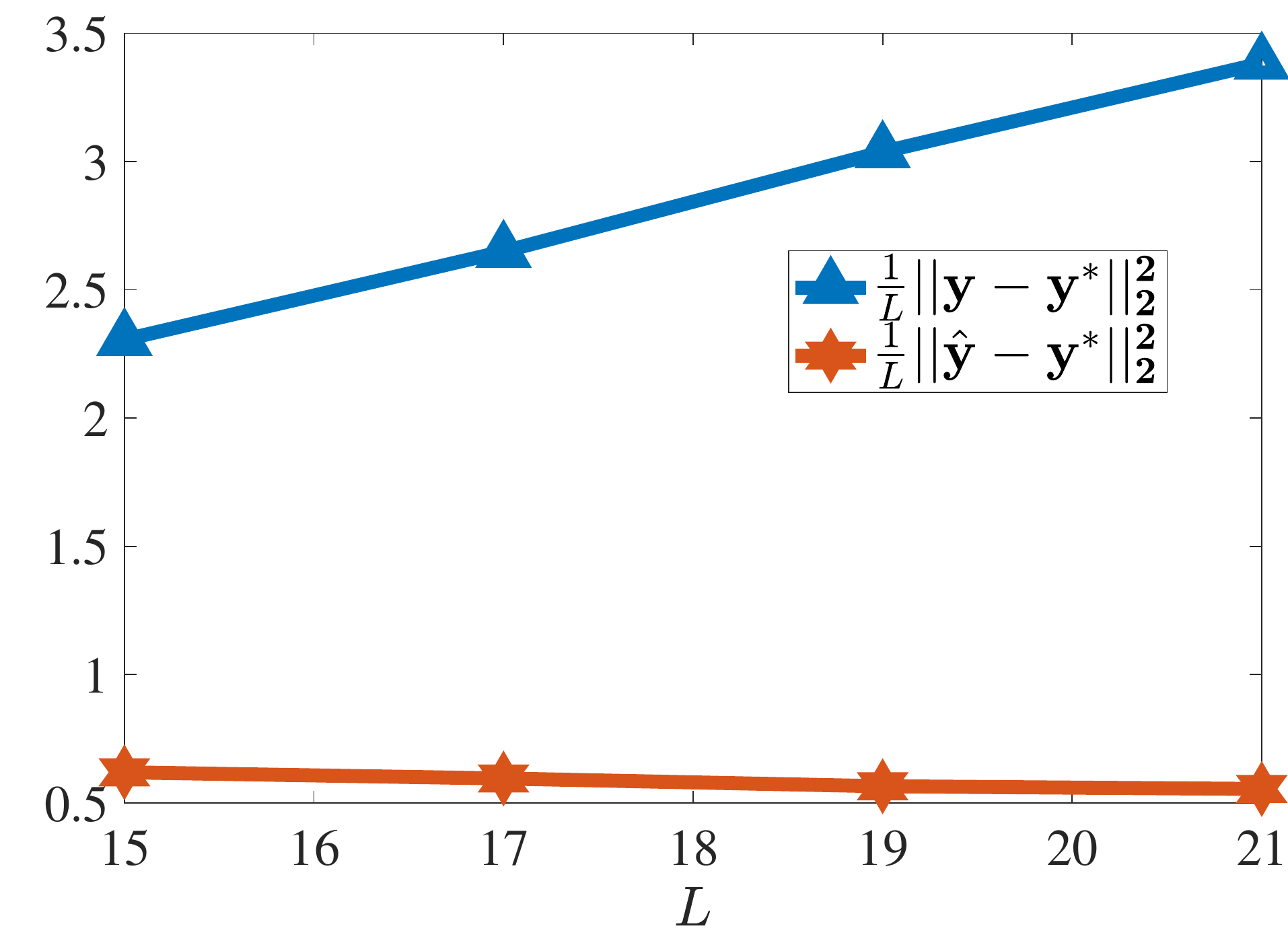}
\caption{The denoising performance of the proposed framework.}
\label{fig: monti carlo}
\end{figure}

On the other hand, we provide a single simulation experiment to demonstrate the parameters' recovery problem using Proposition~\ref{pro: main pro} with the detailed theoretical analysis of the estimation error being left to future work. Here, we set $L=19$, $K=2$, $R=1$, $\sigma_{\omegav}^{2}=0.32$, $(\tau_{1},f_{1})=(0.30,0.94)$, and we let $c_{1}$ and $[\Dm]_{i,j}$ to be as in the third experiment. To estimate the shifts, we discretize the domain $[0,1]^{2}$ with a step size of $10^{-3}$, and then we locate $\hat{\rv}_{j}$ where $||\fv\left(\hat{\rv}_{j}\right)||_{2}^{2}=1$. From Fig~\ref{fig:out3}, we can observe that $(\hat{\tau}_{1},\hat{f}_{1})=\left(0.299,0.939\right)$, which is very close to the true one. In Fig~\ref{fig:out4}, we compare the magnitude of the recovered samples with the true ones. The figure shows that the recovered samples are close to the original ones with a tenuous error. Finally, we find that $|\hv_{1}^{H}\hat{\hv}_{1}|=0.9713$ which confirms the superiority of the framework.

\begin{figure}[h!]
\centering
\subfigure[]{\label{fig:out3}\includegraphics[width=1.72in]{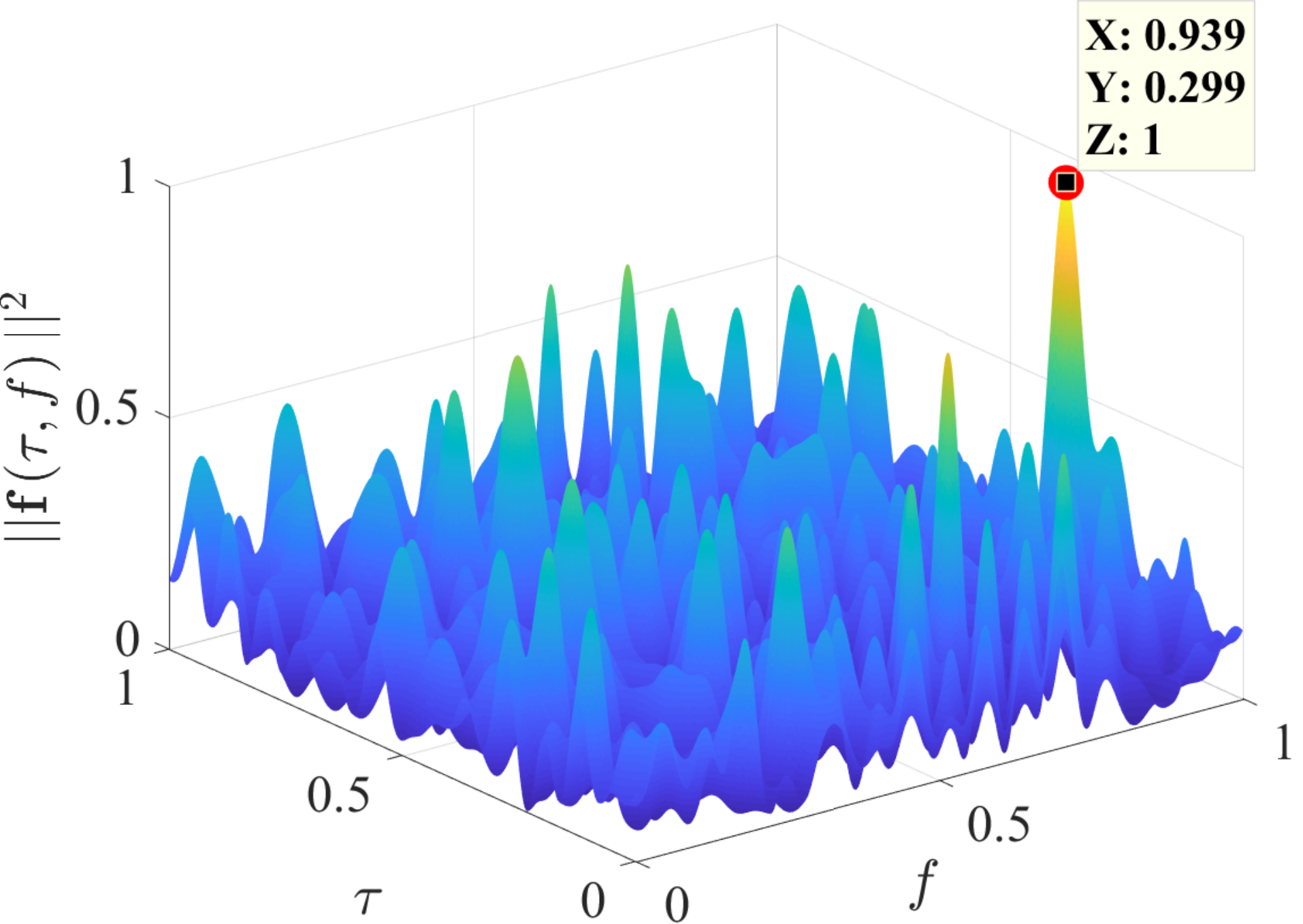}} 
\subfigure[]{\label{fig:out4}\includegraphics[width=1.72in]{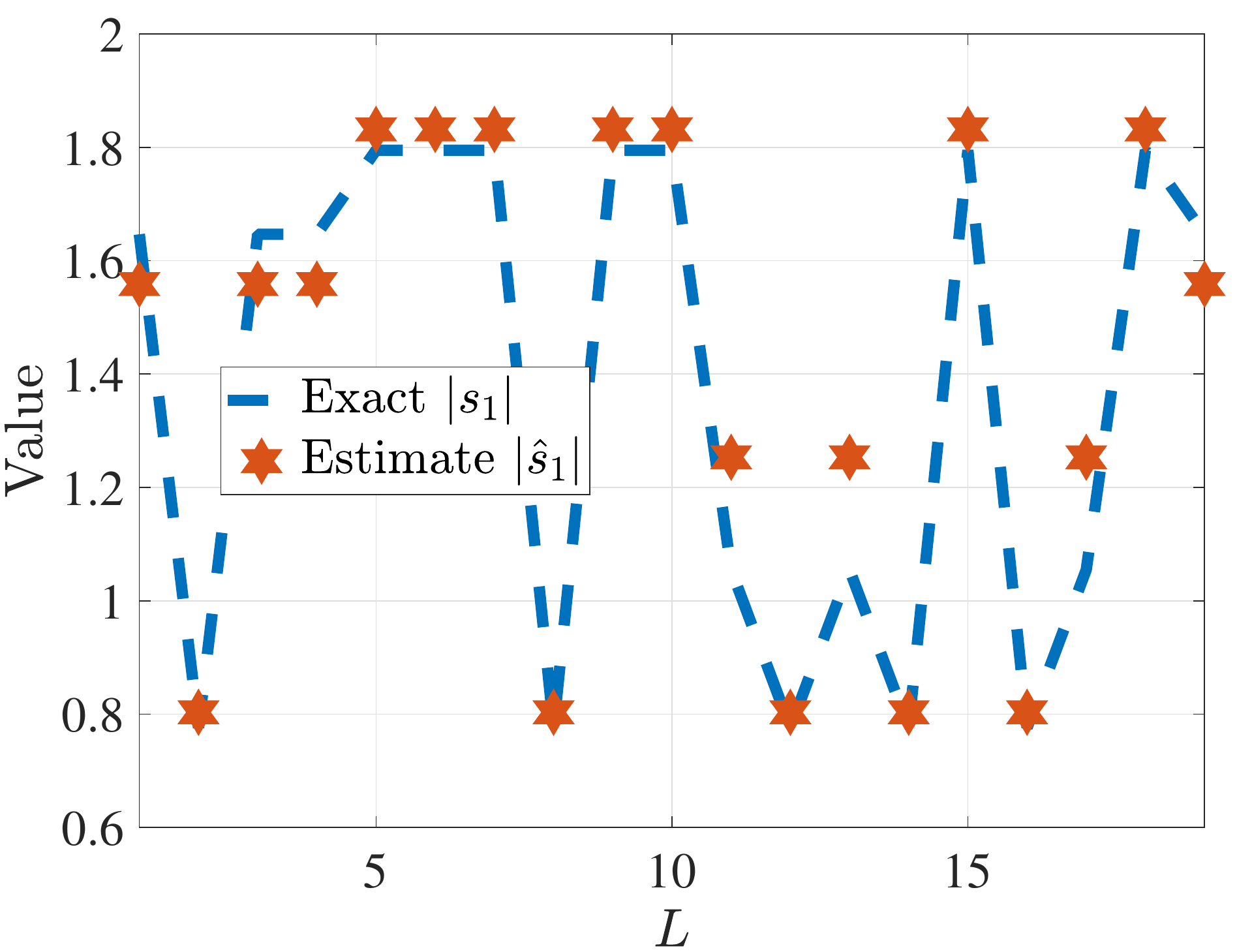}}
\caption{(a) The location of the estimated shift. (b) Comparing the estimated samples with the true ones.}
\label{fig: per 1 n}
\end{figure}

\section{Proof of Theorem~\ref{th: main result}}
\label{sec: proof on theorem}
In this section, we provide detailed proof for Theorem~\ref{th: main result}. We start by defining a vector function $\gv\left(\rv\right) \in \mathbb{C}^{K \times 1}$ such that
\begin{equation}
\label{eq: function g}
\gv\left(\rv\right) =  \sum_{j=1}^{R} c_{j} \hv_{j} \delta\left(\tau-\tau_{j}\right) \delta\left(f-f_{j}\right),
\end{equation}
where $\delta\left(0\right)=1, \delta\left(x\right)=0, \forall x \neq 0$. Based on (\ref{eq: function g}) we can express $\Um_{\text{o}}$ as
\begin{equation}
\label{eq: new U}
\Um_{\text{o}} = \iint \limits_{0}^{1}  \gv\left(\rv\right) \av\left(\rv\right)^{H} d\rv, 
\end{equation}
where $d\rv=d\tau df$. Now let us define the denoising error vector $\ev \in \mathbb{C}^{L \times 1}$ as
\begin{equation}
\label{eq: error vector}
\ev = \yv^{*}-\hat{\yv}= \mathcal{X}(\Um_{\text{o}})- \mathcal{X}(\widehat{\Um})
\end{equation}
and the difference measure vector $\nuv\left(\rv\right) \in \mathbb{C}^{K \times 1}$ as
\begin{equation}
\label{eq: diff vector}
\nuv\left(\rv\right) = \gv\left(\rv\right)- \hat{\gv}\left(\rv\right),
\end{equation}
where $\hat{\gv}\left(\rv\right)$ is the estimate of $\gv\left(\rv\right)$. Finally, consider the following two disjoint subsets
\begin{align}
&\Omega_{\text{close}}\left(j\right)=  \left\lbrace\rv: |\rv-\rv_{j}| \leq 0.2447/N \right\rbrace, \  j=1,\dots, R\nonumber\\
&\Omega_{\text{far}}= [0,1)^{2} / \Omega_{\text{close}} \nonumber
\end{align}
with $\Omega_{\text{close}}= \cup_{j=1}^{R}\Omega_{\text{close}}\left(j\right)$. Note that $\Omega_{\text{close}}$ contains the points in $[0,1]^{2}$ that are close to $\{\rv_{j}\}_{j=1}^{R}$ while $\Omega_{\text{far}}$ includes the points that are far away from it.

Next, the following lemma establishes an upper bound on the denoising error vector $\ev$ based on $\Omega_{\text{close}}\left(j\right)$ and $\Omega_{\text{far}}$.
\begin{lemma}
\label{lemma: error bound}
Define $\phiv\left(\rv\right) \in \mathbb{C}^{K \times 1}$ such that
\begin{equation}
\label{eq: func phi}
\phiv\left(\rv\right) = \mathcal{X}^{*}\left(\ev\right) \av\left(\rv\right).
\end{equation}
Then
\begin{align}
\label{eq: error bound tem}
&||\ev||_{2}^{2}  \leq \sup_{\rv\in[0,1]^{2}} ||\phiv\left(\rv\right)||_{2} \left[ \iint \limits_{\ \Omega_{\text{far}}} ||\nuv\left(\rv\right)||_{2}d\rv+ \sum _{k=0}^{3}T_{k}\right] 
\end{align}
where 
\begin{align}
&T_{0}^{j} = \left|\left|  \ \ \iint \limits_{\Omega_{\text{close}}\left(j\right)} \hspace{-5pt} \nuv\left(\rv\right)d\rv \right|\right|_{2}\hspace{-5pt},  T_{1}^{j} = 2\pi N \left|\left| \ \  \iint \limits_{\Omega_{\text{close}}\left(j\right)} \hspace{-5pt} \left(\tau-\tau_{j}\right) \nuv\left(\rv\right)d\rv\right|\right|_{2} \nonumber\\
&T_{2}^{j} =  \left(2\pi N\right) \left|\left| \ \  \iint \limits_{\Omega_{\text{close}}\left(j\right)} \left(f-f_{j}\right) \nuv\left(\rv\right)d\rv \right|\right|_{2}, \nonumber\\
&T_{3}^{j} = \frac{1}{2}\left(2\pi N\right)^{2} \  \iint \limits_{\Omega_{\text{close}}\left(j\right)} \left(|\tau-\tau_{j}|+|f-f_{j}|\right)^{2} \left|\left|\nuv\left(\rv\right)\right|\right|_{2} d\rv,  \nonumber
\end{align}
and $T_{k} = \sum_{j=1}^{R}T_{k}^{j}$ for $k=0,1,2,3$.
\end{lemma}

\begin{proof}
Based on (\ref{eq: new U}), (\ref{eq: error vector}), and (\ref{eq: diff vector}) we can write
\begin{align}
&||\ev||_{2}^{2} \leq |\left\langle\ev,\ev\right\rangle| = \left| \left\langle \ev,\mathcal{X}\left( \iint \limits_{0}^{1} \nuv\left(\rv\right) \av\left(\rv\right)^{H} d\rv \right) \right\rangle\right|= \nonumber\\
&\left| \iint \limits_{0}^{1} \left\langle \mathcal{X}^{*}\left( \ev\right), \nuv\left(\rv\right) \av\left(\rv\right)^{H} \right\rangle  d\rv\right| =\left|  \iint \limits_{0}^{1} \nuv\left(\rv\right)^{H}\phiv\left(\rv\right) d\rv \right|. \nonumber
\end{align}
By using the triangular inequality and the definitions of $\Omega_{\text{close}}$ and $\Omega_{\text{far}}$ we can bound the above term as
\begin{align}
\label{eq: error proof2}
||\ev||_{2}^{2} &\leq \left|  \iint \limits_{\Omega_{\text{far}}} \nuv\left(\rv\right)^{H}\phiv\left(\rv\right) d\rv \right|  + \left| \iint \limits_{ \ \ \Omega_{\text{close}}} \nuv\left(\rv\right)^{H}\phiv\left(\rv\right) d\rv \right|
\end{align} 
The first term in (\ref{eq: error proof2}) can be bounded as
\begin{align}
\label{eq: error proof term1}
 &\left| \iint \limits_{\Omega_{\text{far}}} \nuv\left(\rv\right)^{H}\phiv\left(\rv\right) d\rv \right|  \leq  \ \ \iint \limits_{\Omega_{\text{far}}} \left|\nuv\left(\rv\right)^{H}\phiv\left(\rv\right) \right| d\rv \nonumber\\
 & \leq  \sup_{\rv\in[0,1]^{2}} ||\phiv\left(\rv\right)||_{2} \iint \limits_{\Omega_{\text{far}}} ||\nuv\left(\rv\right)||_{2} \  d\rv,
\end{align}
where the first inequality is from triangular inequality, while the second one is from Cauchy-Schwarz inequality. To upper bound the second term in (\ref{eq: error proof2}), we first define $\bm{\varphi}\left(\rv\right)$ as
\begin{equation}
\label{eq: var fn definition}
\bm{\varphi}\left(\rv\right) = \bm{\phi}\left(\rv\right)-\bm{\phi}\left(\rv_{j}\right)-\left(\tau-\tau_{j}\right) \bigtriangledown_{\tau}\bm{\phi}\left(\rv_{j}\right)- \left(f-f_{j}\right) \bigtriangledown_{f}\bm{\phi}\left(\rv_{j}\right) 
\end{equation}
for all $\rv_{j} \in \mathcal{R}$. Based on (\ref{eq: var fn definition}), we can write
\begin{align}
\label{eq: error proof 3}
&\left| \ \iint \limits_{\Omega_{\text{close}}} \nuv\left(\rv\right)^{H}\phiv\left(\rv\right) d\rv \right| \leq \sum_{j=1}^{R} \left|  \ \ \iint \limits_{\Omega_{\text{close}}\left(j\right)} \nuv\left(\rv\right)^{H} \bm{\varphi}\left(\rv\right)d\rv\right| \nonumber\\
&+  \sum_{j=1}^{R} \left|  \ \ \iint \limits_{\Omega_{\text{close}}\left(j\right)}  \nuv\left(\rv\right)^{H} \phiv\left(\rv_{j}\right)d\rv \right|\nonumber\\
&+ \sum_{j=1}^{R} \left|  \ \ \iint \limits_{\Omega_{\text{close}}\left(j\right)} \left(\tau-\tau_{j}\right)  \nuv\left(\rv\right)^{H} \bigtriangledown_{\tau}\bm{\phi}\left(\rv_{j}\right)d\rv \right| \nonumber\\
&+  \sum_{j=1}^{R} \left|  \ \ \iint \limits_{\Omega_{\text{close}}\left(j\right)} \left(f-f_{j}\right) \nuv\left(\rv\right)^{H}  \bigtriangledown_{f}\bm{\phi}\left(\rv_{j}\right) d\rv\right|.
\end{align}
Before proceeding forward, it is useful first to introduce different results to facilitate our proof. First, define
\begin{equation}
\label{eq: rho function}
\rho\left(\rv\right) = \xv^{H} \phiv\left(\rv\right)
\end{equation}
with $||\xv||_{2} = 1$. Then, based on (\ref{eq: var fn definition}) and (\ref{eq: rho function}) we can write
\begin{align}
\label{eq: bound q}
&||\bm{\varphi}\left(\rv\right)||_{2}= \sup_{\xv:||\xv||_{2}=1} \left| \left\langle\xv,\bm{\varphi}\left(\rv\right)\right\rangle \right| =  \sup_{\xv}\left| \rho\left(\rv\right)-\rho\left(\rv_{j}\right)-\right.\nonumber\\
&\left.\left(\tau-\tau_{j}\right) \bigtriangledown_{\tau}\rho\left(\rv_{j}\right)- \left(f-f_{j}\right) \bigtriangledown_{f}\rho\left(\rv_{j}\right) \right|.
\end{align}
On the other hand, the two-dimensional Tyler series expansion for $\rho\left(\rv\right)$ around $\rv_{j}$ can be written as
\begin{align}
\label{eq: rho tyler final}
\rho\left(\rv\right) = &\rho\left(\rv_{j}\right)+ \left(\tau-\tau_{j}\right) \bigtriangledown_{\tau} \rho \left(\rv_{j}\right) + \left(f-f_{j}\right) \bigtriangledown_{f} \rho \left(\rv_{j}\right) \nonumber\\
&+\frac{1}{2} \left(\tau-\tau_{j}\right)^{2} \bigtriangledown_{\tau}^{2} \rho \left(\rv_{j}\right)+ \frac{1}{2} \left(f-f_{j}\right)^{2} \bigtriangledown_{f}^{2} \rho \left(\rv_{j}\right) \nonumber\\
&+ \left(\tau-\tau_{j}\right) \left(f-f_{j}\right) \bigtriangledown_{\tau f} \rho \left(\rv_{j}\right).
\end{align} 
Using (\ref{eq: rho tyler final}) we can bound (\ref{eq: bound q}) as
\begin{align}
\label{eq: bound q 2}
&||\bm{\varphi}\left(\rv\right)||_{2} \leq  \frac{1}{2} \sup_{\xv}\left| \left(\tau-\tau_{j}\right)^{2} \bigtriangledown_{\tau}^{2} \rho \left(\rv_{j}\right) \right| \nonumber\\
&+ \frac{1}{2} \sup_{\xv}\left| \left(f-f_{j}\right)^{2} \bigtriangledown_{f}^{2} \rho \left(\rv_{j}\right) \right|\nonumber\\
&+\sup_{\xv}\left|  \left(\tau-\tau_{j}\right) \left(f-f_{j}\right) \bigtriangledown_{\tau f} \rho \left(\rv_{j}\right) \right|.
\end{align}
However, note that
\begin{align}
&\sup_{\xv}\left| \left(\tau-\tau_{j}\right)^{2} \bigtriangledown_{\tau}^{2} \rho \left(\rv_{j}\right) \right| \leq \left(\tau-\tau_{j}\right)^{2} \sup_{\xv, \rv\in \Omega_{\text{close}}\left(j\right)}\left| \bigtriangledown_{\tau}^{2} \rho \left(\rv\right) \right| \nonumber\\
& \leq \left(2\pi N\right)^{2} \left(\tau-\tau_{j}\right)^{2} \sup_{\xv,\rv\in \Omega_{\text{close}}\left(j\right)}\left| \rho \left(\rv\right) \right| \nonumber\\
&=\left(2\pi N\right)^{2} \left(\tau-\tau_{j}\right)^{2} \sup_{\rv\in \Omega_{\text{close}}\left(j\right)}\left|\left| \phiv\left(\rv\right) \right|\right|_{2} ,  
\end{align}
where the second inequality is based on Bernstein's polynomial inequality \cite[Theorem 1.6.2]{tropp2015introduction}. Applying the same inequality to the other terms in (\ref{eq: bound q 2}), and after some algebraic manipulations, we can show that
\begin{align}
\label{eq: bound q final}
&||\bm{\varphi}\left(\rv\right)||_{2} \leq \frac{1}{2} \left(2\pi N\right)^{2} \left(|\tau-\tau_{j}|+|f-f_{j}|\right)^{2}\sup_{\rv\in [0,1]^{2}}\left|\left| \phiv\left(\rv\right) \right|\right|_{2} 
\end{align}
Upon following the same steps, we can show that
\begin{align}
\label{eq: two bounds 1}
&||\bigtriangledown_{\tau}\phiv\left(\rv_{j}\right)||_{2} \leq \left(2\pi N \right) \sup_{\rv\in [0,1]^{2}}\left|\left| \phiv\left(\rv\right) \right|\right|_{2} \\
\label{eq: two bounds 2}
&||\bigtriangledown_{f}\phiv\left(\rv_{j}\right)||_{2} \leq \left(2\pi N \right) \sup_{\rv\in [0,1]^{2}}\left|\left| \phiv\left(\rv\right) \right|\right|_{2}.  
\end{align}
By applying triangular inequality and Cauchy-Schwarz inequality for each term in (\ref{eq: error proof 3}), along with the bounds in (\ref{eq: bound q final}), (\ref{eq: two bounds 1}), and (\ref{eq: two bounds 2}), we can show that
\begin{align}
\label{eq: error bound final}
&\left| \ \iint \limits_{\Omega_{\text{close}}}  \nuv\left(\rv\right)^{H}\phiv\left(\rv\right) d\rv \right| \leq \sup_{\rv\in [0,1]^{2}}\left|\left| \phiv\left(\rv\right) \right|\right|_{2} \left[T_{0}+T_{1}+T_{2}+T_{3}\right]
\end{align}
The proof of Lemma~\ref{lemma: error bound} is concluded by substituting (\ref{eq: error proof term1}) and (\ref{eq: error bound final}) in (\ref{eq: error proof2}).
\end{proof}

To proceed with the proof of Theorem~\ref{th: main result}, we now need to bound each term in (\ref{eq: error bound tem}). For that, we will first introduce some useful results in Lemma~\ref{lemma: bound on phi} and Theorems~\ref{th: main dual cert}, \ref{th: main dual cert 1}, and \ref{th: main dual cert 2}. 
\begin{lemma}
\label{lemma: bound on phi}
Recall the definition of $\phiv\left(\rv\right)$ in (\ref{eq: func phi}) and set the regularization parameter $\mu$ to be 
\begin{equation}
\label{eq: regula}
\mu = 6 \lambda \sigma_{\omegav} ||\Dm||_{F} \sqrt{\log\left(N\right)},
\end{equation}
where $\lambda\geq 1$. Then, with probability at least $1- \frac{CK}{\sqrt{N^{3}}}$ we have
\begin{equation}
\sup_{\rv\in [0,1]^{2}}\left|\left| \phiv\left(\rv\right) \right|\right|_{2} \leq 2 \mu, 
\end{equation}
where $C$ is a numerical constant.
\end{lemma}
The proof of Lemma~\ref{lemma: bound on phi} is provided in Appendix~\ref{app: phi bound proof}.
\begin{theorem}
\label{th: main dual cert}
Let $\{\rv_{j}\}_{j=1}^{R}$ satisfy the separation in (\ref{eq: seperation condition}) and take any vector $\hv_{j} \in \mathbb{C}^{K \times 1}$ with $||\hv_{j}||_{2}=1$. Then, there exists a trigonometric vector polynomial $\fv\left(\rv\right)=\mathcal{X}^{*}\left(\qv\right) \av\left(\rv\right) \in \mathbb{C}^{K \times 1}$, where $\qv \in \mathbb{C}^{K \times 1}$ is the solution to (\ref{eq: dual reg problem}), such that when $L \geq \bar{C}_{1} R K \widetilde{K}^{4} \log^{2}\left(\frac{  \tilde{C}_{1} R^{2} K^{2} L^{3} }{\delta}\right)\log^{2}\left(\frac{ \tilde{C}_{1} (K+1)  L^{3} }{\delta}\right)$ with $\delta >0$ we have 
\begin{equation}
\label{eq: hold assump 1}
\fv\left(\rv_{j}\right)= \text{sign}\left(c_{j}\right) \hv_{j}, \ \ \forall \rv_{j} \in \mathcal{R}
\end{equation}
\begin{equation}
\label{eq: hold assump 3}
||\fv\left(\rv\right)||_{2} \leq  1-C_{2},\ \ \forall \rv \in \Omega_{\text{far}}.
\end{equation}
Moreover, in the region $\Omega_{\text{close}}\left(j\right)$ we have
\begin{equation}
\label{eq: hold assump 2}
||\fv\left(\rv\right)||_{2} \leq  1-C_{2}^{*}N^{2} \left(\left(\tau-\tau_{j}\right)-\left(f-f_{j}\right)\right)^{2},
\end{equation}
\begin{align}
\label{eq: hold assump 4}
||\fv\left(\rv\right)-\text{sign}\left(c_{j}\right)\hv_{j}||_{2} \leq &\bar{C}_{2} N^{2}\left(|\tau-\tau_{j}|+|f-f_{j}|\right)^{2}.
\end{align}
\end{theorem}

\begin{theorem}
\label{th: main dual cert 1}
Let $\{\rv_{j}\}_{j=1}^{R}$ satisfy the minimum separation in (\ref{eq: seperation condition}) and consider the set $\left\lbrace\hv \in \mathbb{C}^{K \times 1}: ||\hv||_{2}=1\right\rbrace$. Then, there exists a vector polynomial $\fv_{1}\left(\rv\right)=\mathcal{X}^{*}\left(\qv_{1}\right) \av\left(\rv\right) \in \mathbb{C}^{K \times 1}$ for some $\qv_{1} \in \mathbb{C}^{K \times 1}$ with
\begin{equation}
\label{eq: conds for f1}
\fv_{1}\left(\rv_{j}\right)=\fv_{1}^{\left(0,1\right)}\left(\rv_{j}\right)= \bm{0}_{K}, \fv_{1}^{\left(1,0\right)}\left(\rv_{j}\right)= \text{sign}\left(c_{j}\right) \hv_{j}, \forall \rv_{j} \in \mathcal{R}
\end{equation}
where $\fv_{1}^{(m,n)}\left(\rv_{j}\right) = \frac{\partial^{m}}{\partial \tau^{m}} \frac{\partial^{n}}{\partial f^{n}}\fv_{1}\left(\rv\right)\bigr\rvert_{\rv = \rv_{j}}$, that satisfies
\begin{equation}
\label{eq: hold assump 3 for 1}
||\fv_{1}\left(\rv\right)||_{2} \leq  \frac{C_{3}}{N},\ \ \forall \rv \in \Omega_{\text{far}},
\end{equation}  
\begin{align}
\label{eq: hold assump 4 for 1}
||\fv_{1}\left(\rv\right)-\hv_{j}\left(\tau-\tau_{j}\right)||_{2} &\leq  \bar{C}_{3} N \left(|\tau-\tau_{j}|+|f-f_{j}|\right)^{2}, \nonumber\\
&\forall \rv \in \Omega_{\text{close}}\left(j\right).
\end{align}
\end{theorem}
\begin{theorem}
\label{th: main dual cert 2}
Let $\{\rv_{j}\}_{j=1}^{R}$ satisfy the minimum separation in (\ref{eq: seperation condition}) and take any vector $\hv_{j} \in \mathbb{C}^{K \times 1}$ with $||\hv_{j}||_{2}=1$. Then, there exists a trigonometric vector polynomial $\fv_{2}\left(\rv\right)=\mathcal{X}^{*}\left(\qv_{2}\right) \av\left(\rv\right) \in \mathbb{C}^{K \times 1}$ for some $\qv_{2} \in \mathbb{C}^{K \times 1}$ with 
\begin{equation}
\label{eq: conds for f2}
\fv_{2}\left(\rv_{j}\right)=  \fv_{2}^{\left(1,0\right)}\left(\rv_{j}\right)=\bm{0}_{K}, \fv_{2}^{\left(0,1\right)}\left(\rv_{j}\right)= \text{sign}\left(c_{j}\right) \hv_{j}, \forall \rv_{j} \in \mathcal{R}
\end{equation}
that satisfies
\begin{equation}
\label{eq: hold assump 3 for 2}
||\fv_{2}\left(\rv\right)||_{2} \leq \frac{C_{4}}{N},\ \ \forall \rv \in \Omega_{\text{far}},
\end{equation}
\begin{align}
\label{eq: hold assump 4 for 2}
||\fv_{2}\left(\rv\right)-\hv_{j}\left(f-f_{j}\right)||_{2} &l\leq \bar{C}_{4} N \left(|\tau-\tau_{j}|+|f-f_{j}|\right)^{2}, \nonumber\\
&\forall \rv \in \Omega_{\text{close}}\left(j\right).
\end{align}
\end{theorem}
The proofs of Theorems~\ref{th: main dual cert}, \ref{th: main dual cert 1}, and \ref{th: main dual cert 2} are the subject of Appendix~\ref{app: Theorems proof}. 

Now we are ready to bound each term in (\ref{eq: error bound tem}).
\begin{lemma}
\label{lemma: bounds for Ts}
Let $C_{5}, \bar{C}_{5}$ be two constants with different values for different terms. Then, the following bounds hold
\begin{align}
& T_{0} \leq   \left| \ \iint \limits_{0}^{1} \nuv\left(\rv\right)^{H}\fv\left(\rv\right) d\rv \right|+ \iint \limits_{\Omega_{\text{far}}} ||\nuv\left(\rv\right)||_{2}d\rv +C_{5} T_{3} \nonumber \\
&T_{1} \leq   (2\pi N) \hspace{-2pt}\left|  \iint \limits_{0}^{1}\hspace{-3pt} \nuv\left(\rv\right)^{H}\fv_{1}\left(\rv\right) d\rv \right|\hspace{-3pt}+ \bar{C}_{5}\hspace{-3pt}\iint \limits_{\Omega_{\text{far}}}\hspace{-3pt} ||\nuv\left(\rv\right)||_{2}d\rv \hspace{-2pt}+C_{5} T_{3} \nonumber \\
&T_{2} \leq   (2\pi N) \hspace{-2pt}\left|  \iint \limits_{0}^{1}\hspace{-3pt} \nuv\left(\rv\right)^{H}\fv_{2}\left(\rv\right) d\rv \right|\hspace{-3pt}+ \bar{C}_{5}\hspace{-3pt}\iint \limits_{\Omega_{\text{far}}}\hspace{-3pt} ||\nuv\left(\rv\right)||_{2}d\rv \hspace{-2pt}+C_{5} T_{3}. \nonumber 
\end{align}
\end{lemma}
The proof of Lemma~\ref{lemma: bounds for Ts} is based on Theorems~\ref{th: main dual cert}, \ref{th: main dual cert 1}, and \ref{th: main dual cert 2} and is appended in Appendix~\ref{app: Ts bounds}.

Upon applying Lemmas~\ref{lemma: bound on phi} and \ref{lemma: bounds for Ts} we can rewrite (\ref{eq: error bound tem}) as 
\begin{align}
\label{eq: error simplified}
&||\ev||_{2}^{2} \leq 2\mu \times \nonumber\\ 
& \hspace{-5pt}\left[\left| \iint \limits_{0}^{1} \nuv\left(\rv\right)^{H}\fv\left(\rv\right) d\rv \right|+2\pi N \left| \iint \limits_{0}^{1} \nuv\left(\rv\right)^{H}\fv_{1}\left(\rv\right) d\rv \right|+2\pi N \times \right.\nonumber\\
& \left. \left| \ \iint \limits_{0}^{1} \nuv\left(\rv\right)^{H}\fv_{2}\left(\rv\right) d\rv \right|\hspace{-2pt}+\bar{C}_{6}\iint \limits_{\Omega_{\text{far}}} \hspace{-2pt}||\nuv\left(\rv\right)||_{2}d\rv + C_{6} T_{3} \right]
\end{align}
In Lemmas~\ref{lemma: bound for 0 1 term}, \ref{lemma: bound for 0 1 term 2}, \ref{lemma: bound for 0 1 term 3}, and \ref{lemma: bound for 0 1 term 4} below, we provide an upper bound for each term in (\ref{eq: error simplified}) separately.
\begin{lemma}
\label{lemma: bound for 0 1 term}
For a constant $C_{9}$ and $\delta >0$, the bound
\begin{align}
\label{eq: bound for int 0 to 1}
&\left| \iint \limits_{0}^{1} \nuv\left(\rv\right)^{H}\fv\left(\rv\right) d\rv \right| \nonumber\\
&\leq \sup_{\rv \in [0,1]^{2}} ||\phiv\left(\rv\right)||_{2} \left( C_{9} \sqrt{\frac{R}{L^{3} K} \log \left( \frac{2\left(K+1\right)}{\delta}\right)}\right) 
\end{align}
occurs with probability at least $1- \delta$ provided that $L \geq C_{9} R K \log \left(\frac{2(K+1)}{\delta}\right)$.
\end{lemma}
The proof of Lemma~\ref{lemma: bound for 0 1 term} is deferred to Appendix~\ref{app: proof for bound 0 1}.
\begin{lemma}
\label{lemma: bound for 0 1 term 2}
For a constant $\bar{C}_{9}$ and $\delta >0$, the bound
\begin{align}
&(2\pi N) \left| \ \iint \limits_{0}^{1} \nuv\left(\rv\right)^{H}\fv_{1}\left(\rv\right) d\rv \right| \nonumber\\
& \leq \sup_{\rv \in [0,1]^{2}} ||\phiv\left(\rv\right)||_{2} \left(\bar{C}_{9} \sqrt{\frac{R}{L^{3} K} \log \left( \frac{2\left(K+1\right)}{\delta}\right)}\right) \nonumber
\end{align}
occurs with probability at least $1- \delta$ provided that $L \geq \bar{C}_{9}R K \log \left(\frac{2(K+1)}{\delta}\right)$.
\end{lemma}
The proof of Lemma~\ref{lemma: bound for 0 1 term 2} is presented in Appendix~\ref{app: proof for bound 0 1 term 2}.

\begin{lemma}
\label{lemma: bound for 0 1 term 3}
There exists a constant $\tilde{C}_{9}$ and $\delta >0$ such that
\begin{align}
&(2\pi N) \left| \ \iint \limits_{0}^{1} \nuv\left(\rv\right)^{H}\fv_{2}\left(\rv\right) d\rv \right| \nonumber\\
&\leq \sup_{\rv \in [0,1]^{2}} ||\phiv\left(\rv\right)||_{2} \left( \tilde{C}_{9} \sqrt{\frac{R}{L^{3} K} \log \left( \frac{2\left(K+1\right)}{\delta}\right)}\right) \nonumber
\end{align}
occurs with probability at least $1- \delta$ provided that $L \geq \tilde{C}_{9} R K \log \left(\frac{2(K+1)}{\delta}\right)$.
\end{lemma}
The proof of Lemma~\ref{lemma: bound for 0 1 term 3}, not detailed here, follows the same steps as that of Lemma~\ref{lemma: bound for 0 1 term 2}, with $\fv_{2}\left(\rv\right)$ in (\ref{eq: final dual for 2}) being used instead of $\fv_{1}\left(\rv\right)$.

\begin{lemma}
\label{lemma: bound for 0 1 term 4}
It occurs with very high probability that
\begin{align}
\label{eq: bound on the two terms}
& \iint \limits_{\Omega_{\text{far}}} ||\nuv\left(\rv\right)||_{2}  d\rv + T_{3} \leq  \nonumber\\
& \sup_{\rv \in [0,1]^{2}} ||\phiv\left(\rv\right)||_{2}\left(C_{10}\sqrt{\frac{R}{L^{3} K} \log \left( \frac{2\left(K+1\right)}{\delta}\right)}\right).  
\end{align}
\end{lemma}
\begin{proof}
To start with, refer to the restriction of the difference measure vector $\nuv$ on the support set $\mathcal{R} = \{\rv_{j}\}_{j=1}^{R}$ by $\mathcal{P}_{\mathcal{R}}\left(\nuv\right)$. Given the fact that the dual polynomial vector $\fv\left(\rv\right)$ interpolates the atoms $\av\left(\rv\right)$ on $\mathcal{R}$, we can write
\begin{align}
\label{eq: TV function 1}
&||\mathcal{P}_{\mathcal{R}}\left(\nuv\right)||_{\text{TV}} =   \iint \limits_{0}^{1} \mathcal{P}_{\mathcal{R}}\left(\nuv^{H}\right)\fv\left(\rv\right) d\rv \nonumber\\
& \leq \left| \iint \limits_{0}^{1} \nuv\left(\rv\right)^{H}\fv\left(\rv\right) d\rv \right| + \left| \ \iint \limits_{\mathcal{R}^{c}} \nuv\left(\rv\right)^{H}\fv\left(\rv\right) d\rv \right|,
\end{align}
where TV stands for the total variation norm while $\mathcal{R}^{c}$ refers to the complement set of $\mathcal{R}$ in $[0,1]^{2}$. The first term in (\ref{eq: TV function 1}) is bounded by Lemma~\ref{lemma: bound for 0 1 term} while for the second term we can write
\begin{align}
\label{eq: TV function 2}
&\left| \iint \limits_{\mathcal{R}^{c}} \nuv\left(\rv\right)^{H}\fv\left(\rv\right) d\rv \right| \leq \sum_{j \in \mathcal{R}} \left| \ \ \iint \limits_{\Omega_{\text{close}}\left(j\right)  \backslash \{\rv_{j}\}} \nuv\left(\rv\right)^{H}\fv\left(\rv\right) d\rv \right| \nonumber\\
& + \left| \ \iint \limits_{\Omega_{\text{far}}} \nuv\left(\rv\right)^{H}\fv\left(\rv\right) d\rv \right|.
\end{align}
By using H\"{o}lder's inequality, the bound in (\ref{eq: hold assump 2}), as well as the definition of $T_{3}^{j}$, we can bound the first term in (\ref{eq: TV function 2}) by
\begin{align}
\label{eq: TV function 3}
&\left| \ \ \iint \limits_{\Omega_{\text{close}}\left(j\right)  \backslash \{\rv_{j}\}} \nuv\left(\rv\right)^{H}\fv\left(\rv\right) d\rv \right| \leq   \iint \limits_{\Omega_{\text{close}}\left(j\right)  \backslash \{\rv_{j}\}} ||\nuv\left(\rv\right)||_{2} ||\fv\left(\rv\right)||_{2} d\rv \nonumber\\
& \leq \  \iint \limits_{\Omega_{\text{close}}\left(j\right)  \backslash \{\rv_{j}\}} ||\nuv\left(\rv\right)||_{2}  d\rv - C_{11} T_{3}^{j},
\end{align}
for a constant $C_{11}$. On the other hand, the second term in (\ref{eq: TV function 2}) can be bounded using H\"{o}lder's inequality and (\ref{eq: hold assump 3}) as
\begin{align}
\label{eq: TV function 4}
&\left| \iint \limits_{\Omega_{\text{far}}} \nuv\left(\rv\right)^{H}\fv\left(\rv\right) d\rv \right| \leq  \iint \limits_{\Omega_{\text{far}}} ||\nuv\left(\rv\right)||_{2} ||\fv\left(\rv\right)||_{2} d\rv \nonumber\\
& \leq (1-C_{2})  \iint \limits_{\Omega_{\text{far}}} ||\nuv\left(\rv\right)||_{2}  d\rv.
\end{align}
Substituting (\ref{eq: TV function 3}) and (\ref{eq: TV function 4}) in (\ref{eq: TV function 2}), then substituting the result back into (\ref{eq: TV function 1}), along with Lemma~\ref{lemma: bound for 0 1 term}, we obtain 
\begin{align}
\label{eq: TV final bound}
&||\mathcal{P}_{\mathcal{R}}\left(\nuv\right)||_{\text{TV}}  \leq \nonumber\\
&\sup_{\rv \in [0,1]^{2}} ||\phiv\left(\rv\right)||_{2} \left( C_{9} \sqrt{\frac{R}{L^{3} K} \log \left( \frac{2\left(K+1\right)}{\delta}\right)}\right) - C_{11} T_{3} \nonumber\\
&+ \sum_{j \in \mathcal{R}} \ \ \iint \limits_{\Omega_{\text{close}}\left(j\right)  \backslash \{\rv_{j}\}} ||\nuv\left(\rv\right)||_{2}  d\rv + (1-C_{2})  \iint \limits_{\Omega_{\text{far}}} ||\nuv\left(\rv\right)||_{2}  d\rv
\end{align}
Now given that 
\begin{align}
&||\mathcal{P}_{\mathcal{R}^{c}}\left(\nuv\right)||_{\text{TV}} = \sum_{j \in \mathcal{R}} \ \ \iint \limits_{\Omega_{\text{close}}\left(j\right)  \backslash \{\rv_{j}\}} \hspace{-5pt}||\nuv\left(\rv\right)||_{2}  d\rv + \hspace{-1pt} \iint \limits_{\Omega_{\text{far}}} \hspace{-2pt}||\nuv\left(\rv\right)||_{2}  d\rv \nonumber
\end{align}
we can write based on (\ref{eq: TV final bound})
\begin{align}
\label{eq: TV final bound 2}
&||\mathcal{P}_{\mathcal{R}^{c}}\left(\nuv\right)||_{\text{TV}} - ||\mathcal{P}_{\mathcal{R}}\left(\nuv\right)||_{\text{TV}} \geq  C_{11} T_{3}+ C_{2} \hspace{-2pt}\iint \limits_{\Omega_{\text{far}}} \hspace{-2pt}||\nuv\left(\rv\right)||_{2}  d\rv\nonumber\\
& -\sup_{\rv \in [0,1]^{2}} ||\phiv\left(\rv\right)||_{2} \left( C_{9} \sqrt{\frac{R}{L^{3} K} \log \left( \frac{2\left(K+1\right)}{\delta}\right)}\right).
\end{align}
To proceed, we will use the following proposition, whose proof is provided in Appendix~\ref{app: proof of U hat atomic}.
\begin{proposition}
\label{pro: upper bound on the diff}
There exists a constant $\bar{C}_{11}$ such that
\begin{align}
\label{equ: final p pc bound second one}
&||\mathcal{P}_{\mathcal{R}^{c}}\left(\nuv\right)||_{\text{TV}} -||\mathcal{P}_{\mathcal{R}}\left(\nuv\right)||_{\text{TV}} \leq  ( \bar{C}_{11} \frac{\mu}{\lambda}+\frac{1}{\lambda})\iint \limits_{\Omega_{\text{far}}} ||\nuv\left(\rv\right)||_{2}d\rv \nonumber\\
&+\bar{C}_{11}  \frac{\mu}{\lambda} \left( T_{3} + \sup_{\rv \in [0,1]^{2}} ||\phiv\left(\rv\right)||_{2} \sqrt{\frac{R}{L^{3} K} \log \left( \frac{2\left(K+1\right)}{\delta}\right)}\right)
\end{align}
\end{proposition}
Based on Proposition~\ref{pro: upper bound on the diff} and (\ref{eq: TV final bound 2}) we can obtain 
\begin{align}
& ( C_{2}-\bar{C}_{11} \frac{\mu}{\lambda}-\frac{1}{\lambda})\iint \limits_{\Omega_{\text{far}}} ||\nuv\left(\rv\right)||_{2}d\rv + (C_{11} -\bar{C}_{11}  \frac{\mu}{\lambda} )T_{3} \leq  \nonumber\\
& (\bar{C}_{11}  \frac{\mu}{\lambda}+C_{9})\hspace{-3pt}\sup_{\rv \in [0,1]^{2}}\hspace{-2pt} ||\phiv\left(\rv\right)||_{2}\sqrt{\hspace{-3pt}\frac{R}{L^{3} K} \log \left( \frac{2\left(K+1\right)}{\delta}\right)} \nonumber
\end{align}
which leads to (\ref{eq: bound on the two terms}) for a large enough $\lambda$.
\end{proof}
Starting from (\ref{eq: error simplified}), and upon applying Lemmas~\ref{lemma: bound on phi} to \ref{lemma: bound for 0 1 term 4}, we can conclude after some manipulations that
\begin{align}
&||\ev||_{2}^{2} \leq  \bar{C}_{12} \mu^{2}\sqrt{\frac{R}{L^{3} K} \log \left( \frac{2\left(K+1\right)}{\delta}\right)}. \nonumber
\end{align}
Substituting for $\mu$ from (\ref{eq: regula}) we obtain
\begin{align}
&||\ev||_{2}^{2} \leq  C_{12} \lambda^{2} \sigma_{\omegav}^{2} ||\Dm||_{F}^{2} \log\left(N\right) \frac{\sqrt{R}}{L^{3/2} \sqrt{K}} \sqrt{\log \left( \frac{2\left(K+1\right)}{\delta}\right)} \nonumber\\
& \leq  C_{12} \lambda^{2} \sigma_{\omegav}^{2} K^{3/2} \sqrt{\frac{R}{L}} \log\left(N\right) \sqrt{\log \left( \frac{2\left(K+1\right)}{\delta}\right)}, \nonumber
\end{align}
where the last inequality from $||\Dm||_{F}^{2}  \leq K ||\Dm ||_{2}^{2}  \leq L K^{2}$. The proof of Theorem~\ref{th: main result} is concluded by setting $\delta = \frac{2}{CN}$.

\section{Conclusions}
\label{sec:conclusion}
In this work, we developed a new mathematical framework for atomic norm denoising in blind 2D super-resolution. The framework recovers a signal that consists of a superposition of time-delayed and frequency-shifted unknown waveforms from its noisy measurements upon solving a regularized LS atomic norm minimization problem. Moreover, we provided an approach to estimate the unknown parameters that characterize the signal. We analyzed the MSE performance of the denoising framework, and we showed that it increases with the noise variance, the number of unknown shifts, and the dimension of the low-dimensional subspace. Furthermore, we showed that the MSE decreases with the number of observed samples. Simulation results that verify the dependency of the MSE on these parameters are provided.  Extensions to this work include: 1) Deriving theoretical bounds on the localization error to investigate how well the estimator localizes the unknown parameters. 2) Studying the framework performance under different noise distributions. 3) Reducing the computational complexity of solving the problem by studying alternative optimization techniques or considering solving the problem using a subset from the observed samples.


%



\appendices
\numberwithin{equation}{section}
\numberwithin{lemma}{section}
\numberwithin{proposition}{section}
\numberwithin{theorem}{section}

\section{Proof of the Dual Problem (\ref{eq: dual reg problem})}
\label{app: dual proof}
Starting from (\ref{eq: reg problem}), we can express the problem as a constrained optimization problem in the form
\begin{align}
\label{eq: a1}
\underset{{\Um} \in \mathbb{C}^{K \times L^2}}{\text{minimize}} \ \ &\frac{1}{2} \left|\left|\yv -\zv\right|\right|_{2}^{2} + \mu \  ||{\Um}||_{\mathcal{A}} \nonumber\\
 &\text{subject to} : \zv =\mathcal{X}({\Um}).
\end{align}
The Lagrangian function of (\ref{eq: a1}) can be written as 
\begin{equation}
L\left({\Um},\zv,\qv\right) =  \frac{1}{2} \left|\left|\yv -\zv\right|\right|_{2}^{2} + \mu \  ||{\Um}||_{\mathcal{A}}+ \left\langle \qv, \zv-\mathcal{X}({\Um})\right\rangle, \nonumber
\end{equation}
where $\qv \in \mathbb{C}^{L \times 1}$ is the Lagrangian variable. Based on that, the dual function can be written as
\begin{align}
\label{eq: a3}
&g\left(\qv\right) = \inf_{{\Um}, \zv} L\left({\Um},\zv,\qv\right) \nonumber\\
&= \inf_{\zv} \left[ \frac{1}{2} \left|\left|\yv -\zv\right|\right|_{2}^{2} + \left\langle\qv,\zv\right\rangle_{\mathbb{R}} \right]+ \inf_{{\Um}} \left[\mu ||{\Um}||_{\mathcal{A}} - \left\langle \qv, \mathcal{X}({\Um})\right\rangle_{\mathbb{R}} \right] \nonumber\\
&=  \left\langle \yv,\qv \right\rangle_{\mathbb{R}}- \frac{1}{2} ||\qv||_{2}^{2}+ \inf_{{\Um}} \left[\mu ||{\Um}||_{\mathcal{A}} - \left\langle \mathcal{X}^{*}(\qv), {\Um}\right\rangle_{\mathbb{R}} \right]
\end{align}
By using the dual atomic norm definition in (\ref{eq: dual atomic for}), we obtain
\begin{equation}
\label{eq: a4}
 \inf_{{\Um}} \left[\mu ||{\Um}||_{\mathcal{A}} - \left\langle \mathcal{X}^{*}(\qv), {\Um}\right\rangle_{\mathbb{R}} \right]=\left\{
                \begin{array}{ll}
                  0 &\text{if} \ \ ||\mathcal{X}^{*}\left(\qv\right)||_{\mathcal{A}}^{*} \leq \mu\\
                  \infty &\text{otherwise}\\
                \end{array}
              \right.
\end{equation}
Substituting (\ref{eq: a4}) in (\ref{eq: a3}), then taking the maximum of $g\left(\qv\right)$, we obtain (\ref{eq: dual reg problem}).

\section{Proof of Lemma~\ref{lemma: suff cond lemma}}
\label{app: suff cond}
The cost function in (\ref{eq: reg problem}) is minimized at $\widehat{\Um}$ if $\forall t >0, \forall \Um \in \mathbb{C}^{K\times L^{2}}$ we have
\begin{align}
&\frac{1}{2} ||\yv -\mathcal{X}(\widehat{\Um}+t(\Um-\widehat{\Um}))||_{2}^{2} + \mu \  ||\widehat{\Um}+t(\Um-\widehat{\Um})||_{\mathcal{A}} \geq \nonumber\\
&\frac{1}{2} ||\yv -\mathcal{X}(\widehat{\Um})||_{2}^{2} + \mu \  ||\widehat{\Um}||_{\mathcal{A}}. \nonumber
\end{align}
The above expression can be written as
\begin{align}
\label{eq: opti express 2}
&2\mu \left[ ||\widehat{\Um}+t(\Um-\widehat{\Um})||_{\mathcal{A}}-||\widehat{\Um}||_{\mathcal{A}}\right] \geq  \nonumber\\
& ||\yv -\mathcal{X}(\widehat{\Um})||_{2}^{2}-||\yv -\mathcal{X}(\widehat{\Um}+t(\Um-\widehat{\Um}))||_{2}^{2}.
\end{align}
On the other hand, and by using the convexity property of the atomic norm, we can conclude after some manipulations that
\begin{equation}
\label{eq: opt con atomic norm}
||\widehat{\Um}+t(\Um-\widehat{\Um})||_{\mathcal{A}}-||\widehat{\Um}||_{\mathcal{A}} \leq t \left(||\Um||_{\mathcal{A}}-||\widehat{\Um}||_{\mathcal{A}} \right).
\end{equation}
Substituting (\ref{eq: opt con atomic norm}) in (\ref{eq: opti express 2}) and then manipulating we obtain
\begin{align}
\label{eq: opti express 3}
&2\mu t \left[ ||\Um||_{\mathcal{A}}-||\widehat{\Um}||_{\mathcal{A}} \right]  \nonumber\\
& \geq ||\yv -\mathcal{X}(\widehat{\Um})||_{2}^{2}-||\yv -\mathcal{X}(\widehat{\Um})-t\mathcal{X}(\Um-\widehat{\Um})||_{2}^{2} \nonumber\\
& \geq 2t\left\langle\yv-\mathcal{X}(\widehat{\Um}),\mathcal{X}(\Um-\widehat{\Um})\right\rangle-t^{2}||\mathcal{X}(\Um-\widehat{\Um})||_{2}^{2}.
\end{align}
Since (\ref{eq: opti express 3}) holds for all $\Um$ and $t \in \left(0,1\right)$, we can divide (\ref{eq: opti express 3}) by $2t$, then, we set $t=0$ in order to obtain
\begin{align}
\label{eq: opti express 4}
&\mu \left[ ||\Um||_{\mathcal{A}}-||\widehat{\Um}||_{\mathcal{A}} \right]  \geq \left\langle\yv-\mathcal{X}(\widehat{\Um}),\mathcal{X}(\Um-\widehat{\Um})\right\rangle.
\end{align}
Based on (\ref{eq: opti express 4}) we can write
\begin{align}
&\mu ||\widehat{\Um}||_{\mathcal{A}}-\left\langle\yv-\mathcal{X}(\widehat{\Um}),\mathcal{X}(\widehat{\Um})\right\rangle  \leq \nonumber\\
&\mu ||\Um||_{\mathcal{A}} -\left\langle\yv-\mathcal{X}(\widehat{\Um}),\mathcal{X}(\Um)\right\rangle, \nonumber
\end{align}
which can be written as 
\begin{align}
\label{eq: opti express 6}
&\mu ||\widehat{\Um}||_{\mathcal{A}}-\left\langle\mathcal{X}^{*}\left(\yv-\mathcal{X}(\widehat{\Um})\right),\widehat{\Um}\right\rangle  \leq \nonumber\\
&\inf_{\Um} \left[\mu ||\Um||_{\mathcal{A}} -\left\langle\mathcal{X}^{*}\left(\yv-\mathcal{X}(\widehat{\Um}\right),\Um\right\rangle\right].
\end{align}
Given the dual atomic norm definition, we can write
\begin{equation}
\label{eq: opti express 7}
 \inf_{{\Um}} \left[||{\Um}||_{\mathcal{A}} - \left\langle \Cm, {\Um}\right\rangle_{\mathbb{R}} \right]=\left\{
                \begin{array}{ll}
                  0 &\text{if} \ \ ||\Cm||_{\mathcal{A}}^{*} \leq 1\\
                  \infty &\text{otherwise}\\
                \end{array}
              \right.
\end{equation}
Based on (\ref{eq: opti express 6}) and (\ref{eq: opti express 7}) we obtain (\ref{eq: suff cond 1}) and (\ref{eq: suff cond 2}).

\section{Proof of Lemma~\ref{lemma: bound on phi}}
\label{app: phi bound proof}
Based on (\ref{eq: error vector}) and (\ref{eq: func phi}) we can write
\begin{align}
&\sup_{\rv\in [0,1]^{2}}\left|\left| \phiv\left(\rv\right) \right|\right|_{2} =\sup_{\rv\in [0,1]^{2}}\left|\left|\mathcal{X}^{*}\left(\mathcal{X}(\Um_{\text{o}})-\mathcal{X}(\widehat{\Um})\right) \av\left(\rv\right)\right|\right|_{2} \nonumber\\
&= \sup_{\rv\in [0,1]^{2}}\left|\left|\mathcal{X}^{*}\left(\yv-\omegav-\mathcal{X}(\widehat{\Um})\right) \av\left(\rv\right)\right|\right|_{2}  \leq \nonumber\\
& \sup_{\rv\in [0,1]^{2}}\left|\left|\mathcal{X}^{*}\left(\yv-\mathcal{X}(\widehat{\Um})\right) \av\left(\rv\right)\right|\right|_{2} + \sup_{\rv\in [0,1]^{2}}\left|\left|\mathcal{X}^{*}\left(\omegav\right) \av\left(\rv\right)\right|\right|_{2} \nonumber\\
\label{eq: phi bou 1}
& = \left|\left|\mathcal{X}^{*}\left(\yv-\mathcal{X}(\widehat{\Um})\right) \right|\right|_{\mathcal{A}}^{*}+\left|\left|\mathcal{X}^{*}\left(\omegav\right) \right|\right|_{\mathcal{A}}^{*} \\
\label{eq: phi bou 2}
& \leq 2 \mu,
\end{align}
where (\ref{eq: phi bou 1}) is based on the definition of the dual atomic norm while (\ref{eq: phi bou 2}) is based on Lemma~\ref{lemma: suff cond lemma} and Lemma~\ref{lemma: sub phi bound}.

\begin{lemma}
\label{lemma: sub phi bound}
Let $\mu$ be as in (\ref{eq: regula}) with $\lambda \geq 1$, then
\begin{equation}
\text{Pr}\left[\left|\left|\mathcal{X}^{*}\left(\omegav\right) \right|\right|_{\mathcal{A}}^{*} \leq \frac{\mu}{\lambda} \right]  \geq 1- \frac{CK}{\sqrt{N^{3}}}. \nonumber
\end{equation}
\end{lemma}

\begin{proof}
It is shown in \cite[Appendix~C]{suliman2018blind} that
\begin{align}
\mathcal{X}^{*}\left(\omegav\right)  \av\left(\rv\right)= \sum_{p=-N}^{N} [\omegav]_{p}\frac{1}{L} \sum_{l,k=-N}^{N} e^{\frac{i2\pi k \left(p-l\right)}{L} } e^{-i2\pi \left(k \tau+p f \right)}\dv_{l} \nonumber
\end{align}
As a result, we can obtain after some manipulations
\begin{align}
\label{eq: lemma proof 3}
&\left(||\mathcal{X}^{*}\left(\omegav\right) ||_{\mathcal{A}}^{*}\right)^{2}=\sup_{\rv \in [0,1]^{2}} \sum_{n=1}^{K} \frac{1}{L^{2}} \sum_{p,p',k,k',l,l'=-N}^{N}[\omegav]_{p} [\omegav^{H}]_{p'} \times \nonumber\\
&[\dv_{l}]_{n}[\dv_{l'}^{H}]_{n} e^{i \frac{2\pi}{L}\left(k\left(p-l\right)-k'\left(p'-l'\right)\right)} e^{-i2\pi\tau\left(k-k'\right)} e^{-i2\pi f\left(p-p'\right)} \nonumber\\
&=: \sup_{\rv \in [0,1]^{2}} Q\left(\rv\right).
\end{align}
Let $\rv_{1}$ and $\rv_{2}$ be any two vectors in $[0,1]^{2}$, then
\begin{align}
\label{eq: lemma proof 4}
&Q\left(\rv_{1}\right) - Q\left(\rv_{2}\right)= \nonumber\\
&Q\left(\tau_{1},f_{1}\right)-Q\left(\tau_{1},f_{2}\right)+Q\left(\tau_{1},f_{2}\right)-Q\left(\tau_{2},f_{2}\right) \nonumber\\
&\leq \left|Q\left(\tau_{1},f_{1}\right)-Q\left(\tau_{1},f_{2}\right)\right|+\left|Q\left(\tau_{1},f_{2}\right)-Q\left(\tau_{2},f_{2}\right) \right| \nonumber\\
& \leq \left|f_{1}-f_{2}\right|\sup_{x_{1}\in [0,1]}\left|Q^{(0,1)}\left(\tau_{1},x_{1}\right)\right|\nonumber\\
&+\left|\tau_{1}-\tau_{2}\right|\sup_{x_{2}\in [0,1]}\left|Q^{(1,0)}\left(x_{2},f_{2}\right)\right|.
\end{align}
Applying Bernstein's polynomial inequality to (\ref{eq: lemma proof 4}) we obtain
\begin{align}
\label{eq: lemma proof 5}
&Q\left(\rv_{1}\right) - Q\left(\rv_{2}\right) \leq \left(2\pi N\right) \left|f_{1}-f_{2}\right|  \sup_{x_{1}\in [0,1]}\left|Q\left(\tau_{1},x_{1}\right)\right|  \nonumber\\
&+\left(2\pi N\right) \left|\tau_{1}-\tau_{2}\right|\sup_{x_{2}\in [0,1]}\left|Q\left(x_{2},f_{2}\right)\right| \nonumber\\
& \leq \left(4 \pi N\right) ||\rv_{1} -\rv_{2}||_{\infty} \sup_{\rv \in [0,1]^{2}} Q\left(\rv\right).
\end{align}
Let $\rv_{2}$ to acquire any value in a set of grid points defined by $l_{1}$ and $l_{2}$ such that $l_{1},l_{2} \in \mathbb{Z}$ with $\tau_{2}=0, \frac{1}{M}, \dots, \frac{l_{1}}{M}, \dots, \frac{M-1}{M}$ and $f_{2}=0, \frac{1}{M}, \dots, \frac{l_{2}}{M}, \dots, \frac{M-1}{M}$; $M \geq L$. Then based on (\ref{eq: lemma proof 5}) we can write
\begin{align}
&\sup_{\rv \in [0,1]^{2}} Q\left(\rv\right) \leq \max_{l_{1},l_{2}=0,\dots,M-1} \sup_{\rv \in [0,1]^{2}}\left[ Q\left(\frac{l_{1}}{M},\frac{l_{2}}{M}\right) \right. \nonumber\\
&\left.+\left(4 \pi N\right) \max \left(|\tau-\frac{l_{1}}{M}|,|f-\frac{l_{2}}{M}|\right) \sup_{\rv \in [0,1]^{2}} Q\left(\rv\right) \right] \nonumber
\end{align}
which leads to 
\begin{align}
\label{eq: lemma proof 6 temp}
\sup_{\rv \in [0,1]^{2}} Q\left(\rv\right) \leq \left(1-\frac{4 \pi N}{M}\right)^{-1} \max_{l_{1},l_{2}=0.\dots, M-1} Q\left(\frac{l_{1}}{M},\frac{l_{2}}{M}\right).
\end{align}
Substituting (\ref{eq: lemma proof 6 temp}) in (\ref{eq: lemma proof 3}) we get 
\begin{align}
\left(||\mathcal{X}^{*}\left(\omegav\right) ||_{\mathcal{A}}^{*}\right)^{2} \leq \left(1-\frac{4 \pi N}{M}\right)^{-1} \hspace{-7pt}\max_{l_{1},l_{2}=0.\dots, M-1} Q\left(\frac{l_{1}}{M},\frac{l_{2}}{M}\right). \nonumber
\end{align}
Note that we can also write based on (\ref{eq: lemma proof 3})
\begin{align}
&Q\left(\frac{l_{1}}{M},\frac{l_{2}}{M}\right)= \nonumber\\
&\sum_{n=1}^{K} \left|\frac{1}{L} \sum_{p,k=-N}^{N} [\omegav]_{p} \sum_{l=-N}^{N} e^{\frac{i2\pi k \left(p-l\right)}{L} } e^{\frac{-i2\pi}{M} \left(l_{1} k+l_{2}p \right)}[\dv_{l}]_{n} \right|^{2} \nonumber\\
&=: \sum_{n=1}^{K} \left|\widetilde{Q}_{n}\left(\frac{l_{1}}{M},\frac{l_{2}}{M}\right)\right|^{2}. \nonumber 
\end{align}
Given the assumptions on $\omegav$ and $\dv_{l}$ we can deduce that
\begin{align}
&\mathbb{E}\left[\widetilde{Q}_{n}\left(\frac{l_{1}}{M},\frac{l_{2}}{M}\right)\right] = 0 \nonumber
\end{align}
and
\begin{align}
&\text{var}\left(\widetilde{Q}_{n}\left(\frac{l_{1}}{M},\frac{l_{2}}{M}\right)\right) = \frac{1}{L^2} \text{var}\left(\omegav\right) \sum_{p,k=-N}^{N} \sum_{l=-N}^{N} \left|[\dv_{l}]_{n}\right|^{2} \nonumber\\
&= \sigma_{\omegav}^{2} \sigma_{n}^{2}, \nonumber
\end{align}
where $\sigma_{n}^{2}:=\sum_{l=-N}^{N} \left|[\dv_{l}]_{n}\right|^{2}$. Hence, $\widetilde{Q}_{n}\left(\frac{l_{1}}{M},\frac{l_{2}}{M}\right)$ can be viewed as a Gaussian variable of zero mean and a variance equal to $\sigma_{\omegav}^{2} \sigma_{n}^{2}$. Now based on the concentration property on Gaussian variables we can write 
\begin{align}
&\text{Pr}\left(\left(||\mathcal{X}^{*}\left(\omegav\right) ||_{\mathcal{A}}^{*}\right)^{2} \geq \frac{\mu^{2}}{\lambda^{2}}\right) \leq  \nonumber\\
&\text{Pr}\left( \max_{l_{1},l_{2}=0.\dots, M-1} \sum_{i=1}^{K} \left|\widetilde{Q}_{n}\left(\frac{l_{1}}{M},\frac{l_{2}}{M}\right)\right|^{2} \geq \left(1-\frac{4 \pi N}{M}\right)\frac{\mu^{2}}{\lambda^{2}}\right)\nonumber\\
&\leq M^{2} \text{Pr}\left(\sum_{n=1}^{K} \hspace{-3pt} \left|\widetilde{Q}_{n}\hspace{-2pt}\left(\frac{l_{1}}{M},\frac{l_{2}}{M}\hspace{-2pt} \right)\hspace{-2pt} \right|^{2}\hspace{-3pt} \geq \left(1-\frac{4 \pi N}{M}\right)\frac{\mu^{2}}{\lambda^{2}}\hspace{-3pt}\right)\hspace{-2pt}\leq M^{2} \times  \nonumber\\
\label{eq: lemma proof 10}
& \text{Pr}\left(\hspace{-1pt}\sum_{n=1}^{K} \hspace{-1pt}\left|\widetilde{Q}_{n}\hspace{-2pt}\left(\frac{l_{1}}{M},\frac{l_{2}}{M}\right)\right|^{2} \hspace{-5pt} \geq\hspace{-3pt} \left(\hspace{-2pt}1-\hspace{-2pt} \frac{4 \pi N}{M}\right)\hspace{-1pt}36 \sigma_{\omegav}^{2}\hspace{-1pt}\log\left(N\right)\hspace{-2pt}\sum_{n=1}^{K}\hspace{-2pt} \sigma_{n}^{2} \hspace{-3pt}\right)\\
& \leq M^{2}K \text{Pr}\hspace{-2pt} \left(\hspace{-2pt}\left|\widetilde{Q}_{n}\hspace{-3pt} \left(\frac{l_{1}}{M},\frac{l_{2}}{M}\right)\hspace{-2pt} \right| \hspace{-3pt} \geq \hspace{-3pt}\sqrt{1-\hspace{-2pt} \frac{4 \pi N}{M}}6 \sigma_{\omegav}\sigma_{n} \sqrt{\log\left(N\right)}\hspace{-2pt}\right)\nonumber\\
\label{eq: lemma proof 12}
& \leq2 M^{2}K \exp\left( \frac{-36\left(1-\frac{4 \pi N}{M}\right) \sigma_{\omegav}^{2}\sigma_{n}^{2} \log\left(N\right)}{2\sigma_{\omegav}^{2}\sigma_{n}^{2}} \right)\\
\label{eq: lemma proof 13}
& \leq2 M^{2}K \exp\left( - 18\left(1-\frac{4 \pi N}{M}\right) \log\left(N\right) \right)\nonumber\\
& \leq 50 \pi^{2} N^{2} K \frac{1}{N^{3.6}} \leq \frac{C K}{\sqrt{N^{3}}},
\end{align}
where (\ref{eq: lemma proof 10}) is obtained by setting $\mu = 6\lambda \sigma_{\omegav} ||\Dm||_{F} \sqrt{\log\left(N\right)}$ with $||\Dm||_{F}^{2}= \sum_{n=1}^{K}\sigma_{n}^{2}$ while (\ref{eq: lemma proof 12}) is based on the concentration property on Gaussian variables, i.e., for $x \sim \mathcal{N}\left(0,\sigma^{2}_{x}\right)$ we have $\text{Pr}\left(|x|\geq t\right) \leq 2 \exp\left(\frac{-t^{2}}{2\sigma^{2}_{x}}\right)$. Finally, (\ref{eq: lemma proof 13}) is obtained by setting $M=5 \pi N$. 
\end{proof}

\section{Proofs of Theorems~\ref{th: main dual cert}, \ref{th: main dual cert 1}, and \ref{th: main dual cert 2} }
\label{app: Theorems proof}
\subsection{Proof of Theorem~\ref{th: main dual cert}}
We start by recalling some of the results in \cite{suliman2018blind}. First, retrieve the expression of $\fv\left(\rv\right)$ in \cite[Equation (39)]{suliman2018blind}, i.e.,
\begin{align} 
\label{eq: final dual}
\fv\left(\rv\right) =\sum_{j=1}^{R} &\Mm_{\left(0,0\right)}\left(\rv,\rv_{j}\right) \alphav_{j} + \Mm_{\left(1,0\right)}\left(\rv,\rv_{j}\right) \betav_{j}\nonumber\\
&+ \Mm_{\left(0,1\right)}\left(\rv,\rv_{j}\right) \gammav_{j},
\end{align}
where $\Mm_{\left(m,n\right)} \left(\rv,\rv_{j}\right) \in \mathbb{C}^{K \times K}, m,n=0,1$ is given by \cite[Equation (50)]{suliman2018blind} whereas $\alphav_{j}, \betav_{j},\gammav_{j}  \in \mathbb{C}^{K \times 1}$ are vector parameters selected in a way such that $\forall \rv_{j} \in \mathcal{R}$ we have
\begin{align}
\label{eq: theo con 1}
&\fv\left(\rv_{j}\right)= \text{sign}\left(c_{j}\right)\hv_{j} \\
\label{eq: theo con 2}
&\fv^{\left(1,0\right)}\left(\rv_{j}\right)= \bm{0}_{K\times 1}, \ \ \ \fv^{\left(0,1\right)}\left(\rv_{j}\right)= \bm{0}_{K\times 1}.
\end{align}
As discussed in \cite[Sections~VI-B and VI-E]{suliman2018blind}, the formulation of $\fv\left(\rv\right)$ in (\ref{eq: final dual}) ensures directly that (\ref{eq: hold assump 1}) and (\ref{eq: hold assump 3}) are both satisfied. To prove (\ref{eq: hold assump 2}), we can write using Taylor series expansion at $\rv_{j} \in \mathcal{R}$
\begin{align}
\label{eq: tho proof taylor}
&||\fv\left(\rv\right)||_{2}= ||\fv\left(\rv_{j}\right)||_{2}+\left(\tau-\tau_{j}\right) \bigtriangledown_{\tau}||\fv\left(\rv\right)||_{2} |_{\rv=\rv_{j}} \nonumber\\
&+\left(f-f_{j}\right) \bigtriangledown_{f}||\fv\left(\rv\right)||_{2} |_{\rv=\rv_{j}}+ \frac{1}{2}\left(\tau-\tau_{j}\right)^{2} \bigtriangledown_{\tau}^{2}||\fv\left(\rv\right)||_{2} |_{\rv=\rv_{j}}\nonumber\\
&+ \frac{1}{2}\left(f-f_{j}\right)^{2} \bigtriangledown_{f}^{2}||\fv\left(\rv\right)||_{2} |_{\rv=\rv_{j}}\nonumber\\
&+\left(\tau-\tau_{j}\right)\left(f-f_{j}\right) \bigtriangledown_{\tau f}||\fv\left(\rv\right)||_{2} |_{\rv=\rv_{j}}.
\end{align}
Based on (\ref{eq: theo con 2}), the second and the third terms in (\ref{eq: tho proof taylor}) will vanish, thus, (\ref{eq: tho proof taylor}) can be written after some manipulations as 
\begin{align}
\label{eq: tho proof taylor 2}
&||\fv\left(\rv\right)||_{2}=\nonumber\\
& ||\fv\left(\rv_{j}\right)||_{2}+ \frac{1}{2}\left(\tau-\tau_{j}\right)^{2} ||\fv\left(\rv_{j}\right)||_{2}^{-1}\left\langle\fv^{(2,0)}\left(\rv_{j}\right),\fv\left(\rv_{j}\right)\right\rangle \nonumber\\
&+ \frac{1}{2}\left(f-f_{j}\right)^{2} ||\fv\left(\rv_{j}\right)||_{2}^{-1}\left\langle\fv^{(0,2)}\left(\rv_{j}\right),\fv\left(\rv_{j}\right)\right\rangle\nonumber\\
&+\left(\tau-\tau_{j}\right)\left(f-f_{j}\right)||\fv\left(\rv_{j}\right)||_{2}^{-1}\left\langle\fv^{(1,1)}\left(\rv_{j}\right),\fv\left(\rv_{j}\right)\right\rangle
\end{align}
Upon using the results obtained in \cite[Proof of Lemma~16]{suliman2018blind}, and based on the fact that $||\fv\left(\rv_{j}\right)||_{2} \leq 1$, we can show that
\begin{equation}
\label{eq: the proof taylor 3}
||\fv\left(\rv_{j}\right)||_{2}^{-1} \left\langle\fv^{(m,n)}\left(\rv_{j}\right),\fv\left(\rv_{j}\right)\right\rangle \leq  -\bar{C}_{2}^{*}N^{2}
\end{equation}
for $(m,n) \in \{(2,0), (0,2)\}$. Furthermore, we also have
\begin{equation}
\label{eq: the proof taylor 4}
||\fv\left(\rv_{j}\right)||_{2}^{-1} \left\langle\fv^{(1,1)}\left(\rv_{j}\right),\fv\left(\rv_{j}\right)\right\rangle \leq   \bar{C}_{2}^{*} N^{2}.
\end{equation}
Substituting (\ref{eq: the proof taylor 3}) and (\ref{eq: the proof taylor 4}) in (\ref{eq: tho proof taylor 2}), and then manipulating, we obtain (\ref{eq: hold assump 2}) for a constant $\bar{C}_{2}^{*}$.

To prove (\ref{eq: hold assump 4}), define $\kv\left(\rv\right)= \bar{\fv}\left(\rv\right)-\text{sign}\left(c_{j}\right)\hv_{j}$ where $\bar{\fv}\left(\rv\right)=\mathbb{E}\left[\fv\left(\rv\right)\right]$. Based on that, we can write
\begin{equation}
\label{eq: tho proof taylor 5}
||\fv\left(\rv\right)-\text{sign}\left(c_{j}\right)\hv_{j}||_{2}\leq ||\fv\left(\rv\right)-\bar{\fv}\left(\rv\right)||_{2} + ||\kv\left(\rv\right)||_{2}.
\end{equation}
It is shown in \cite[Lemma~13]{suliman2018blind} that the first term in (\ref{eq: tho proof taylor 5}) is very small with probability at least $1- \delta$ when the lower bound on $L$, as provided in Theorem~\ref{th: main dual cert}, is satisfied. As for the second term, we assume without loss of generality that $\rv_{j}=\bm{0}$ and bound $||\kv\left(\rv\right)||_{2}$ in the interval $[0,0.2447/N]$. Upon writing the Taylor series expansion at $\rv_{j}=\bm{0}$, we can obtain 
\begin{align}
&||\kv\left(\rv\right)||_{2} = ||\kv\left(\bm{0}\right)+\tau \kv^{(1,0)}\left(\bm{0}\right)+ f \kv^{(0,1)}\left(\bm{0}\right)+ \nonumber\\
&\frac{1}{2} \tau^{2} \kv^{(2,0)}\left(\rv\right)|_{\rv=\bar{\rv}}+\frac{1}{2} f^{2} \kv^{(0,2)}\left(\rv\right)|_{\rv=\bar{\rv}}+\tau f \kv^{(1,1)}\left(\rv\right)|_{\rv=\bar{\rv}}||_{2} \nonumber\\
&\leq \frac{1}{2} \tau^{2} ||\kv^{(2,0)}\left(\rv\right)|_{\rv=\bar{\rv}}||_{2}+\frac{1}{2} f^{2} ||\kv^{(0,2)}\left(\rv\right)|_{\rv=\bar{\rv}}||_{2} \nonumber\\
&+\tau f ||\kv^{(1,1)}\left(\rv\right)|_{\rv=\bar{\rv}}||_{2} \nonumber\\
&\leq \frac{1}{2} 4.3856 N^{2}  \left(\tau^{2}+ f^{2}\right) + \left(\tau f \right) 0.8270 N^{2}, \nonumber
\end{align}
where $\bar{\rv} \in \Omega_{\text{close}}\left(j\right)$. Herein, the first inequality is based on the fact that $\kv\left(\bm{0}\right), \kv^{(1,0)}\left(\bm{0}\right),$ and $\kv^{(0,1)}\left(\bm{0}\right)$ are all equal to zero. In contrast, for the second inequality, we leverage different results from \cite[Appendix~K]{suliman2018blind} (in particular equations (K.1), (K.3), and (K.4)). Finally, arranging the terms results in (\ref{eq: hold assump 4}). 
\subsection{Proof of Theorem~\ref{th: main dual cert 1}}
\label{subsec: proof theorem 3}
To prove Theorem~\ref{th: main dual cert 1}, we extend the result in \cite[Lemma~7]{candes2013super}. For that, we first set $\fv_{1}\left(\rv\right)$ to be
\begin{align} 
\label{eq: final dual for 1}
\fv_{1}\left(\rv\right) =\sum_{j=1}^{R} &\Mm_{\left(0,0\right)}\left(\rv,\rv_{j}\right) \alphav^{(1)}_{j} + \Mm_{\left(1,0\right)}\left(\rv,\rv_{j}\right) \betav^{(1)}_{j}\nonumber\\
&+ \Mm_{\left(0,1\right)}\left(\rv,\rv_{j}\right) \gammav^{(1)}_{j},
\end{align}
where $\alphav^{(1)}_{j}, \betav^{(1)}_{j},\gammav^{(1)}_{j}  \in \mathbb{C}^{K \times 1}$ are vector parameters selected such that (\ref{eq: conds for f1}) is satisfied. Similar to the results in \cite[Sections~VI-A and VI-B]{suliman2018blind} we have
\begin{align} 
\label{eq: final dual for 1 2}
&\bar{\fv}_{1}\left(\rv\right):=\mathbb{E}\left[\fv_{1}\left(\rv\right)\right] =\sum_{j=1}^{R} \widebar{M}\left(\rv-\rv_{j}\right)\Id_{K} \bar{\alphav}^{(1)}_{j}\nonumber\\
& +\widebar{M}^{\left(1,0\right)}\left(\rv-\rv_{j}\right) \Id_{K}\bar{\betav}^{(1)}_{j}+ \widebar{M}^{\left(0,1\right)}\left(\rv-\rv_{j}\right) \Id_{K} \bar{\gammav}^{(1)}_{j}
\end{align}
where $\bar{\alphav}^{(1)}_{j}, \bar{\betav}^{(1)}_{j}$, and $\bar{\gammav}^{(1)}_{j}$ are the solutions of 
\begin{equation}
\label{eq: conds for f1 average}
\bar{\fv}_{1}\left(\rv_{j}\right)=\bar{\fv}_{1}^{\left(0,1\right)}\left(\rv_{j}\right)= \bm{0}_{K}, \bar{\fv}_{1}^{\left(1,0\right)}\left(\rv_{j}\right)= \text{sign}\left(c_{j}\right) \hv_{j}, \forall \rv_{j} \in \mathcal{R}
\end{equation}
whereas $\widebar{M}^{(m,n)}\left(\rv-\rv_{j}\right)$ is given by \cite[Equation~64]{suliman2018blind}. Based on (\ref{eq: final dual for 1 2}), we can express (\ref{eq: conds for f1 average}) in a matrix-vector form as
\begin{equation}
\left( \underbrace{{\begin{bmatrix} \widebar{\Em}^{(0,0)} & \widebar{\Em}^{(1,0)}   & \widebar{\Em}^{(0,1)}   \\ \widebar{\Em}^{(0,1)} &\widebar{\Em}^{(1,1)}   & \widebar{\Em}^{(0,2)} \\ \widebar{\Em}^{(1,0)}   &\widebar{\Em}^{(2,0)} &\widebar{\Em}^{(1,1)}  
  \end{bmatrix}}}_{\widebar{\Em}} \otimes \Id_{\text{K}} \right) \begin{bmatrix} \bar{\alphav}^{(1)} \\ \bar{\betav}^{(1)}  \\  \bar{\gammav}^{(1)} \end{bmatrix}  =  \begin{bmatrix}  \bm{0}_{RK } \\  \bm{0}_{RK} \\ \hv \end{bmatrix}, \nonumber 
\end{equation}
where $\widebar{\Em}^{(m',n')} \in \mathbb{C}^{R \times R}$ with $[\widebar{\Em}^{(m',n')}]_{(l,k)} := \widebar{M}^{(m',n')}\left(\rv_l - \rv_k\right)$ while $\bar{\alphav}^{(1)}  = [ \bar{\alphav}_{1}^{(1) T}, \dots , \bar{\alphav}_{R}^{(1) T} ]^{T}, \bar{\betav}^{(1)}  = [ \bar{\betav}_{1}^{(1) T}, \dots ,\bar{\betav}_{R}^{(1) T} ]^{T}$, and $\bar{\gammav}^{(1)} = [\bar{\gammav}_{1}^{(1) T}, \dots , \bar{\gammav}_{R}^{(1) T} ]^{T}$. Finally, $\hv = \left[ \text{sign}\left(c_{1}\right)\hv_{1}^{T}, \hdots, \text{sign}\left(c_{R}\right)\hv_{R}^{T}\right]^{T} \in \mathbb{C}^{RK \times 1}$. It is shown in \cite[Proposition~3]{suliman2018blind} that $\widebar{\Em}$ and $\widebar{\Em} \otimes \Id_{\text{K}}$ are invertible, and therefore, $\bar{\alphav}^{(1)}, \bar{\betav}^{(1)},$ and $\bar{\gammav}^{(1)}$ are well-defined. In order to prove Theorem~\ref{th: main dual cert 1}, we first need to find upper bounds on these coefficients. For that, define
\begin{align}
& \widebar{\Em}_{1}  = {\begin{bmatrix} \widebar{\Em}^{(0,0)} & \widebar{\Em}^{(0,1)}  \\ \widebar{\Em}^{(0,1)} &\widebar{\Em}^{(0,2)} \end{bmatrix}},  \ \ \widebar{\Em}_{2}  = {\begin{bmatrix} \widebar{\Em}^{(1,0)} \\ \widebar{\Em}^{(1,1)} \end{bmatrix}}, \nonumber\\
&  \widebar{\Em}_{3}  = {\begin{bmatrix} \widebar{\Em}^{(1,0)} & \widebar{\Em}^{(1,1)} \end{bmatrix}}, \ \ \ \widebar{\Em}_{4}  = \widebar{\Em}^{(2,0)}, \nonumber
\end{align}
and let $\bar{\etav}^{(1)} = [\bar{\gammav}^{(1) T}, \ \bar{\betav}^{(1) T} ]^{T}$. By using the block matrix inversion lemma for the $2 \times 2$ block matrices we can write
\begin{equation}
\label{eq: solution with block matr}
{\begin{bmatrix} \bar{\alphav}^{(1)} \\ \bar{\etav}^{(1)}
  \end{bmatrix}} =  \left(\begin{bmatrix}  -\widebar{\Em}_{1}^{-1} \widebar{\Em}_{2}\\  \Id_{\text{R} \times \text{R}} \end{bmatrix}\Sm_{3}^{-1} \otimes  \Id_{\text{K}}\right) \hv, 
\end{equation}
where
\begin{equation}
\label{eq: matrix S}
\Sm_{3}= \widebar{\Em}_{4} -\widebar{\Em}_{3} \widebar{\Em}_{1}^{-1} \widebar{\Em}_{2}.
\end{equation}
Upon defining 
\begin{align}
\label{eq: matrix S1}
\Sm_{1}= \widebar{\Em}^{(0,0)} - \widebar{\Em}^{(0,1)} \left.\widebar{\Em}^{(0,2)}\right.^{-1} \widebar{\Em}^{(0,1)} \\ 
\label{eq: matrix S2}
\Sm_{2}= \widebar{\Em}^{(1,0)} - \widebar{\Em}^{(0,1)} \left.\widebar{\Em}^{(0,2)}\right.^{-1} \widebar{\Em}^{(1,1)},
\end{align}
we can easily show that 
\begin{equation}
\label{eq: matrix S new}
\Sm_{3}= \widebar{\Em}^{(2,0)} + \Sm_{2}^{T} \Sm_{1}^{-1} \Sm_{2}- \widebar{\Em}^{(1,1)} \left.\widebar{\Em}^{(0,2)}\right.^{-1} \widebar{\Em}^{(1,1)}. 
\end{equation}
The definitions of $\Sm_{1}, \Sm_{2},$ and $\Sm_{3}$ allow us to rewrite (\ref{eq: solution with block matr}) after some algebraic manipulations as 
\begin{equation}
\label{eq: solution with block new}
\begin{bmatrix} \bar{\alphav}^{(1)} \\ \bar{\gammav}^{(1)} \\   \bar{\betav}^{(1)} \end{bmatrix} =  \left(\begin{bmatrix}  -\Sm_{1}^{-1}\Sm_{2}\\ \Sm_{4} \\ \Id_{\text{R} \times \text{R}} \end{bmatrix}\Sm_{3}^{-1} \otimes  \Id_{\text{K}}\right) \hv, 
\end{equation}
where 
\begin{equation}
\Sm_{4} = \left.\widebar{\Em}^{(0,2)}\right.^{-1} \widebar{\Em}^{(0,1)} \Sm_{1}^{-1} \Sm_{2} -  \left.\widebar{\Em}^{(0,2)}\right.^{-1} \widebar{\Em}^{(1,1)}.
\end{equation}  
To prove the bounds in (\ref{eq: hold assump 3 for 1}) and (\ref{eq: hold assump 4 for 1}), we first need to derive an upper bound on the Euclidean norms of $\bar{\alphav}^{(1)}, \bar{\betav}^{(1)}$, and $\bar{\gammav}^{(1)}$. For that, we will first highlight some important results. 

First, note that $\widebar{\Em}^{(0,0)},\widebar{\Em}^{(1,1)}, \widebar{\Em}^{(2,0)},$ and $\widebar{\Em}^{(0,2)}$ are symmetric matrices while $\widebar{\Em}^{(1,0)}$ and $\widebar{\Em}^{(0,1)}$ are antisymmetric matrices. Thus, $\Sm_{1}$ is symmetric matrix, whereas $\Sm_{2}$ is antisymmetric. Moreover, a symmetric matrix $\Gm $ is invertible if
\begin{equation}
\label{eq:inver of summetric}
||\Id-\Gm||_{\infty} < 1,
\end{equation}
while for any symmetric matrix $\Gm$ we have
\begin{equation}
\label{eq: symetric bound}
||\Gm^{-1}||_{\infty} \leq \frac{1}{1-||\Id_{\text{R}}-\Gm||_{\infty}}.
\end{equation}
Finally, we need the following proposition, which is based on the results obtained in \cite[Appendix C]{candes2014towards}.
\begin{proposition}
\label{pro: various bounds}
The following bounds hold true under the minimum separation in (\ref{eq: seperation condition}) 
\begin{align}
\label{eq: conditon mat 1}
&||\Id_{\text{R}} -\widebar{\Em}^{(0,0)}||_{\infty}  \leq  0.04854 \\
\label{eq: conditon mat 2}
&||\widebar{\Em}^{(1,0)}||_{\infty}= ||\widebar{\Em}^{(0,1)}||_{\infty}  \leq  7.723\times 10^{-2} N \\
\label{eq: conditon mat 3}
&||\widebar{\Em}^{(1,1)}||_{\infty} \leq  0.1576 N^{2} \\
\label{eq: conditon mat 4}
&||\left.\widebar{\Em}^{(0,2)}\right.^{-1}||_{\infty}  \leq  \frac{0.3399}{N^{2}} \\
\label{eq: conditon mat 5}
&||\ |\widebar{M}^{(2,0)}\left(0\right)|\Id_{\text{R}} -\widebar{\Em}^{(2,0)}||_{\infty}  \leq  0.3539 N^{2}.
\end{align}
\end{proposition}
Based on (\ref{eq: matrix S1}) and Proposition~\ref{pro: various bounds} we can write
\begin{align}
||\Id_{\text{R}} -\Sm_{1}||_{\infty} &\leq  ||\Id_{\text{R}} -\widebar{\Em}^{(0,0)}||_{\infty} +|| \widebar{\Em}^{(0,1)} ||_{\infty}^{2} ||\left.\widebar{\Em}^{(0,2)}\right.^{-1} ||_{\infty} \nonumber\\
& \leq 0.0506. \nonumber
\end{align}
Hence, $\Sm_{1}$ is invertible and we can obtain
\begin{equation}
\label{eq: bound on S inverse}
||\Sm_{1}^{-1}||_{\infty} \leq \frac{1}{1-||\Id_{\text{R}}-\Sm_{1}||_{\infty}} \leq  1.0533. 
\end{equation} 
Next, Proposition~\ref{pro: various bounds} allows us to bound $\Sm_{2}$ as
\begin{align}
\label{eq: bounding S2}
||\Sm_{2}||_{\infty} &\leq  ||\widebar{\Em}^{(1,0)}||_{\infty} +|| \widebar{\Em}^{(0,1)} ||_{\infty} ||\left.\widebar{\Em}^{(0,2)}\right.^{-1} ||_{\infty}  || \widebar{\Em}^{(1,1)} ||_{\infty} \nonumber\\
& \leq  0.0814 N.
\end{align}
Applying Proposition~\ref{pro: various bounds} along with (\ref{eq: bound on S inverse}) and (\ref{eq: bounding S2}) yields
\begin{align}
&||\ |\widebar{M}^{(2,0)}\left(0\right)|\Id_{\text{R}} -\Sm_{3}||_{\infty}\leq ||\ |\widebar{M}^{(2,0)}\left(0\right)|\Id_{\text{R}} -\widebar{\Em}^{(2,0)}||_{\infty} \nonumber\\
& + ||\Sm_{2}||_{\infty}^{2} ||\Sm_{1}^{-1}||_{\infty} + || \widebar{\Em}^{(1,1)}||_{\infty}^{2} ||\left.\widebar{\Em}^{(0,2)}\right.^{-1}||_{\infty} \nonumber\\
& \leq 0.3693 N^{2} \nonumber
\end{align}
which when combined with $\widebar{M}^{(2,0)}\left(0\right)= -\frac{\pi^{2}}{3}N\left(N+4\right)$ and (\ref{eq: symetric bound}) yields
\begin{equation}
||\Sm_{3}^{-1}||_{\infty} \hspace{-3pt} \leq \hspace{-2pt}\frac{1}{|\widebar{M}^{(2,0)}\left(0\right)| - ||\ |\widebar{M}^{(2,0)}\left(0\right)|\Id_{\text{R}} -\Sm_{3}||_{\infty}} \hspace{-2pt}\leq\hspace{-2pt} \frac{0.3424}{N^{2}}   \nonumber
\end{equation}
Finally, we can also show using Proposition~\ref{pro: various bounds} that 
\begin{align}
||\Sm_{4}||_{\infty} &\leq  0.0558. \nonumber
\end{align}
Based on the previous results, (\ref{eq: solution with block new}), and \cite[Lemma 5.3]{yang2014exact}, we can conclude that 
\begin{align}
\label{eq: bound on alpha}
&\max_{j=1, \dots, R} ||\bar{\alphav}^{(1)}_{j}||_{2} \leq ||\Sm_{1}^{-1}\Sm_{2} \Sm_{3}^{-1} \otimes  \Id_{\text{K}} ||_{\infty} \max_{j=1,\dots, R}||\hv_{j}||_{2} \nonumber\\
& \leq ||\Sm_{1}^{-1}||_{\infty} ||\Sm_{2}||_{\infty}  ||\Sm_{3}^{-1}||_{\infty} \leq  \frac{0.0294}{N}
\end{align}
\begin{align} 
\label{eq: bound on beta}
\max_{j=1, \dots, R} ||\bar{\betav}^{(1)}_{j}||_{2} &\leq ||\Sm_{3}^{-1}||_{\infty} \leq  \frac{0.3424}{N^{2}}.
\end{align}
\begin{align}
\label{eq: bound on gamma}
&\max_{j=1, \dots, R} ||\bar{\gammav}^{(1)}_{j}||_{2} \leq ||\Sm_{4}||_{\infty}  ||\Sm_{3}^{-1}||_{\infty} \leq  \frac{ 0.0191}{N^{2}}.
\end{align}
Now to show (\ref{eq: hold assump 3 for 1}) we write
\begin{align}
\label{eq: bound step 1}
||\fv_{1}\left(\rv\right)||_{2} \leq ||\fv_{1}\left(\rv\right)-\bar{\fv_{1}}\left(\rv\right)||_{2}+ ||\bar{\fv_{1}}\left(\rv\right)||_{2} \leq {C}^{*}||\bar{\fv_{1}}\left(\rv\right)||_{2}
\end{align}
where we apply the result in \cite[Lemma~13]{suliman2018blind}. Upon using (\ref{eq: final dual for 1 2}), along with (\ref{eq: bound on alpha}), (\ref{eq: bound on beta}), and (\ref{eq: bound on gamma}), we can obtain
\begin{align}
\label{eq: bound step 2}
&||\bar{\fv_{1}}\left(\rv\right)||_{2} \leq \max_{j=1, \dots, R} ||\bar{\alphav}^{(1)}_{j}||_{2} \sum_{j=1}^{R} |\widebar{M}\left(\rv-\rv_{j}\right) | \nonumber\\
&+\max_{j=1, \dots, R} ||\bar{\betav}^{(1)}_{j}||_{2} \sum_{j=1}^{R} |\widebar{M}^{\left(1,0\right)}\left(\rv-\rv_{j}\right) |\nonumber\\
&+\max_{j=1, \dots, R} ||\bar{\gammav}^{(1)}_{j}||_{2} \sum_{j=1}^{R} |\widebar{M}^{\left(0,1\right)}\left(\rv-\rv_{j}\right) | \nonumber\\
&\leq  \frac{0.0294}{N} \bar{C} +  \left(\left(\frac{ 0.0191}{N^{2}} +  \frac{0.3424}{N^{2}}\right)\right) \hat{C} N \leq \frac{C}{N},
\end{align}
where the last inequality is based on the results in \cite[Section 2.3, Appendix~C]{candes2014towards}. Substituting (\ref{eq: bound step 2}) in (\ref{eq: bound step 1}) we get (\ref{eq: hold assump 3 for 1}). 

To prove (\ref{eq: hold assump 4 for 1}) we write
\begin{align}
||\fv_{1}\left(\rv\right)- \hv_{j}\left(\tau-\tau_{j}\right)||_{2} &\leq ||\fv_{1}\left(\rv\right)-\bar{\fv_{1}}\left(\rv\right)||_{2}+ ||\kv_{1}\left(\rv\right)||_{2} \nonumber\\
&\leq C^{*}||\kv_{1}\left(\rv\right)||_{2}, \nonumber
\end{align}
where $\kv_{1}\left(\rv\right):= \bar{\fv}_{1}\left(\rv\right)-\hv_{j}\left(\tau-\tau_{j}\right)$ and the last inequality is based on \cite[Lemma~13]{suliman2018blind}. Without loss of generality, we assume that $\rv_{j}=\bm{0}$ and bound $\kv_{1}\left(\rv\right)$ in the interval of $\Omega_{\text{close}}\left(j\right)$. By using Tylor series expansion around $\rv_{j}$ we obtain
\begin{align}
&||\kv_{1}\left(\rv\right)||_{2} = ||\kv_{1}\left(\bm{0}\right)+\tau  \kv^{(1,0)}_{1}\left(\bm{0}\right) +f  \kv^{(0,1)}_{1}\left(\bm{0}\right)\nonumber\\
&+\frac{1}{2} \tau^{2} \kv^{(2,0)}_{1}\left(\rv\right)|_{\rv=\bar{\rv}} +\frac{1}{2} f^{2} \kv^{(0,2)}_{1}\left(\rv\right)|_{\rv=\bar{\rv}} \nonumber\\
&+ \tau f  \ \kv_{1}^{(1,1)}\left(\rv\right)|_{\rv=\bar{\rv}}||_{2} \leq C N \left(\tau+f\right)^{2}, \nonumber
\end{align}
where $\bar{\rv} \in \Omega_{\text{close}}\left(j\right)$. Here, the last inequality is based on the fact that $\kv_{1}\left(\bm{0}\right)$, $\kv^{(1,0)}_{1}\left(\bm{0}\right)$, and $\kv^{(0,1)}_{1}\left(\bm{0}\right)$ are all equal to zero, the triangular inequality, and Proposition~\ref{pro: various bounds 2} given below.
\begin{proposition}
\label{pro: various bounds 2}
The following bounds hold
\begin{align}
&||\kv^{(2,0)}_{1}\left(\rv\right)|_{\rv=\bar{\rv}}||_{2} \leq 3.4462 N \nonumber\\
&||\kv^{(0,2)}_{1}\left(\rv\right)|_{\rv=\bar{\rv}}||_{2} \leq 0.9779 N\nonumber\\
&||\kv^{(1,1)}_{1}\left(\rv\right)|_{\rv=\bar{\rv}}||_{2} \leq 0.7409 N\nonumber. 
\end{align}
\end{proposition}
The proof of Proposition~\ref{pro: various bounds 2}, which is omitted here, is based on (\ref{eq: bound on alpha}), (\ref{eq: bound on beta}), and (\ref{eq: bound on gamma}), as well as the results obtained in \cite[Lemma 2.3 and Section C.2]{candes2014towards}. This completes the proof of Theorem~\ref{th: main dual cert 1}. 

\subsection{Proof of Theorem~\ref{th: main dual cert 2}}

The proof of Theorem~\ref{th: main dual cert 2} can be conducted in the same way as that of Theorem~\ref{th: main dual cert 1}. For that, we set $\fv_{2}\left(\rv\right)$ to be 
\begin{align} 
\label{eq: final dual for 2}
\fv_{2}\left(\rv\right) =\sum_{j=1}^{R} &\Mm_{\left(0,0\right)}\left(\rv,\rv_{j}\right) \alphav^{(2)}_{j} + \Mm_{\left(1,0\right)}\left(\rv,\rv_{j}\right) \betav^{(2)}_{j}\nonumber\\
&+ \Mm_{\left(0,1\right)}\left(\rv,\rv_{j}\right) \gammav^{(2)}_{j},
\end{align}
where $\alphav^{(2)}_{j}, \betav^{(2)}_{j},\gammav^{(2)}_{j}  \in \mathbb{C}^{K \times 1}$ are selected such that (\ref{eq: conds for f2}) is satisfied. 



\section{Proof of Lemma~\ref{lemma: bounds for Ts} }
\label{app: Ts bounds}
Starting from Lemma~\ref{lemma: error bound} we can write
\begin{align}
\label{eq: app bound T 1}
&T_{0} = \sum_{j=1}^{R} \left|\left|  \ \ \iint \limits_{\Omega_{\text{close}}\left(j\right)} \nuv\left(\rv\right)d\rv \right|\right|_{2}= \sum_{j=1}^{R} \ \ \iint \limits_{\Omega_{\text{close}}\left(j\right)} \nuv\left(\rv\right)^{H}d\rv \times  \nonumber\\
& \frac{\ \ \iint \limits_{\Omega_{\text{close}}\left(j\right)} \nuv\left(\bar{\rv}\right)d\bar{\rv}}{\left|\left|  \ \ \iint \limits_{\Omega_{\text{close}}\left(j\right)} \nuv\left(\bar{\rv}\right)d\bar{\rv} \right|\right|_{2}}=  \underbrace{\sum_{j=1}^{R} \ \ \iint \limits_{\Omega_{\text{close}}\left(j\right)} \nuv\left(\rv\right)^{H} \fv\left(\rv\right) d\rv}_{T_{0,1}} +  \nonumber\\
&\underbrace{\sum_{j=1}^{R} \ \ \iint \limits_{\Omega_{\text{close}}\left(j\right)} \nuv\left(\rv\right)^{H}\left[  \frac{\ \ \iint \limits_{\Omega_{\text{close}}\left(j\right)} \nuv\left(\bar{\rv}\right)d\bar{\rv}}{\left|\left|  \ \ \iint \limits_{\Omega_{\text{close}}\left(j\right)} \nuv\left(\bar{\rv}\right)d\bar{\rv} \right|\right|_{2}}-\fv\left(\rv\right)\right]d\rv}_{T_{0,2}}
\end{align}
By using (\ref{eq: hold assump 3}) we can bound $T_{0,1} $ by
\begin{align}
\label{eq: app bound T 2}
T_{0,1} &\leq  \left| \ \iint \limits_{0}^{1} \nuv\left(\rv\right)^{H}\fv\left(\rv\right) d\rv \right| + \ \iint \limits_{\Omega_{\text{far}}} ||\nuv\left(\rv\right)||_{2} d\rv. 
\end{align}
On the other hand, $T_{0,2}$ can be bounded using (\ref{eq: hold assump 4}) as follows
\begin{align}
\label{eq: app bound T 3}
T_{0,2} &\leq  \sum_{j=1}^{R} \ \ \iint \limits_{\Omega_{\text{close}}\left(j\right)} ||\nuv\left(\rv\right)||_{2}\left(\bar{C}_{2}N^{2}\left(|\tau-\tau_{j}|+|f-f_{j}|\right)^{2}\right)d\rv \nonumber\\
& = C_{5}T_{3}.
\end{align}
Substituting (\ref{eq: app bound T 2}) and (\ref{eq: app bound T 3}) in (\ref{eq: app bound T 1}) leads to the bound on $T_{0}$ as in Lemma~\ref{lemma: bounds for Ts}.

By following the same steps that led to (\ref{eq: app bound T 1}) we can obtain
\begin{align}
\label{eq: app bound T 4}
&T_{1} = \underbrace{(2 \pi N) \sum_{j=1}^{R} \ \ \iint \limits_{\Omega_{\text{close}}\left(j\right)}  \nuv\left(\rv\right)^{H} \fv_{1}\left(\rv\right) d\rv}_{T_{1,1}}  +  \nonumber\\
& \underbrace{2 \pi N \sum_{j=1}^{R} \ \iint \limits_{\Omega_{\text{close}}\left(j\right)} \hspace{-3pt} \nuv\left(\rv\right)^{H}\hspace{-3pt} \left[ \hspace{-1pt} \frac{(\tau-\tau_{j}) \iint \limits_{\Omega_{\text{close}}\left(j\right)} \nuv\left(\bar{\rv}\right)d\bar{\rv}}{\left|\left|  \ \iint \limits_{\Omega_{\text{close}}\left(j\right)}  \nuv\left(\bar{\rv}\right)d\bar{\rv} \right|\right|_{2}}-\fv_{1}\left(\rv\right)\right]d\rv}_{T_{1,2}}
\end{align}
The first term in (\ref{eq: app bound T 4}) can be bounded using (\ref{eq: hold assump 3 for 1}) as
\begin{align}
T_{1,1} &\leq  (2 \pi N) \left| \ \iint \limits_{0}^{1} \nuv\left(\rv\right)^{H}\fv_{1}\left(\rv\right) d\rv \right| + \bar{C}_{5} \iint \limits_{\Omega_{\text{far}}} ||\nuv\left(\rv\right)||_{2} d\rv \nonumber
\end{align}
while for the second term, we can write based on (\ref{eq: hold assump 4 for 1}) 
\begin{align}
&T_{1,2} \leq  (2 \pi N)  \times \nonumber\\
& \sum_{j=1}^{R}\ \iint \limits_{\Omega_{\text{close}}\left(j\right)} \hspace{-3pt} ||\nuv\left(\rv\right)||_{2}\left(\bar{C}_{3}N\left(|\tau-\tau_{j}|+|f-f_{j}|\right)^{2}\right)\hspace{-1pt}d\rv= C_{5}T_{3}. \nonumber
\end{align}
Substituting the above results in (\ref{eq: app bound T 4}), and then manipulating, we obtain the bound on $T_{1}$ as in Lemma~\ref{lemma: bounds for Ts}. 

Upon following the same previous steps, and by using $\fv_{2}\left(\rv\right)$, (\ref{eq: hold assump 3 for 2}), and (\ref{eq: hold assump 4 for 2}), we can obtain the upper bound on $T_{2}$.

\section{Proof of Lemma~\ref{lemma: bound for 0 1 term} }
\label{app: proof for bound 0 1}
The proof of Lemma~\ref{lemma: bound for 0 1 term} is based on Matrix Bernstein inequality lemma, which is given below.
\begin{lemma}
\label{le: matrix berns}
\cite[Theorem 1.6.2]{tropp2015introduction}
Let $\Sm_{1},\dots,\Sm_{n}$ be $ N_{1} \times N_{2}$ independent, centred, uniformly bounded random matrices with $\mathbb{E}\left[\Sm_{k}\right]=\bm{0}$ and $||\Sm_{k}||_{2} \leq q$ for $k= 1, \dots, n. $
Moreover, define the sum
\begin{equation}
\Zm = \sum_{k=1}^{n} \Sm_{k} \nonumber \nonumber
\end{equation}
and denote the matrix variance statistic of the sum by $\nu(\Zm)$, i.e.,
\begin{equation}
\nu(\Zm) := \max \left\lbrace\big|\big|\mathbb{E}\left[\Zm^{H}\Zm\right]\big|\big|_{2}, \big|\big|\mathbb{E}\left[\Zm\Zm^{H}\right]\big|\big|_{2} \right\rbrace. \nonumber
\end{equation}
Then, for every $t \geq 0$ we have
\begin{align}
\text{Pr}\left[||\Zm||_{2} \geq t\right] \leq \left(N_{1} +N_{2}\right) \exp \left( \frac{-t^{2}/2}{\nu(\Zm)+ q t/3}\right). \nonumber
\end{align}
\end{lemma}
Now starting from the left-hand side of (\ref{eq: bound for int 0 to 1}), and based on the fact that $\fv\left(\rv\right)=\mathcal{X}^{*}\left(\qv\right) \av\left(\rv\right)$, we can write
\begin{align}
\label{eq: proof 0 1 no 1}
&\left| \ \iint \limits_{0}^{1} \nuv\left(\rv\right)^{H}\fv\left(\rv\right) d\rv \right| = \left| \ \iint \limits_{0}^{1} \left\langle \mathcal{X}^{*}\left(\qv\right), \nuv\left(\rv\right)\av\left(\rv\right)^{H}\right\rangle d\rv \right| \nonumber\\
&=  \left| \left\langle\qv,\mathcal{X} \left( \iint \limits_{0}^{1} \nuv\left(\rv\right)\av\left(\rv\right)^{H }d\rv \right) \right\rangle \right| = \left|\left\langle\qv,\ev\right\rangle\right| \nonumber\\
&= \left|\left\langle\left( \mathcal{X} \mathcal{X}^{*}\right)^{-1}\qv, \mathcal{X} \mathcal{X}^{*}\left(\ev\right)\right\rangle\right| = \left|\left\langle\tilde{\fv}\left(\rv\right),\phiv\left(\rv\right)\right\rangle\right|,
\end{align}
where we set $\tilde{\fv}\left(\rv\right) := \mathcal{X}^{*} \left( \mathcal{X} \mathcal{X}^{*}\right)^{-1}\left(\qv\right) \av\left(\rv\right)$. Based on H\"{o}lder's inequality, we can obtain based on (\ref{eq: proof 0 1 no 1})
\begin{align}
\label{eq: proof 0 1 no 2}
&\left| \ \iint \limits_{0}^{1} \nuv\left(\rv\right)^{H}\fv\left(\rv\right) d\rv \right| =  \left|\ \iint \limits_{0}^{1}\phiv\left(\rv\right)^{H}\tilde{\fv}\left(\rv\right) d \rv\right| \leq \nonumber\\
&\iint \limits_{0}^{1} \left|\phiv\left(\rv\right)^{H}\tilde{\fv}\left(\rv\right) d \rv\right| \leq \sup_{\rv \in [0,1]^{2}} ||\phiv\left(\rv\right)||_{2} \left(\iint \limits_{0}^{1} ||\tilde{\fv}\left(\rv\right) ||_{2} d\rv\right).
\end{align}
Upon using the definitions of $\mathcal{X}$ and $\widetilde{\Dm}_{p}$, we can write
\begin{equation}
\left[ \mathcal{X} \mathcal{X}^{*}\left(\qv\right)\right]_{r} =\text{Tr}\left(\sum_{p=-N}^{N} \widetilde{\Dm}_{r} \widetilde{\Dm}_{p}^{H} [\qv]_{p}\right) =  [\text{diag}\left(\uv\right) \qv]_{r}, \nonumber
\end{equation}
where $\mathcal{X} \mathcal{X}^{*} = \text{diag}\left(\uv\right) = \text{diag}\left([u_{-N}, \dots, u_{N}]\right)$ with
\begin{equation}
u_{m}= \sum_{k,l,p'=-N}^{N} e^{i 2\pi \frac{k(m-p')}{L}} \dv_{(m-l)}^{H} \dv_{(p'-l)} = L \sum_{i=-N}^{N} ||\dv_{i}||_{2}^{2}  \nonumber
\end{equation} 
Note that the last equality above is based on the fact that
\begin{align}
     \sum_{k=-N}^{N} e^{i2\pi \frac{k \left(m-p'\right)}{L}}  =\left\{
                \begin{array}{ll}
                  L &\text{if} \ \ p' = m\\
                  0 &\text{if} \ \ p' \neq m\\
                \end{array}
              \right. \nonumber
\end{align}
Based on the derivation of $\fv\left(\rv\right)$ in \cite[Section~VI]{suliman2018blind} we can express $\tilde{\fv}\left(\rv\right)$ as
\begin{align} 
\label{eq: new exp for f bar}
\tilde{\fv}\left(\rv\right) = \sum_{j=1}^{R} &\widetilde{\Mm}_{\left(0,0\right)}\left(\rv,\rv_{j}\right) \alphav_{j} + \widetilde{\Mm}_{\left(1,0\right)}\left(\rv,\rv_{j}\right) \betav_{j}\nonumber\\
&+ \widetilde{\Mm}_{\left(0,1\right)}\left(\rv,\rv_{j}\right) \gammav_{j}.
\end{align}
Herein, $\widetilde{\Mm}_{(m,n)}\left(\rv, \rv_{j}\right)$ with $m,n =0,1$ can be expressed using \cite[Equation (53)]{suliman2018blind}, and the definition of $\mathcal{X} \mathcal{X}^{*}$ as
\begin{align}
\label{eq: M titlde expre}
&\widetilde{\Mm}_{(m,n)}\left(\rv, \rv_{j}\right) = \nonumber\\
& \frac{1}{L^{2}P^{2}} \sum_{p,l,l',k,k'=-N}^{N} g_{k'}\left(-i2\pi k' \right)^{m} g_{p} \left(-i2\pi p \right)^{n} e^{i 2 \pi \frac{\left(kl-k'l'\right)}{L}} \times \nonumber\\
&  e^{-i2\pi \left(k\tau-k'\tau_{j}\right)} e^{-i2\pi p \left(f-f_{j}\right)}\frac{\dv_{\left(p-l\right)} \dv_{\left(p-l'\right)}^{H}}{\sum_{i=-N}^{N} ||\dv_{i}||_{2}^{2}},
\end{align}
with $P = \frac{N}{2}+1$ and
\begin{equation}
\label{eq: fejer coe}
g_{n}= \frac{1}{P} \sum_{l=\text{max}\{n-P,-P\}}^{\text{min}\{n+P,P\}} \left(1-\frac{|l|}{P}\right) \left(1-\frac{|n-l|}{P}\right). 
\end{equation}
Upon following the same steps that led to the matrix-vector form of $\fv\left(\rv\right)$ in \cite[Section~VI-C (Equation (107)]{suliman2018blind}, we can write $\tilde{\fv}\left(\rv\right)$ in a matrix-vector form as
\begin{equation}
\label{tilte f matrix-vector}
\tilde{\fv}\left(\rv\right) = \left.\Tm\left(\rv\right)\right.^{H} \Lm \hv,
\end{equation}
where the matrix $\Tm\left(\rv\right) \in \mathbb{C}^{3RK \times K}$ is given by
\begingroup
\fontsize{9.4pt}{9.9pt}
\begin{align}
\label{eq: matrix T}
& \left.\Tm\left(\rv\right)\right.^{H}:=  \nonumber\\
&\left[ \widetilde{\Mm}_{(0,0)}\left(\rv,\rv_{1}\right), \hdots, \widetilde{\Mm}_{(0,0)}\left(\rv,\rv_{R}\right), \frac{1}{\kappa}\widetilde{\Mm}_{(1,0)}\left(\rv,\rv_{1}\right), \hdots,  \right.\nonumber\\ &\left.\frac{1}{\kappa}\widetilde{\Mm}_{(1,0)}\left(\rv,\rv_{R}\right), \frac{1}{\kappa}\widetilde{\Mm}_{(0,1)}\left(\rv,\rv_{1}\right), \hdots,  \frac{1}{\kappa}\widetilde{\Mm}_{(0,1)}\left(\rv,\rv_{R}\right) \right]
\end{align}
\endgroup
with $\kappa = \sqrt{\frac{\pi^{2}}{3}\left(N^{2}+4N\right)}$. Moreover, $\Lm \in \mathbb{C}^{3RK \times RK}$ contains the first $RK$ columns of the inverse of the matrix $\Em$, which is defined in \cite[Equation (90)]{suliman2018blind} as
\begin{equation}
\label{eq: matrix poly}
\Em ={\begin{bmatrix} \Em^{(0,0)}_{(0,0)} & \frac{1}{\kappa}\Em^{(0,0)}_{(1,0)}   & \frac{1}{\kappa}\Em^{(0,0)}_{(0,1)}  \\ -\frac{1}{\kappa}\Em^{(1,0)}_{(0,0)}   & -\frac{1}{\kappa^2}\Em^{(1,0)}_{(1,0)} &-\frac{1}{\kappa^2}\Em^{(1,0)}_{(0,1)}  \\ -\frac{1}{\kappa}\Em^{(0,1)} _{(0,0)} &-\frac{1}{\kappa^2}\Em^{(0,1)}_{(1,0)}   & -\frac{1}{\kappa^2}\Em^{(0,1)}_{(0,1)} 
  \end{bmatrix}}
\end{equation}
where $\Em^{(m',n')}_{(m,n)} \in \mathbb{C}^{RK \times RK}$ consists of $R \times R$ block matrices of size $K \times K$ with the matrix at the $\left(l,k\right)$ location being given by $[\Em^{(m',n')}_{(m,n)}]_{(l,k)} := \widetilde{\Mm}^{(m',n')}_{(m,n)}\left(\rv_{l},\rv_{k}\right)$. Finally, $\Lm\hv=\begin{bmatrix} \alphav^{H}, \kappa \betav^{H}  , \kappa \gammav^{H} \end{bmatrix}^{H}.$ Based on (\ref{tilte f matrix-vector}) we can write
\begin{eqnarray}
\label{eq: proof 0 1 no 3}
\iint \limits_{0}^{1}\hspace{-2pt} ||\tilde{\fv}\left(\rv\right) ||_{2} d\rv \leq \hspace{-3pt} \iint \limits_{0}^{1}\hspace{-3pt}\Bigg(||\underbrace{\left.\Delta\Tm\right.^{H} \Lm \hv}_{\vv_{1}\left(\rv\right)} ||_{2} + ||\underbrace{\widebar{\Tm}^{H} \Lm \hv}_{\vv_{2}\left(\rv\right)}  ||_{2} \hspace{-2pt}\Bigg) d\rv
\end{eqnarray}
where $\Delta \Tm = \Tm-\widebar{\Tm}$ with $\widebar{\Tm} = \mathbb{E}\left[\Tm\right]$\footnote{Here, and for the rest of the paper, we will drop the dependency of $\Tm, \widebar{\Tm},$ and $\Delta \Tm$ on $\rv$ to simplify the expressions.}. Now consider the following definition
\begin{align}
\label{eq: new matrix W}
\Wm\left(\rv\right) &:=\left.\Delta \Tm\right.^{H} \Lm  =\begin{bmatrix} \Wm_{1}, \dots, \Wm_{R}\end{bmatrix}  \in \mathbb{C}^{K \times RK},
\end{align}
where $\Wm_{j} \in \mathbb{C}^{K \times K}$. Upon using the definition of $\hv$ in Appendix~\ref{subsec: proof theorem 3}, and based on (\ref{eq: new matrix W}), we can write $\vv_{1}\left(\rv\right)$ as
\begin{equation}
\label{eq: redefine v vector}
\vv_{1}\left(\rv\right) =  \sum_{j=1}^{R} \Wm_{j} \ \text{sign}\left(c_{j}\right) \hv_{j} =: \sum_{j=1}^{R} \wv_{j}.
\end{equation}
It is easy to show that $\vv_{1}\left(\rv\right)$ is a sum of independent zero-mean vectors. Therefore, we can apply the Matrix Bernstein inequality in Lemma~\ref{le: matrix berns} to obtain a probability measure on the bound of $||\vv_{1}\left(\rv\right)||_{2}$. For that, we first need to find $q\left(\vv_{1}\right)$ and $\nu\left(\vv_{1}\right)$ as in Lemma~\ref{le: matrix berns}. 

Starting from (\ref{eq: redefine v vector}), we can write 
\begin{align}
&\big|\big|\wv_{j}\big|\big|_{2}  = \big|\big|\Wm_{j} \text{sign}\left(c_{j}\right) \hv_{j}\big|\big|_{2} \leq \big|\big|\Wm_{j}\big|\big|_{2}\leq \big|\big|\Wm\left(\rv\right)\big|\big|_{2}   \nonumber\\
& =\big|\big|\Delta \Tm \Lm\big|\big|_{2}  \leq  \frac{\hat{C}_{7}R}{L^2} =: q\left(\vv_{1}\right), \nonumber
\end{align}
where the second inequality is based on the fact that $\Wm_{j}\left(\rv\right)$ is a sub-matrix of $\Wm\left(\rv\right)$ whereas the last inequality is based on the fact that $||\Lm||_{2}\leq C$, given in \cite[Equation (105)]{suliman2018blind}, and Lemma~\ref{lemma: bound on delta T} which is given at the end of this appendix.

To obtain $\nu\left(\vv_{1}\right)$ we have
\begin{align}
&\nu\left(\vv_{1}\right) = \Bigg|\Bigg|\sum_{j=1}^{R}\mathbb{E}\left[\wv_{j}^{H}\wv_{j}\right]\Bigg|\Bigg|_{2}\hspace{-4pt} = \Bigg|\Bigg|\sum_{j=1}^{R} \mathbb{E}\left[\hv_{j}^{H}\Wm_{j}^{H}\Wm_{j}  \hv_{j}\right]\Bigg|\Bigg|_{2} \nonumber\\
&= \sum_{j=1}^{R} \text{Tr}\left(\mathbb{E}\left[\Wm_{j}^{H}\Wm_{j}\right]\mathbb{E}\left[\hv_{j}\hv_{j}^{H}\right]\right) = \frac{1}{K}\Bigg|\Bigg|\sum_{j=1}^{R} \mathbb{E}\left[\Wm_{j}^{H}\Wm_{j}\right]\Bigg|\Bigg|_{2}\nonumber \\
&= \frac{1}{K} \bigg|\bigg|\mathbb{E}\left[\text{Tr}\left(\left.\Wm\left(\rv\right)\right.^{H}\Wm\left(\rv\right)\right)\right]\bigg|\bigg|_{2}= \nonumber\\
&\frac{1}{K} \Big|\Big|\text{Tr}\left(\mathbb{E}\left[\Lm \Lm^{H}\Delta \Tm \Delta \Tm^{H}\right]\right)\Big|\Big|_{2} \leq 3R\Big|\Big|\mathbb{E}\left[\Lm \Lm^{H}\Delta \Tm \Delta \Tm^{H}\right]\Big|\Big|_{2} \nonumber\\
& \leq 3R C ||\mathbb{E}\left[\Delta\Tm \Delta\Tm^{H}\right]||_{2} \leq \frac{C_{7}^{*}R^{2}}{L^{4}} 
\end{align}
for a numerical constant $C_{7}^{*}$. Here, the fourth equality is based on Lemma~\ref{lemma: covarince of h}, given at the end of this appendix, while the first inequality is based on the fact that for $\Am \in \mathbb{C}^{m \times m}$, $|\text{Tr}\left(\Am \right)| \leq m ||\Am||_{2}$. Finally, the last inequality is based on Lemma~\ref{lemma: bound on delta T}. Given that $L \geq RK$, we can finally conclude that
\begin{equation}
\label{eq: final bound eta}
\nu\left(\vv_{1}\right) \leq \frac{C_{7}^{*}R}{L^{3}K}.
\end{equation}
We are now ready to apply Lemma~\ref{le: matrix berns} to obtain a probability bound on $||\vv_{1}\left(\rv\right)||_{2}$ as follow
\begin{align}
\text{Pr}\left[ ||\vv_{1}\left(\rv\right)||_{2} \geq t \right] &\leq \left(K+1\right) \exp\left(\frac{-t^{2}/2}{\frac{t}{3}\hat{C}_{7}\frac{R}{L^{2}}+C_{7}^{*}\frac{R}{L^{3}K}}\right) \nonumber\\
\label{eq: s vector bound 3}
&\leq \left(K+1\right) \exp\left(\frac{-t^{2}L^{3}K}{4 C_{7}^{*}R}\right) \leq \frac{\delta}{2}
\end{align}
Here, the second inequality in (\ref{eq: s vector bound 3}) is valid for $t \in \left[0, \frac{3 C^{*}_{7}}{\hat{C}_{7}}\frac{1}{L K}\right]$ while the last inequality holds for $t= 2 \sqrt{\frac{C_{7}^{*}R}{L^3 K} \log \left(\frac{2(K+1)}{\delta}\right)}$. Note that this choice of $t$ falls in the range $ \left[0, \frac{3 C^{*}_{7}}{\hat{C}_{7}}\frac{1}{L K}\right]$ for $L \geq  C_{7} RK \log\left(\frac{2(K+1)}{\delta}\right)$ where $C_{7} = \frac{4\hat{C}_{7}^{2}}{9 C_{7}^{*}}$.

To obtain a probability upper bound on the second term in (\ref{eq: proof 0 1 no 3}), we set $\vv_{2}\left(\rv\right)=\widebar{\Tm}^{H} \Lm \hv:= \widetilde{\Wm}\left(\rv\right) \hv$ with 
\begin{align}
\widetilde{\Wm}\left(\rv\right) =\begin{bmatrix} \widetilde{\Wm}_{1}, \dots, \widetilde{\Wm}_{R}\end{bmatrix}  \in \mathbb{C}^{K \times RK}. \nonumber
\end{align}
Similar to $\vv_{1}\left(\rv\right)$ we can write
\begin{equation}
\label{eq: redefine v2 vector}
\vv_{2}\left(\rv\right) =  \sum_{j=1}^{R} \widetilde{\Wm}_{j} \ \text{sign}\left(c_{j}\right) \hv_{j} =: \sum_{j=1}^{R} \tilde{\wv}_{j}.
\end{equation}
Given that $\vv_{2}\left(\rv\right)$ is a sum of independent zero-mean vectors, we can apply Lemma~\ref{le: matrix berns} to obtain an upper bound on $||\vv_{2}\left(\rv\right)||_{2}$.  Starting from (\ref{eq: redefine v2 vector}) we can write 
\begin{align}
&\left|\left|\tilde{\wv}_{j}\right|\right|_{2}  = \big|\big| \widetilde{\Wm}_{j} \text{sign}\left(c_{j}\right) \hv_{j}\big|\big|_{2} \leq \big|\big| \widetilde{\Wm}_{j}\big|\big|_{2}\leq \big|\big| \widetilde{\Wm}\left(\rv\right)\big|\big|_{2}  \nonumber\\
& =\big|\big| \widebar{\Tm}^{H} \Lm\big|\big|_{2}  \leq  \frac{\hat{C}_{8}R}{L^2K} =: q\left(\vv_{2}\right), \nonumber
\end{align}
where the last inequality is based on $||\Lm||_{2}\leq C$, given in \cite{suliman2018blind}, and Lemma~\ref{lemma: bound on Tbar} at the end of this appendix. For $\nu\left(\vv_{2}\right)$, we can show using the same steps that let to $\nu\left(\vv_{1}\right)$, along with Lemma~\ref{lemma: bound on Tbar}, that
\begin{equation}
\nu\left(\vv_{2}\right) = \frac{C_{8}^{*}R}{L^{3}K^{2}}. \nonumber
\end{equation}
Now we apply Matrix Bernstein inequality lemma to obtain
\begin{align}
\text{Pr}\left[ ||\vv_{2}||_{2} \geq t \right] &\leq \left(K+1\right) \exp\left(\frac{-t^{2}/2}{\frac{t}{3}\hat{C}_{8}\frac{R}{L^{2}K}+C_{8}^{*}\frac{R}{L^{3}K^{2}}}\right) \nonumber\\
\label{eq: s2 vector bound 3}
&\leq \left(K+1\right) \exp\left(\frac{-t^{2}L^{3}K^{2}}{4 C_{8}^{*}R}\right) \leq \frac{\delta}{2}
\end{align}
for $t \in \left[ 0, \frac{3C_{8}^{*}}{\hat{C}_{8}} \frac{1}{LK}\right]$. The last inequality in (\ref{eq: s2 vector bound 3}) holds true for $t = 2 \sqrt{\frac{C_{8}^{*}R}{L^{3}K^{2}} \log \left(\frac{2(K+1)}{\delta}\right)}$, which belongs to the range $\left[ 0, \frac{3C_{8}^{*}}{\hat{C}_{8}} \frac{1}{LK}\right]$ for $L \geq  C_{8} R  \log\left(\frac{2(K+1)}{\delta}\right)$ where $C_{8} =  \frac{4\hat{C}_{8}^{2}}{9C_{8}^{*}}.$

Based on (\ref{eq: s vector bound 3}) and (\ref{eq: s2 vector bound 3}) we can conclude
\begin{align}
&\text{Pr}\left[ ||\tilde{\fv}\left(\rv\right) ||_{2} \geq C_{9} \sqrt{\frac{R}{L^{3}K} \log \left(\frac{2(K+1)}{\delta}\right)}\right]  \leq  \nonumber\\
& \text{Pr}\left[ ||\vv_{1}||_{2} + ||\vv_{2}||_{2} \geq  C_{9} \sqrt{\frac{R}{L^{3}K} \log \left(\frac{2(K+1)}{\delta}\right)}\right] \leq \delta \nonumber
\end{align}
provided that $L \geq  C_{9} RK \log\left(\frac{2(K+1)}{\delta}\right)$. Substituting the above result in (\ref{eq: proof 0 1 no 2}) we obtain (\ref{eq: bound for int 0 to 1}).
\begin{lemma}
\label{lemma: bound on delta T}
There exists a constant $\bar{C}_{8}$ such that
\begin{align}
\label{eq: bound on delta T first}
||\Delta\Tm||_{2} \leq  \frac{\bar{C}_{8}R}{L^{2}} \\
\label{eq: bound on delta T second}
\left|\left|\mathbb{E}\left[\Delta\Tm \Delta\Tm^{H}\right]\right|\right|_{2} \leq  \frac{\bar{C}_{8}R}{L^{4}}.
\end{align}
\end{lemma}
The proof of Lemma~\ref{lemma: bound on delta T} is given in Appendix~\ref{app: lemma delta T bound}.
\begin{lemma}
\label{lemma: bound on Tbar}
There exists a constant $\tilde{C}_{8}$ such that
\begin{align}
\label{eq: bound on Tbar first}
||\widebar{\Tm}||_{2} \leq  \frac{\tilde{C}_{8}R}{L^{2}K} \\
\label{eq: bound on Tbar second}
\left|\left|\mathbb{E}\left[\widebar{\Tm} \ \widebar{\Tm}^{H}\right]\right|\right|_{2} \leq  \frac{\tilde{C}_{8}R}{L^{4}K}.
\end{align}
\end{lemma}
The proof of Lemma~\ref{lemma: bound on Tbar}, which is omitted here, follows the same steps as that of Lemma~\ref{lemma: bound on delta T}.   

\begin{lemma}\normalfont\cite[Lemma~21]{yang2016super}
\label{lemma: covarince of h}
Let $\hv_{j} \in \mathbb{C}^{K \times 1}$ have i.i.d. entries on the complex unit sphere. Then, $\mathbb{E}\left[\hv_{j}\hv_{j}^{H}\right] = \frac{1}{K} \Id_{\text{K}}.$  
\end{lemma}

\section{Proof of Lemma~\ref{lemma: bound for 0 1 term 2} }
\label{app: proof for bound 0 1 term 2}
Starting from the formulation of $\fv_{1}\left(\rv\right)$ in (\ref{eq: final dual for 1}), and upon following the same steps that let to (\ref{eq: proof 0 1 no 2}), we can show that
\begin{align}
\label{eq: proof 0 1 no term 2 1}
&\left| \ \iint \limits_{0}^{1} \nuv\left(\rv\right)^{H}\fv_{1}\left(\rv\right) d\rv \right| \leq \sup_{\rv \in [0,1]^{2}} ||\phiv\left(\rv\right)||_{2}\hspace{-2pt} \left(\iint \limits_{0}^{1} ||\tilde{\fv}_{1}\left(\rv\right) ||_{2} d\rv \hspace{-2pt}\right)
\end{align}
where $\tilde{\fv}_{1}\left(\rv\right) := \mathcal{X}^{*} \left( \mathcal{X} \mathcal{X}^{*}\right)^{-1}\left(\qv_{1}\right) \av\left(\rv\right)$. Now based on the conditions on $\fv_{1}\left(\rv\right)$ in Theorem~\ref{th: main dual cert 1} we formulate the following linear system 
\begin{equation}
\label{eq: matrix poly delta new}
\underbrace{{\begin{bmatrix} {\Em}_{(0,0)}^{(1,0)}   &\frac{1}{\kappa}{\Em}_{(1,0)}^{(1,0)} &\frac{1}{\kappa}{\Em}_{(0,1)}^{(1,0)}  \\ {\Em}_{(0,0)}^{(0,0)} & \frac{1}{\kappa}{\Em}_{(1,0)}^{(0,0)}   & \frac{1}{\kappa}{\Em}_{(0,1)}^{(0,0)}   \\ {\Em}_{(0,0)}^{(0,1)} &\frac{1}{\kappa}{\Em}_{(1,0)}^{(0,1)}   & \frac{1}{\kappa}{\Em}_{(0,1)}^{(0,1)} \end{bmatrix}}}_{{\Em}_{1}} \begin{bmatrix} \alphav^{1}\\ \kappa \betav^{1}\\ \kappa \gammav^{1}\end{bmatrix} = \begin{bmatrix}  \hv \\  \bm{0}_{RK } \\  \bm{0}_{RK} \end{bmatrix}
\end{equation} 
where ${\Em}^{(m',n')}_{(m,n)} \in \mathbb{C}^{RK \times RK}$ consists of $R \times R$ block matrices of size $K \times K$ with the matrix at the $\left(l,k\right)$ location given by $[{\Em}^{(m',n')}_{(m,n)}]_{(l,k)} := \widetilde{\Mm}^{(m',n')}_{(m,n)}\left(\rv_{l},\rv_{k}\right)$. It follows from (\ref{eq: matrix poly delta new}) that
\begin{equation}
\begin{bmatrix} \alphav^{1}\\ \kappa \betav^{1}\\ \kappa \gammav^{1}\end{bmatrix} = {\Em}_{1}^{-1} \begin{bmatrix}   \hv \\  \bm{0}_{RK } \\  \bm{0}_{RK} \end{bmatrix}= {\Lm}_{1} \hv, \nonumber 
\end{equation}
where ${\Lm}_{1}$ is $3RK \times RK$ matrix that contains the first $RK$ columns in ${\Em}_{1}^{-1}$. However, based on (\ref{eq: bound on alpha}), (\ref{eq: bound on beta}), and (\ref{eq: bound on gamma}) we have $\big|\big|[ \left.\alphav^{1}\right.^{H}, \kappa \left.\betav^{1}\right.^{H}, \kappa \left.\gammav^{1}\right.^{H}]^{H}\big|\big|_{2} \leq C\sqrt{R}/N$, for a numerical constant $C$. Given that $||\hv||_{2}= \sqrt{R},$ we can conclude that $||{\Lm}_{1}||_{2} \leq \hat{C}/N$.

Similar to (\ref{tilte f matrix-vector}) we can express $\tilde{\fv}_{1}\left(\rv\right)$ as
\begin{equation}
\tilde{\fv}_{1}\left(\rv\right) = \left.\Tm\left(\rv\right)\right.^{H} {\Lm}_{1} \hv, \nonumber
\end{equation}
where $\Tm\left(\rv\right)$ is given by (\ref{eq: matrix T}). Hence, we have
\begin{equation}
\iint \limits_{0}^{1} ||\tilde{\fv}_{1}\left(\rv\right) ||_{2} d\rv \leq \hspace{-3pt} \iint \limits_{0}^{1}\hspace{-3pt}\left(||\underbrace{\left.\Delta\Tm\right.^{H} {\Lm}_{1} \hv}_{\vv_{1}^{1}\left(\rv\right)} ||_{2} + ||\underbrace{\widebar{\Tm}^{H} {\Lm}_{1} \hv}_{\vv_{2}^{1}\left(\rv\right)}  ||_{2} \hspace{-3pt} \right) d\rv \nonumber
\end{equation}
By following the same steps in Appendix~\ref{app: proof for bound 0 1} we can show that
\begin{eqnarray}
\label{eq: proof 0 1 no term 2 values 1}
q\left(\vv_{1}^{1}\right) = \hat{C}_{10} \frac{R}{L^{2}N},  &&\nu\left(\vv_{1}^{1}\right) =C^{*}_{10} \frac{R}{ L^{3} N^{2} K} \\
\label{eq: proof 0 1 no term 2 values 2}
q\left(\vv_{2}^{1}\right) = \tilde{C}_{10}\frac{R}{L^{2}N K}, &&\nu\left(\vv_{2}^{1}\right) =  \bar{C}_{10} \frac{R}{L^{3}N^{2}K^{2}} 
\end{eqnarray}
Now, using Lemma~\ref{le: matrix berns} along with (\ref{eq: proof 0 1 no term 2 values 1}) and (\ref{eq: proof 0 1 no term 2 values 2}), and upon following the same steps in the proof of Lemma~\ref{lemma: bound for 0 1 term}, we can show that  $\iint \limits_{0}^{1} ||\tilde{\fv}_{1}\left(\rv\right) ||_{2} d\rv  \leq \bar{C}_{9} \sqrt{\frac{R}{L^{3} N^{2} K} \log \left( \frac{2\left(K+1\right)}{\delta}\right)}$ occurs with probability at least $(1-\delta)$ provided that  $L \geq \bar{C}_{9} R K \log \left(\frac{2(K+1)}{\delta}\right)$. The proof of Lemma~\ref{lemma: bound for 0 1 term 2} is concluded by substituting the above result in (\ref{eq: proof 0 1 no term 2 1}).

\section{Proof of Lemma~\ref{lemma: bound on delta T} }
\label{app: lemma delta T bound}
To prove (\ref{eq: bound on delta T first}) we set
\begin{equation}
\label{eq: new delta T defition}
\Delta \Tm^{H} = \left[\Delta\Tm_{1}, \Delta\Tm_{2}, \Delta\Tm_{3}\right]; \ \Delta\Tm_{i} \in \mathbb{C}^{K \times R K}
\end{equation}
where $\Delta\Tm_{1}$ contains the matrices $\widetilde{\Mm}_{(0,0)}\left(\rv,\rv_{j}\right)-\mathbb{E}[\widetilde{\Mm}_{(0,0)}\left(\rv,\rv_{j}\right)], j=1,\dots, R$, $\Delta\Tm_{2}$ have the matrices $\frac{1}{\kappa}\widetilde{\Mm}_{(1,0)}\left(\rv,\rv_{j}\right)-\frac{1}{\kappa}\mathbb{E}[\widetilde{\Mm}_{(1,0)}\left(\rv,\rv_{j}\right)], j=1,\dots, R$, and $\Delta\Tm_{3}$ contains the matrices $\frac{1}{\kappa}\widetilde{\Mm}_{(0,1)}\left(\rv,\rv_{j}\right)-\frac{1}{\kappa}\mathbb{E}[\widetilde{\Mm}_{(0,1)}\left(\rv,\rv_{j}\right)], j=1,\dots, R$. Based on (\ref{eq: M titlde expre}) and (\ref{eq: G assumption 3}) we can now write
\begin{align}
&\Big|\Big|\widetilde{\Mm}_{(0,0)}\left(\rv,\rv_{j}\right)-\mathbb{E}[\widetilde{\Mm}_{(0,0)}\left(\rv,\rv_{j}\right)]\Big|\Big|_{2} = \frac{1}{L^{2}P^{2}} \times \nonumber\\
&\Bigg|\Bigg| \sum_{p,l,l'=-N}^{N}  g_{p} e^{-i2\pi p \left(f-f_{j}\right)}\sum_{k'=-N}^{N} g_{k'} e^{ \frac{-i 2 \pi k'(p-l')}{L}}  e^{i2\pi k'\tau_{j}}  \times \nonumber\\
&  \sum_{k=-N}^{N} e^{ \frac{i 2 \pi k (p-l)}{L}} e^{-i2\pi k\tau} \left[ \frac{\dv_{l} \dv_{l'}^{H}}{\sum_{i=-N}^{N} ||\dv_{i}||_{2}^{2}} - \frac{1}{L K} \Id_{\text{K}}\mathbbm{1}_{\left(l,l'\right)}\right] \Bigg|\Bigg|_{2} \nonumber
\end{align}
where $\mathbbm{1}_{\left(l,l'\right)}$ is the indicator function, i.e., $\mathbbm{1}_{\left(l,l'\right)} = 1$ if $l=l'$. It is shown in \cite[Proof of Lemma~3]{heckel2016super} that
\begin{equation}
\frac{1}{P} \sum_{k'=-N}^{N} g_{k'} e^{ \frac{-i 2 \pi k'(p-l')}{L}}  e^{i2\pi k'\tau_{j}} \leq  C \min \left(1, \frac{1}{p^{4}}\right) \nonumber
\end{equation}
and that
\begin{equation}
C \sum_{p=-N}^{N} \min \left(1, \frac{1}{p^{4}}\right) \left|\frac{1}{P}\sum_{k=-N}^{N} e^{ \frac{i 2 \pi k (p-)l}{L}} e^{-i2\pi k\tau} \right|\leq C^{*} \nonumber
\end{equation}
Hence, we can obtain based on (\ref{eq: fejer coe}) and the above bound
\begin{align}
&\Big|\Big|\widetilde{\Mm}_{(0,0)}\left(\rv,\rv_{j}\right)-\mathbb{E}[\widetilde{\Mm}_{(0,0)}\left(\rv,\rv_{j}\right)]\Big|\Big|_{2} \leq \frac{C}{L^{2}} \times \nonumber\\
&\Bigg|\Bigg| \sum_{l,l'=-N}^{N}\left[ \frac{\dv_{l} \dv_{l'}^{H}}{\sum_{i=-N}^{N} ||\dv_{i}||_{2}^{2}} - \frac{1}{L K} \Id_{\text{K}}\mathbbm{1}_{\left(l,l'\right)}\right]   \Bigg|\Bigg|_{2}  \leq \frac{\bar{C}}{L^{2}} \nonumber
\end{align}
for some constant $\bar{C}$. As a result, we can conclude that 
\begin{equation}
\label{eq: bound of delta T 1}
||\Delta \Tm_{1}||_{2} \leq \frac{\hat{C}R}{L^{2}}
\end{equation} 
for some constant $\hat{C}$. Upon following the same steps, we can show that the same bound in (\ref{eq: bound of delta T 1}) exits for $||\Delta \Tm_{2}||_{2} $ and $||\Delta \Tm_{3}||_{2}$. This completes the proof of (\ref{eq: bound on delta T first}).

To prove (\ref{eq: bound on delta T second}), we start from the left-hand side of (\ref{eq: bound on delta T second}), and we write based on (\ref{eq: new delta T defition})
\begin{align}
&\left|\left|\mathbb{E}\left[\Delta\Tm \Delta\Tm^{H}\right]\right|\right|_{2} \leq  \left|\text{Tr}\left( \mathbb{E}\left[\Delta\Tm^{H}\Delta\Tm\right]\right)\right| \nonumber\\
&\leq  \sum_{i=1}^{3}\left|\text{Tr}\left( \mathbb{E}\left[\Delta\Tm_{i}\Delta\Tm_{i}^{H}\right]\right)\right|, \nonumber
\end{align}
where the first inequality is based on the fact that for a symmetric positive definite matrix $\Am$ we have $|\text{Tr}\left(\Am\right) | \geq ||\Am||_{2}.$

Given the definition of $\Delta\Tm_{1}$ we can write
\begin{align}
\label{eq: proof second 2}
&\left|\text{Tr}\left(\mathbb{E}\left[\Delta\Tm_{1}\Delta\Tm_{1}^{H}\right]\right)\right| \leq \nonumber\\
&\sum_{j=1}^{R}\left|\text{Tr}\left(\mathbb{E}\left[ \left(\widetilde{\Mm}_{(0,0)}\left(\rv,\rv_{j}\right)-\mathbb{E}[\widetilde{\Mm}_{(0,0)}\left(\rv,\rv_{j}\right)]\right) \times\right.\right. \right.\nonumber\\
&\left. \left.\left.\left(\widetilde{\Mm}_{(0,0)}\left(\rv,\rv_{j}\right)-\mathbb{E}[\widetilde{\Mm}_{(0,0)}\left(\rv,\rv_{j}\right)]\right)^{H}\right]\right)\right|.
\end{align}
By using (\ref{eq: M titlde expre}) we can write
\begin{align}
&\left|\text{Tr}\left(\mathbb{E}\left[ \left(\widetilde{\Mm}_{(0,0)}\left(\rv,\rv_{j}\right)-\mathbb{E}[\widetilde{\Mm}_{(0,0)}\left(\rv,\rv_{j}\right)]\right) \times\right.\right. \right.\nonumber\\
&\left. \left.\left.\left(\widetilde{\Mm}_{(0,0)}\left(\rv,\rv_{j}\right)-\mathbb{E}[\widetilde{\Mm}_{(0,0)}\left(\rv,\rv_{j}\right)]\right)^{H}\right]\right)\right| \leq  \frac{1}{(LP)^{4}} \times \nonumber\\
& \left|\text{Tr}\left(\mathbb{E}\left[ \sum_{p,l,l'=-N}^{N} \sum_{k',k=-N}^{N} g_{k'} e^{-i 2 \pi  \frac{k'(p-l')}{L}}  e^{i2\pi k'\tau_{j}} e^{ i 2 \pi k \left(\frac{p-l}{L}-\tau\right)} \times \right.\right. \right.\nonumber\\
&\left. \left.\left. g_{p} e^{-i2\pi p \left(f-f_{j}\right)} \left[ \frac{\dv_{l} \dv_{l'}^{H}}{\sum_{i=-N}^{N} ||\dv_{i}||_{2}^{2}} - \frac{1}{L K} \Id_{\text{K}}\mathbbm{1}_{\left(l,l'\right)}\right] \times \right.\right. \right.\nonumber\\
&\left. \left.\left.\sum_{p_{1},l_{1},l'_{1}=-N}^{N} \sum_{k'_{1},k_{1}=-N}^{N} g_{k'_{1}} e^{ -i 2 \pi\frac{ k'_{1}(p_{1}-l'_{1})}{L}}  e^{i2\pi k'_{1}\tau_{j}} e^{ i 2 \pi k_{1} \left(\frac{p_{1}-l_{1}}{L}-\tau\right)} \times \right.\right. \right.\nonumber\\
&\left. \left.\left. g_{p_{1}} e^{-i2\pi p_{1} \left(f-f_{j}\right)} \left[ \frac{\dv_{l_{1}} \dv_{l'_{1}}^{H}}{\sum_{i=-N}^{N} ||\dv_{i}||_{2}^{2}} - \frac{1}{L K} \Id_{\text{K}}\mathbbm{1}_{\left(l_{1},l'_{1}\right)}\right]^{H}\right]\right)\right|  \nonumber\\
& \leq \frac{C}{L^{4}} \left|\text{Tr}\left(\sum_{l,l',l_{1},l'_{1}=-N}^{N} \mathbb{E} \left[\frac{\dv_{l} \dv_{l'}^{H}  \dv_{l'_{1}} \dv_{l_{1}}^{H}}{\left(\sum_{i=-N}^{N} ||\dv_{i}||_{2}^{2}\right)^{2}}\right.\right. \right.\nonumber\\
&\left. \left. \left.- \frac{ \dv_{l} \dv_{l'}^{H} \mathbbm{1}_{\left(l_{1},l'_{1}\right)}}{L K\sum_{i=-N}^{N} ||\dv_{i}||_{2}^{2}} - \frac{\dv_{l'_{1}} \dv_{l_{1}}^{H} \mathbbm{1}_{\left(l,l'\right)}}{L K\sum_{i=-N}^{N} ||\dv_{i}||_{2}^{2}} \right] \right)+\frac{1}{LK}\right| \leq \frac{\bar{C}}{L^{4}} \nonumber
\end{align}
Here, the last inequality is based on the fact that the first term inside the absolute value can be shown to be upper bounded by $(LK+L+K)/LK$ while the second and the third terms can be shown to be bounded by $C/K$. Thus, we can conclude based on (\ref{eq: proof second 2}) that
\begin{equation}
\label{eq: proof second 4}
\left|\text{Tr}\left( \mathbb{E}\left[\Delta\Tm_{1}\Delta\Tm_{1}^{H}\right]\right)\right| \leq \frac{C^{*}R}{L^{4}}.
\end{equation}
Upon following the same steps with $\left|\text{Tr}\left( \mathbb{E}\left[\Delta\Tm_{2}\Delta\Tm_{2}^{H}\right]\right)\right|$ and $\left|\text{Tr}\left( \mathbb{E}\left[\Delta\Tm_{3}\Delta\Tm_{3}^{H}\right]\right)\right|$, we can show that the bound in (\ref{eq: proof second 4}) holds for these terms. This completes the proof of (\ref{eq: bound on delta T second}).

\section{Proof of Proposition~\ref{pro: upper bound on the diff} }
\label{app: proof of U hat atomic}
To prove Proposition~\ref{pro: upper bound on the diff}, we need the below proposition.
\begin{proposition}
\label{pro: bound of U hat atomic norm}
Based on (\ref{eq: system in vector}) and (\ref{eq: function g}) we have
\begin{equation}
\label{eq: bound of U hat atomic norm equ}
||\hat{\gv}\left(\rv\right)||_{\text{TV}} \leq  ||\gv\left(\rv\right)||_{\text{TV}} + \lambda^{-1} \iint \limits_{0}^{1} ||\nuv\left(\rv\right)||_{2} d\rv.
\end{equation} 
\end{proposition}
\begin{proof} (Proposition~\ref{pro: bound of U hat atomic norm})
First note that
\begin{align}
&\frac{1}{2}\big|\big|\yv- \mathcal{X}(\widehat{\Um})\big|\big|_{2}^{2} +\mu \big|\big|\widehat{\Um}\big|\big|_{\mathcal{A}}\leq \frac{1}{2}\big|\big|\yv- \mathcal{X}(\Um_{\text{o}})\big|\big|_{2}^{2} +\mu \big|\big|\Um_{\text{o}}\big|\big|_{\mathcal{A}} \nonumber
\end{align}
Given the fact that $\yv =\mathcal{X}(\Um_{\text{o}})+\omegav$ we can write
\begin{align}
&\mu \big|\big|\widehat{\Um}\big|\big|_{\mathcal{A}} \leq \frac{1}{2}\left[||\omegav||_{2}^{2} -\big|\big|\omegav- \mathcal{X}(\widehat{\Um}-\Um_{\text{o}})\big|\big|_{2}^{2}   \right] + \mu \big|\big|\Um_{\text{o}}\big|\big|_{\mathcal{A}}\nonumber\\
& \leq  \left|\left\langle\omegav, \mathcal{X}(\widehat{\Um}-\Um_{\text{o}}) \right\rangle\right| + \mu \big|\big|\Um_{\text{o}}\big|\big|_{\mathcal{A}} = \left|\left\langle\omegav,\ev\right\rangle\right|+\mu \big|\big|\Um_{\text{o}}\big|\big|_{\mathcal{A}} \nonumber 
\end{align}
By using (\ref{eq: function g}), we can conclude that $\big|\big|\Um_{\text{o}}\big|\big|_{\mathcal{A}} = ||\gv\left(\rv\right)||_{\text{TV}}$. Combining this result with the above inequality leads to
\begin{equation}
\label{eq: app u hat 2}
\mu ||\hat{\gv}\left(\rv\right)||_{\text{TV}} \leq \mu ||\gv\left(\rv\right)||_{\text{TV}} + \left|\left\langle\omegav,\ev\right\rangle\right|.
\end{equation} 
Starting from the second term in (\ref{eq: app u hat 2}) we can write
\begin{align}
\label{eq: app u hat 3}
\left|\left\langle\omegav,\ev\right\rangle\right| &=  \left|\iint \limits_{0}^{1} \left\langle \nuv\left(\rv\right), \mathcal{X}^{*}\left(\omegav\right)\av\left(\rv\right)^{H} \right\rangle d\rv \right|  \nonumber\\
&\leq  ||\mathcal{X}^{*}\left(\omegav\right)||_{\mathcal{A}}^{*} \iint \limits_{0}^{1} ||\nuv\left(\rv\right)||_{2} d\rv, 
\end{align}

where the inequality is based on H\"{o}lder's inequality and the definition of the dual atomic norm. By applying Lemma~\ref{lemma: sub phi bound} in (\ref{eq: app u hat 3}), then substituting the result in (\ref{eq: app u hat 2}), we obtain (\ref{eq: bound of U hat atomic norm equ}). This completes the proof of Proposition~\ref{pro: bound of U hat atomic norm}.
\end{proof}
Starting from the second term in (\ref{eq: bound of U hat atomic norm equ}) we can write
\begin{equation}
\label{equ: bound of U hat 1 terms}
\iint \limits_{0}^{1} ||\nuv\left(\rv\right)||_{2} d\rv  \leq \iint \limits_{\Omega_{\text{far}}} ||\nuv\left(\rv\right)||_{2} d\rv+\iint \limits_{\Omega_{\text{close}}} ||\nuv\left(\rv\right)||_{2} d\rv.
\end{equation}
Based on (\ref{eq: error bound final}) and Lemmas~\ref{lemma: bound on phi} to \ref{lemma: bound for 0 1 term 3} we can conclude that
\begin{align}
\label{equ: bound of U hat 2 terms}
&\iint \limits_{\Omega_{\text{close}}} ||\nuv\left(\rv\right)||_{2} d\rv \leq \bar{C}_{11} \mu \left( T_{3}  +   \iint \limits_{\Omega_{\text{far}}} ||\nuv\left(\rv\right)||_{2}d\rv \right.\nonumber\\
& \left. + \sup_{\rv \in [0,1]^{2}} ||\phiv\left(\rv\right)||_{2}\left( \sqrt{\frac{R}{L^{3} K} \log \left( \frac{2\left(K+1\right)}{\delta}\right)}\right)\right)
\end{align}
On the other hand, note that 
\begin{align}
\label{eq: g and g hat differe}
&||\hat{\gv}\left(\rv\right)||_{\text{TV}} = ||\gv\left(\rv\right)-\nuv||_{\text{TV}} = ||\gv\left(\rv\right)-\mathcal{P}_{\mathcal{R}}\left(\nuv\right)||_{\text{TV}}+\nonumber\\
&||\mathcal{P}_{\mathcal{R}^{c}}\left(\nuv\right)||_{\text{TV}} \geq ||\gv\left(\rv\right)||_{\text{TV}}-||\mathcal{P}_{\mathcal{R}}\left(\nuv\right)||_{\text{TV}}+||\mathcal{P}_{\mathcal{R}^{c}}\left(\nuv\right)||_{\text{TV}}
\end{align}
Finally, based on Proposition~\ref{pro: bound of U hat atomic norm}, equations (\ref{equ: bound of U hat 1 terms}), (\ref{equ: bound of U hat 2 terms}), and (\ref{eq: g and g hat differe}) we obtain (\ref{equ: final p pc bound second one}). This completes the proof of Proposition~\ref{pro: upper bound on the diff}.

\bibliographystyle{IEEEbib}
\bibliography{Accepted_Ext_64}

\end{document}